\documentclass[aps,prx,twocolumn,footinbib,
showpacs,showkeys,longbibliography]{revtex4-2}

\usepackage{graphicx}
\usepackage{indentfirst}
\usepackage{physics}
\usepackage{braket}
\usepackage{float}
\usepackage{amsmath}
\usepackage{amssymb}
\usepackage{verbatim}
\usepackage{wasysym}
\usepackage{epstopdf}
\usepackage{CJK}
\usepackage{esint}
\usepackage{color}
\usepackage{xcolor}
\usepackage{epsfig}
\usepackage{subfigure}
\usepackage{amsfonts}
\usepackage{footmisc}
\usepackage{scrextend}
\usepackage{multirow}
\usepackage{mathtools}

\usepackage{xr-hyper}
\usepackage[hyperfootnotes=false]{hyperref}

\usepackage[english]{babel}
\usepackage{url}
\usepackage{bm}
\definecolor{darkblue}{rgb}{0,0,0.5}
\hypersetup{
colorlinks=true,
linkcolor=black,
filecolor=black,
citecolor=darkblue,
urlcolor=darkblue,
}

\newcommand{\defeq}{\vcentcolon=}

\DeclareMathOperator*{\argmin}{arg\,min}

\newcommand\mc[1]{\mathcal{#1}}

\newcommand\bs[1]{\boldsymbol{#1}}
\newcommand\bff[1]{\text{\textbf{#1}}}

\newtheorem{theorem}{Theorem}
\newtheorem{lemma}{Lemma}
\newtheorem{corollary}{Corollary}

\newtheorem{defin}{Definition}
\newenvironment{proof}[1][Proof]{\noindent\textbf{#1.} }{\ \rule{0.5em}{0.5em}}

\begin{document}

\title{End-to-End Capacities of Hybrid Quantum Networks}
\author{Cillian Harney}
\author{Alasdair I.~Fletcher}
\author{Stefano Pirandola}
\affiliation{Department of Computer Science, University of York, York YO10 5GH, United Kingdom}
\begin{abstract}
Future quantum networks will be hybrid structures, constructed from complex architectures of quantum repeaters interconnected by quantum channels that describe a variety of physical domains; predominantly optical-fiber and free-space links. 
In this hybrid setting, the interplay between the channel quality within network sub-structures must be carefully considered, and is pivotal for ensuring high-rate end-to-end quantum communication. 
In this work, we combine recent advances in the theory of point-to-point free-space channel capacities and end-to-end quantum network capacities in order to develop critical tools for the study of hybrid, free-space quantum networks. We present a general formalism for studying the capacities of arbitrary, hybrid quantum networks, before specifying to the regime of atmospheric and space-based quantum channels. 
We then introduce a class of modular quantum network architectures which offer a realistic and readily analysable framework for hybrid quantum networks. By considering a physically motivated, highly connected modular structure we are able to idealize network performance and derive channel conditions for which optimal performance is guaranteed. This allows us to reveal vital properties for which distance-independent rates are achieved, so that the end-to-end capacity has no dependence on the physical separation between users. Our analytical method elucidates key infrastructure demands for a future satellite-based global quantum internet, and for hybrid wired/wireless metropolitan quantum networks. 
\end{abstract}
\maketitle

\section{Introduction}

The current internet is a vast classical network, designed to facilitate global communication and distributed information processing \cite{SlepianNets,CoverThomas, TanenbaumNets, GamalNets}. The inherent robustness of classical information allows for hybrid, flexible network architectures which operate in both optical-fiber and free-space, befitting an environment and mode of application. A future quantum internet will aim to play an analogous role for quantum information technologies \cite{KimbleQI,UniteQInt,RazaviQNet,AdvCrypt}, but the inherent fragility of quantum information makes achieving high-rates over long distances much more challenging. In pursuit of this goal, theoretical and experimental progress in the study of hybrid quantum networks is fundamental and necessary. 

The ultimate limits of fiber-networks are well understood. By means of the Pirandola-Laurenza-Ottaviani-Banchi (PLOB) bound, it is known that the capacity of a fiber-link decays exponentially with respect to the link-length with a precise law \cite{PirPatron09,PLOB}. The PLOB bound has been used to understand the end-to-end network capacities of fiber-based quantum architectures \cite{End2End}, to assess the limits of realistic, random network structures \cite{QuntaoRandQNets, ZhangQInt} and idealized, highly-connected, analytical architectures \cite{OPGQN}. 
These investigations have provided essential insight and motivation for the construction of high performance quantum networks, elucidating key physical properties and network characteriztics. Chiefly, to develop a high-performance, fiber-based quantum network one must carefully consider not only connectivity, but nodal density and maximum link length; leading to high resource demands for large-scale designs.

However, quantum networks will not be limited to just optical-fiber but will collaborate with free-space methods of communication. On the ground, the flexibility of free-space links are obviously more suitable for mobile quantum devices and short-range connections. Meanwhile, the ability to establish ground-to-satellite and intersatellite free-space connections offers remarkable short-cuts for global quantum communications \cite{LiaoSQC,Micius2017,RenSQC,Micius2020,VillarSQC,Micius2021,AdvSQC}. Such connections bypass many decibels of loss that would be otherwise experienced on the ground, and utilize the dynamic nature of satellites to achieve high-rates over global distances. 

Determining the ultimate limits of free-space quantum channels is difficult, requiring tools from quantum information theory \cite{Mike_Ike, GaussRev, SerafiniCV}, optics \cite{GoodmanOptics,SveltoLasers,BohrenScatter} and turbulence theory \cite{TatarskiiTurb, MajumdarTurb, KaushalTurb,AndrewsTurb}. Recent advancements have placed tight upper-bounds on the quantum capacities of point-to-point free-space channels, using a modified PLOB bound that accounts for atmospheric fading processes \cite{FS, SQC}. With these results in hand, we have the ingredients to go beyond the point-to-point scenario and quantitatively study the ultimate limits of communications in free-space quantum networks. 

In this work, we combine results from Refs.~\cite{End2End} and \cite{FS,SQC} in order to place bounds on the end-to-end capacities of generally hybrid quantum networks. In particular, we put forward a formalism for studying the capacities of quantum networks whose channels are described by free-space, fiber, or any medium that can be generalized as a fading channel. This treatment is then specified to fading processes that are experienced by optical transmissions through the atmosphere, or in space.

Furthermore, we introduce a framework to investigate the ultimate limits of hybrid, modular quantum networks. We focus on a modular network design which consists of disjoint sub-networks (or communities) connected to a large-scale backbone network used to mediate intercommunity quantum communication. For the first time, this provides the tools to investigate highly relevant quantum network models, such as: Globally distant fiber sub-networks connected to a satellite backbone network, offering insight into the resource requirements of a satellite-based quantum internet; and wireless, free-space sub-networks on the ground interconnected via a fiber backbone, presenting a useful model for studying hybrid metropolitan networks. 

Extending the techniques of Ref.~\cite{OPGQN} we employ ideally connected structures within different parts of the modular network. In doing so, we are able to derive simple, yet powerful analytical constraints which promise distance-independent rates for modular quantum architectures, and thus optimal performance. These results provide valuable insight into the ultimate limits of hybrid networks,
can help to motivate future quantum network designs and provide a valuable platform upon which to further develop realistic free-space quantum networks.

\subsection{Paper Structure}
This paper is structured as follows: In Section \ref{sec:Prelims} we provide a comprehensive review of quantum networks, point-to-point capacities of general fading channels and the end-to-end network capacities of quantum networks generally composed of fading channels. We then specify this review to optical free-space quantum communications, summarizing recent progress in the determination of ultimate limits for a number of key settings. In Section \ref{sec:Mod} we formalize a network architecture for the study of hybrid, modular quantum networks. We further specify an idealized network architecture which allows us to establish properties that guarantee optimal end-to-end performance and distance-independence. Finally, Section \ref{sec:HybridFS} applies the machinery from the previous sections to investigate the optimal performance of hybrid quantum architectures. In particular, we establish network constraints for communication between remote fiber-based sub-networks connected to a satellite backbone, and for ground-based free-space sub-networks connected to a fiber backbone. Concluding remarks and future investigative paths are then discussed.

\section{Preliminaries \label{sec:Prelims}}

\subsection{Quantum Networks}

An arbitrary quantum network can be described as a finite, undirected graph $\mc{N} = (P,E)$ where $P$ is the set of all network nodes (points/vertices) on the graph, and $E$ is the set of all edges within the network. Each node $\bs{x} \in P$ represents a local register of quantum systems which can be exchanged with connected neighbors. Meanwhile each edge in the network is denoted by the unordered pair ${(\bs{x},\bs{y})\in E}$, and is used to represent a quantum channel $\mc{E}_{\bs{xy}}$ through which users can exchange quantum systems. Two nodes $\bs{x}$ and $\bs{y}$ are connected if the edge $(\bs{x}, \bs{y})$ exists within $E$. \par

It is important to note that the physical orientation of each channel $\mc{E}_{\bs{xy}}$ in the network can be a forwards or backwards and need not be specified. Under the assistance of two-way classical communications (CCs), the optimal transmission of quantum information is connected with optimal entanglement distribution. It does not depend on the physical direction of system exchange but the Local Operations (LOs) that are applied at each point, and thus the direction of teleportation. This refers to the \textit{logical direction} of quantum communication. If an undirected edge represents a physically asymmetric channel (i.e.~$\mc{E}_{\bs{xy}} \neq \mc{E}_{\bs{yx}}$) then the users can always enable a teleportation protocol that uses the most efficient physically directed channel. Hence, the logical flow of quantum information can always be made independent of the physical flow of quantum systems.

Consider two end-users Alice and Bob, who reside at remote nodes within the quantum network, $\bs{\alpha}$ and $\bs{\beta}$ respectively. An end-to-end communication protocol between Alice and Bob can be most broadly captured by a general adaptive protocol. Alice and Bob propagate the exchange of quantum systems between nodes throughout the network in accordance with an overarching protocol in order to establish some global, target quantum state. This may involve point-to-point, or point-to-multi-point exchanges dependent on the nature of the protocol. Quantum interactions between any two nodes is alternated with adaptive network-wide LOCCs, allowing for consistent optimization of the network protocol. Communication is complete when the target state is eventually established between the end-users after a number of uses of the network. \par

The optimal performance over any such network protocol is captured via the generic two-way assisted network capacity $\mc{C}(\mc{N})$ which describes the ultimate rate with which a desired target state can be established. If the target is a private state, this refers to the network secret-key capacity $K(\mc{N})$ measured in secret-bits per network use. If it is a maximally entangled state then this becomes the network entanglement distribution capacity $E(\mc{N})$ measured in entanglement bits (ebits) per network use. 

Note that we consider a general information-theoretic definition of a quantum repeater as a middle third-party helping the quantum communication between a sender and a receiver (therefore not connected by a direct link). In practice, there are many possible physical realizations, e.g., see Refs.~\cite{Borregaard_1WRep,Childress_FTQComms,Simon_Repeaters} among others.

\subsubsection{Network Routing}

Thanks to the interconnectivity of quantum networks, there is no unique path that network interactions must follow in order to establish end-to-end quantum communication. However, there exist two fundamental classes of routing strategy under which all protocols can be described: \textit{Single-path} or \textit{multi-path} routing.\par

Single-path routing is the simplest network communication method which utilizes point-to-point communications in a sequential manner. Quantum systems are exchanged from node-to-node followed by LOCC operations after each transmission until eventually communication has been established between the end-users. One may define a single-path network capacity ${\mc{C}}^s(\mc{N})$ which describes the optimal performance obtained via a sequential end-to-end strategy. A repeater chain forms a particular instance of quantum network under single-path routing where there only exists one unique end-to-end route.\par

A more powerful strategy is multi-path routing, which properly exploits the multitude of possible end-to-end routes available in a quantum network. A user may exchange an initially multi-partite quantum state with a number of neighboring receiver nodes, who may each then perform their own point-to-multi-point exchanges along their unused edges. Again, every exchange of quantum systems can be interleaved with adaptive network LOCCs, and this process continues until multi-point interactions are carried out with the end-users. A multi-path routing strategy in which all channels in the network are used precisely once per end-to-end transmission is known as a \textit{flooding protocol}. This is achieved via non-overlapping point-to-multi-point transmissions at each network node, such that receiving nodes only choose to transmit along unused edges for subsequent connections. Therefore, there exists a multi-path network capacity (or flooding capacity) ${\mc{C}}^m(\mc{N})$ which describes the optimal network performance out of all possible strategies (single or multi-path).

\subsection{Quantum Networking over Fading Channels}

The effect of \textit{fading} refers to the temporal variation of transmissivity along a bosonic lossy channel. The transmissivity along a fading channel is not fixed, but instead follows a probability distribution described by the dynamics of the environment. For example, the propagation of bosonic modes through low-altitude free-space instigates a fading channel thanks to chaotic processes in the atmosphere. 
The impact of fading on a communications channel is described via its \textit{speed}, i.e.~the ability for a receiver to resolve the dynamics of the transmissivity fluctuations. Slow-fading implies that the users can resolve the fading dynamics and accurately perform channel estimation because either the fading process is weak or the users possess sufficiently fast detectors. On the other hand, fast-fading refers to the situation where the users cannot reconcile the dynamics of the channel and can only estimate the statistical distribution of the channel transmissivity \cite{UsenkoFading,PanosFading}. It is clear that fast-fading poses a more formidable task for communicators. \par

More precisely, a bosonic lossy fading channel is defined as an ensemble of lossy channels in accordance with some probability density function $F(\tau)$ which describes the instantaneous transmissivity along the channel. We denote a lossy fading channel as the ensemble
\begin{equation}
\mc{E}_{F}(\eta) \defeq \{ F(\tau) ; \mc{E}_{\tau}\}
\end{equation}
where $\mc{E}_{\tau}$ is a lossy channel with fixed, instantaneous transmissivity $\tau \in [0,\eta]$ and $\eta$ is the maximum transmissivity that is attainable along the channel.

\subsubsection{Capacities of Fading Channels\label{sec:FadeCaps}}

The absolute maximum rate that two parties can transmit qubits, establish secret-keys, or distribute entanglement
over bosonic lossy channels is known exactly via the PLOB bound \cite{PLOB}.
This states that generic two-way assisted capacity of a quantum channel is precisely
\begin{equation}
\mc{C}(\mc{E}_{\eta}) =\mc{B}(\eta) \defeq -\log_2(1-\eta) \label{eq:PLOBb}
\end{equation}
measured in bits per channel use, and where we introduce $\mc{B}(\eta)$ as the capacity function for stable lossy channels. While this assumes a fixed transmissivity $\eta$, 
the PLOB bound can be readily employed to study fading channels \cite{PLOB, CondChanSim}. Thanks to convexity properties of the relative entropy of entanglement (REE) over ensembles of channels, the capacity of a lossy fading channel can be bounded according to,
\begin{equation}
\mc{C}[\mc{E}_F(\eta)] \leq \mc{B}_{F}(\eta) \defeq  \int_{0}^\eta d\tau~F(\tau)\> \mc{B}(\tau). \label{eq:P2Pfad}
\end{equation}
where we have defined $\mc{B}_{F}(\eta)$ as the capacity function for lossy fading channels.
This can be interpreted as a generalization of Eq.~(\ref{eq:PLOBb}), modified to include potential fading processes. Indeed, it is simple to retrieve the standard bound $-\log_2(1-{\eta})$ for fixed lossy channels by considering a trivial probability distribution where only one transmissivity value is possible, $\eta$. Hence, this format is conveniently general and allows one to describe any lossy bosonic channel (with or without fading). \par

More generally, lossy channels will also be exposed to thermal-noise, resulting in a thermal-loss channel $\mc{E}_{\tau,\bar{n}}$. This channel equates to mixing an input mode with a thermal mode of mean photon number $\bar{n}_{\text{e}} = \bar{n}/(1-\tau)$ on a beam splitter of transmissivity $\tau$, effectively adding $\bar{n}$ photons to the signal mode. The capacity of thermal-loss channels is not known exactly, but upper-bounds have been derived through the techniques developed for the PLOB bound. For a fixed thermal-loss channel $\mc{E}_{\tau,\bar{n}}$ the capacity can be upper-bounded via \cite{PLOB}
\begin{align}
\mc{C}(\mc{E}_{\tau,\bar{n}})\leq \mc{L}({\tau,\bar{n}}) \defeq -\log_2\left(\tau^{\bar{n}_{\text{e}}} (1-\tau)\right) - h(\bar{n}_{\text{e}}), \label{eq:ThB}
\end{align}
where $h(x) \defeq (x+1)\log_2(x+1) - x\log_2(x)$. Otherwise $\mc{L}({\tau,\bar{n}}) = 0$ when $\tau < \bar{n}$, meaning that there exists a minimum transmissivity at which communication can be reliably secured. 

Analogous to the pure-loss setting, a thermal-lossy fading channel can be described by the ensemble 
\begin{equation}
\mc{E}_F(\eta,\bar{n}) \defeq \{ F(\tau, \bar{n}) ; \mc{E}_{\tau,\bar{n}}\},
\end{equation}
 where it is possible that both transmissivity and thermal noise are probabilistic and described within a probability density function $F(\tau, \bar{n})$. Typically, thermal noise can always be considered constant by either assuming stable operational conditions, or by minimising (maximising) its potential value for best-case (worst-case) rates. This allows us to consider the simpler ensemble ${\mc{E}_F(\eta,\bar{n})= \{ F(\tau) ; \mc{E}_{\tau,\bar{n}}\}}$ on which we place the following upper-bound of its capacity \cite{FS,SQC},
\begin{equation}
\mc{C}[\mc{E}_F(\eta,\bar{n})] \leq {\mc{L}}_F(\eta,\bar{n}) \defeq  \int_{\bar{n}}^\eta d\tau~F(\tau)\> \mc{L}({\tau,\bar{n}}).\label{eq:P2PTLfad}
\end{equation}
Here we have defined ${\mc{L}}_F(\eta,\bar{n})$ as a tight capacity bounding-function for thermal-lossy fading channels \footnote{The word \textit{tight} in this context refers to how close the upper-bound is from its best known lower-bound. Indeed, there exists a lower-bound on the capacity of a point-to-point thermal-loss channel based on its reverse coherent information (RCI) \cite{RCInfo}. Hence, throughout our work we implicitly refer to tight upper-bounds on thermal-loss channel capacities (and subsequently, network capacities) as those which in conjunction with the RCI can tightly sandwich the exact capacity.}. Intuitively, one can never outperform the pure-loss PLOB bound in the presence of thermal-noise, hence we can always write 
\begin{equation}
\mc{C}[\mc{E}_F(\eta,\bar{n})]  \leq {\mc{L}}_F(\eta,\bar{n}) \leq {\mc{B}}_F(\eta).
\end{equation}

\subsubsection{Capacities of Fading Networks \label{sec:GenFadNet}}

We can combine the theory from these previous sections in order to provide a general model for quantum networks with fading channels. Indeed, we may construct a quantum network $\mc{N} = (P,E)$ such that all edges $(\bs{x},\bs{y})\in E$ are generally associated with a unique thermal-lossy fading channel, 
\begin{equation}
\mc{E}_{\bs{xy}} =  \mc{E}_{F_{\bs{xy}}}({\eta_{\bs{xy}}}, \bar{n}_{\bs{xy}}),~\forall (\bs{x},\bs{y}) \in E.
\end{equation} 
In this way, each network edge not only possesses a unique maximum transmissivity $\eta_{\bs{xy}}$ and thermal-noise properties $\bar{n}_{\bs{xy}}$, but also a unique instantaneous transmissivity probability density function $F_{\bs{xy}}$ through which each edge can adopt its own fading dynamics (or lack thereof). This allows for a description of network channels within different environmental media such as fiber channels, ground-based free-space channels, or free-space channels beyond the atmosphere. Furthermore, we can retrieve pure-loss fading channels via $\bar{n}_{\bs{xy}} = 0$.

It has been shown that the capacities of quantum networks can be derived through a combination of quantum information theoretic tools and ideas from classical network theory \cite{End2End}. By transforming the notion of classical network cuts into entanglement cuts of a quantum network, one can determine compact, analytical expressions for the network capacities of arbitrary architectures. Consider a pair of end-users within the fading network $\bs{\alpha},\bs{\beta} \in P$. 
We define an entanglement cut $C$ as a means of disconnecting/partitioning the network into two disjoint super-users $\bff{A}$ and $\bff{B}$ such that $\bs{\alpha} \in \bff{A}$ and $\bs{\beta} \in \bff{B}$ and $P = \bff{A}\cup\bff{B}$. A network cut $C$ generates an associated cut-set,
\begin{equation}
\tilde{C} = \{ (\bs{x},\bs{y})\in E ~|~\bs{x}\in \bff{A}, \bs{y} \in \bff{B}\},
\end{equation}
defining a collection of network edges which when removed successfully partitions the network.

As discussed in previous sections, single-path routing can be thought of as a generalization of repeater-chains, where end-to-end communication is established via the sequential exchange of quantum systems along a designated path. 
The single-path network capacity ${\mc{C}}^{s}({\mc{N}})$ is bounded by determining the network cut $C_{\min}$ that generates the smallest, maximum single-edge capacity in the cut set. For lossy and thermal-lossy fading networks, we may define the single-path capacity quantities \cite{End2End},
\begin{align}
&\mc{B}^s(\mc{N}) \defeq \min_C \max_{(\bs{x},\bs{y})\in \tilde{C}}  \mc{B}_{F_{\bs{xy}}}(\eta_{\bs{xy}}) \label{eq:SPR_loss},\\
&{\mc{L}}^s(\mc{N}) \defeq \min_C \max_{(\bs{x},\bs{y})\in \tilde{C}}  {\mc{L}}_{F_{\bs{xy}}}(\eta_{\bs{xy}},\bar{n}_{\bs{xy}}) \label{eq:SPR_th},
\end{align}
which upper-bound the single-path network capacity,
\begin{equation}
{\mc{C}}^{s}({\mc{N}}) \leq {\mc{L}}^s(\mc{N})  \leq {\mc{B}}^s(\mc{N}).  \label{eq:SPR}
\end{equation}
In the absence of thermal noise, the single path capacity ${\mc{B}}^s(\mc{N})$ is achievable and equates to performing sequential communication along the optimal route $\omega^*$ in the network. For an arbitrary network, finding the optimal route is equivalent to solving the well known widest-path problem and can be solved efficiently \cite{Dijkstra}. It is unknown whether the thermal upper-bound in Eq.~(\ref{eq:P2PTLfad}) is achievable, hence in the presence of thermal noise ${\mc{L}}^s(\mc{N})$ presents a tight upper-bound. \par

More powerful network protocols employ multi-path routing \cite{Alas_Flooding}. The multi-path network capacity is associated with an optimal flooding protocol, and is found by locating the entanglement cut $C_{\min}$ which minimizes the multi-edge capacity over all cut-sets. For lossy and thermal-lossy fading networks we can compute the multi-path quantities \cite{End2End},
\begin{align}
&\mc{B}^m(\mc{N}) \defeq \min_C \sum_{(\bs{x},\bs{y})\in \tilde{C}} \mc{B}_{F_{\bs{xy}}}(\eta_{\bs{xy}}) \label{eq:MPR_loss},\\
&{\mc{L}}^m(\mc{N}) \defeq \min_C \sum_{(\bs{x},\bs{y})\in \tilde{C}} {\mc{L}}_{F_{\bs{xy}}}(\eta_{\bs{xy}},\bar{n}_{\bs{xy}}), \label{eq:MPR_th}
\end{align}
which upper-bound the multi-path network capacity,
\begin{equation}
{\mc{C}}^{m}(\mc{N}) \leq {\mc{L}}^m(\mc{N})  \leq {\mc{B}}^m(\mc{N})  \label{eq:MPR}.
\end{equation}
Once more, for pure-loss based networks the flooding capacity $\mc{B}^m(\mc{N})$ is achievable, and equates to solving the classical maximum-flow minimum-cut problem according to a network of capacity achieving links. For general quantum networks with an arbitrary architectures, this problem requires a numerical treatment and can be solved efficiently \cite{FordFlow,KarpFlow,OrlinFlow}. Once more, since the thermal-loss upper bound is not guaranteed to be achievable ${\mc{L}}^m(\mc{N})$ is instead a tight upper-bound.

\subsection{Free-Space Quantum Communication}

Consider two remote parties Alice and Bob who are separated by a distance $z$, and employ quantum communications based upon a quasi-monochromatic optical mode ($\Delta \lambda$-nm large and $\Delta t$-sec long). This may be characterized by a Gaussian beam with wavelength $\lambda$, initial field spot-size $w_0$ and curvature $R_0$. Communication consists of transmitting a directed beam towards a receiver with circular aperture of radius $a_R$. Here we assume that the initial spot-size $w_0$ is sufficiently small with respect to the transmitter aperture of radius $a_T$ so that there is no relevant diffraction caused by the transmitter.\par

The atmospheric effects which characterize free-space channels are variable with respect to altitude, due to changes in atmospheric density. Therefore specifying the trajectory of a Gaussian beam through free-space is pivotal in capturing channel quality. To this end, for any point-to-point communications task we may assume a general beam trajectory $L$ and introduce the following altitude/propagation functions respectively, $h_{L}(z)$ and $z_{L}(h)$. Using these functions we can retain a geometry independent framework for our study until we wish to specify to a particular setting.\par

\subsubsection{Free-Space Transmissivity}

Free-space diffraction is a universal contributor to loss. As a beam propagates in free-space its waist will widen as a function of the distance it travels, 
\begin{equation}
w_{\text{d}}^2(z) = w_0^2 \left[ \left(1-\left(\frac{z}{R_0}\right)^2\right) + \left(\frac{z}{z_R}\right)^2 \right]
\end{equation}
where $z_R \defeq {\pi w_0^2}/{\lambda}$ is the Rayleigh range. A target receiver will then only detect a portion of the spread beam since its aperture is finite in size, inducing a diffraction-limited transmissivity,
\begin{equation}
\eta_{\text{d}}(z) = 1 - \exp\left[-\frac{2 a_R^2}{w_{\text{d}}^2}\right] .
\end{equation}
It is also useful to define a diffraction induced transmissivity in the far-field regime, $z \gg z_R$, making the approximation $\eta_{\text{d}} \approx \eta_{\text{d}}^{\text{far}} ~{\defeq} ~{2 a_R^2}/{w_{\text{d}}^2}$. This loss quantity exists regardless of the specific environmental setting considered, from ground-based links to intersatellite connections.  \par

Propagation through the atmosphere incurs further loss due to aerosol absorption and Rayleigh/Mie scattering; an effect known as atmospheric extinction. At a fixed altitude $h$, this loss can be accurately described via the Beer-Lambert equation \cite{BohrenScatter}. Since beam trajectories may be variable in altitude, we can generally define the extinction-induced transmissivity as
\begin{equation}
\eta_{\text{atm}}(z) = \exp \left[ - \int_0^z dz \> \alpha[h_L(z)]  \right],
\end{equation}
where $\alpha(h) = \alpha_0 e^{-{h}/{\tilde{h}}}$ is the extinction factor, $\tilde{h} = 6600\ \text{m}$, and ${\alpha_0}$ is the extinction factor at sea-level. For $\lambda = 800\ \text{nm}$ it follows that ${\alpha_0 \approx 5\times10^{-6}\ \text{m}^{-1}}$ \par

Finally, there exist inevitable internal losses associated with the detector setup, due to imperfect fiber-couplings, sub-optimal quantum detector efficiency, and more. This inefficiency-induced transmissivity can be as low as $\eta_{\text{eff}} \approx 0.4$ and must be considered to capture realistic performance. All of these effects can be used to describe a fixed, maximum transmissivity of a free-space connection, 
\begin{equation}
\eta(z) \defeq \eta_{\text{eff}}\>\eta_{\text{atm}}(z)\>\eta_{\text{d}}(z). \label{eq:etatot}
\end{equation}
Importantly, $\eta$ can be readily modified to consider variable altitude beam trajectories and written as a function of a chosen spatial geometry to account for different extinction properties throughout the atmosphere.

\subsubsection{Atmospheric Fading \label{sec:Turb}}
It is remarkably optimistic to assume that a free-space transmission deterministically undergoes a pure-loss channel characterized by Eq.~(\ref{eq:etatot}) only. The chaotic behavior of air-flow, temperature and pressure throughout the atmosphere invites further complications for free-space transmissions, causing inaccuracies in the point-to-point trajectory known as \textit{beam wandering}. As a result, we must incorporate fading for a more accurate characterization.\par

Turbulence is used to describe how a free-space propagating beam is perturbed by fluctuations in the atmospheric refractive index, caused by spatial variations in pressure and temperature. Propagating beams interact with small turbulent air-flows on a fast time-scale, too fast for communicators to monitor or resolve. This causes the beam waist to broaden and forces us to define a \textit{short-term spot-size} $w_{\text{st}}$ which is larger than the diffraction-induced spot size, $w_{\text{d}} < w_{\text{st}}$. On a slower time-scale, the beam will undergo deflections by significantly larger eddies in the atmosphere. This slower time-scale may be reconcilable by the communicators, and manifests as a wandering of the beam centroid. This wandering can be described by a Gaussian random walk of the centroid with variance $\sigma_{\text{t}}^2$ which is a functional of the beam trajectory, operational setup, conditions, and more.\par

Wandering is not exclusively caused by turbulence, and one must also consider pointing errors caused by jitter and imperfect targeting. These effects also occur on a reasonably slow time-scale of order $0.1-1\>\text{s}$, and may be resolved by the receiver. This introduces an additional wandering variance $\sigma_{\text{p}}^2$, e.g.~a $1\>\mu\text{rad}$ pointing error at the transmitter causes a variance ${\sigma_{\text{p}}^2 \approx (10^{-6} z)^2}$ (where $z$ is in meters). Overall, these effects combine to induce Gaussian centroid wandering with variance 
$
\sigma^2 = \sigma_{\text{t}}^2 + \sigma_{\text{p}}^2.
$

The ability for communicators to resolve these wandering dynamics is dependent on their time-scale. The behavior of turbulence is variable, with regimes ranging from weak to strong turbulence. Increasing turbulent strength can be modeled as an increasingly faster fading process, such that a receiver loses the ability to reconcile the wandering dynamics. For stronger levels of turbulence, it is possible to define a \textit{long-term spot-size} $w_{\text{lt}}$ which averages over the wandering caused by both small turbulent eddies and larger eddy deflections, ${w_{\text{d}} < w_{\text{st}} < w_{\text{lt}}}$. Indeed, the turbulence-induced variance is defined with respect to the long-term and short-term quantities $\sigma_{\text{t}}^2 = w_{\text{lt}}^2 - w_{\text{st}}^2$. However, rigorous studies of strong turbulence will require further considerations, for which work is currently underway \cite{Masoud_StrongTurb}.

In this work, we focus on the regime of weak turbulence and the concept of short-term beam spot sizes. These can be used to provide precise descriptions of free-space quantum channels on the ground at short-range, and for ground-to-satellite communication along trajectories with small zenith angles \cite{FS,FanteFS1,FanteFS2}.

\subsubsection{Weak Turbulence}

For communications undergoing weak turbulence, the beam wandering acts on a time scale of $10-100\ \text{ms}$ and can be fully resolved with a sufficiently fast detector. In this case, analytical expressions can be found for the short-term spot size $w_{\text{st}}$ and the centroid wandering variance $\sigma^2$. Consider a beam with wavenumber $k={2\pi}/{\lambda}$ following a free-space trajectory $L$ (and its associated altitude function $h_L(z)$). Then the spherical-wave coherence length is given by,
\begin{equation}
\rho_{0}(L) = \left[ 1.46k^2 \int_{0}^z d\zeta \left( 1-\frac{\zeta}{z}\right)^{\frac{5}{3}} C_n^2 \left[ h_L(\zeta) \right] \right]^{-\frac{3}{5}},
\end{equation}
where $C_n^2$ denotes the refractive index structure constant, used to measure the strength of fluctuations in the atmospheric refractive index. This quantity has an explicit dependence on the beam's trajectory, since this may be variable in altitude, and is typically described via the Hufnagel-Valley model (See Appendix C of Ref.~\cite{FS}). Provided that Yura's condition is satisfied $\phi \defeq 0.33(\rho_0/w_0)^{\frac{1}{3}} \ll 1$ \cite{YuraFS} then we can write \cite{FS},
\begin{align}
w_{\text{st}}^2 &\approx w_{\text{d}}^2 + 2\left( \frac{\lambda z}{\pi \rho_{0}} \right)^2 (1-\phi)^2, \label{eq:rho0} \\
\sigma_{\text{t}}^2 &\approx 2\left( \frac{\lambda z}{\pi \rho_{0}} \right)^2 \left[1-(1-\phi)^2\right]. \label{eq:sigmaTB}
\end{align}
The short-term spot size can be used to update the diffraction induced transmissivity to account for fast beam interaction with small turbulent eddies in the atmosphere. That is, 
\begin{equation}
\eta_{\text{st}} \defeq 1 - \exp\left[-\frac{2 a_R^2}{w_{\text{st}}^2} \right] ~\underset{z \gg z_R}{\approx}~\eta_{\text{st}}^{\text{far}}\defeq \frac{2 a_R^2}{w_{\text{st}}^2}. \label{eq:etast}
\end{equation}
where we have simultaneously introduced a far-field approximation, $ \eta_{\text{st}}^{\text{far}}$ when the propagation distance is very large $z \gg z_R$.
\par
Updating the diffraction-induced transmissivity in Eq.~(\ref{eq:etatot}), we may write a new maximum transmissivity incorporating weakly turbulent effects, 
$
\eta = \eta_{\text{eff}}\>\eta_{\text{atm}}\>\eta_{\text{st}}.
$
This represents the optimal transmissivity parameter that can be achieved when the beam centroid $\vec{x}_C$ is perfectly aligned with the receiver centroid $\vec{x}_R$, i.e.~the centroid deflection is $r \defeq \| \vec{x}_C - \vec{x}_R \| = 0$. 
However, due to turbulence and pointing errors, the beam centroid now undergoes a Gaussian random walk with variance $\sigma^2$, invoking a fading channel.
We can then connect the non-zero centroid deflection $r \geq 0$ to an instantaneous transmissivity $\tau(r)$ to precisely capture the fading process. Gaussian wandering induces a Weibull distribution for the centroid deflection, which results in an instantaneous transmissivity probability density function $F_{\sigma}[\tau(r)]$ \cite{FS}. Defining the functions,
\begin{align}
f_0(x) &\defeq [1 - \exp(-2x) I_0 (2x)]^{-1},\\
f_1(x) &\defeq \exp(-2x) I_1 (2x),
\end{align}
where $I_n$ is the modified Bessel function of the first kind for $n=0,1$, we can introduce the following shape and scale parameters, 
\begin{align}
\gamma &= \frac{4 \eta_{\text{st}}^{\text{far}} f_0 ( \eta_{\text{st}}^{\text{far}} ) f_1 ( \eta_{\text{st}}^{\text{far}} )}{ \ln\left[2 \eta_{\text{st}} f_0(\eta_{\text{st}}^{\text{far}})\right]},~r_0 = \frac{a_R}{\ln\left [2 \eta_{\text{st}} f_0(\eta_{\text{st}}^{\text{far}}) \right]^{\frac{1}{\gamma}}}.
\end{align}
With these, we can now write the instantaneous transmissivity probability density function,
\begin{equation}
F_{\sigma}(\tau) = \frac{r_0^2}{\gamma \sigma^2 \tau} \ln(\frac{\eta}{\tau})^{\frac{2}{\gamma}-1} \exp\left[ -\frac{r_0^2}{2\sigma^2} \ln(\frac{\eta}{\tau})^{\frac{2}{\gamma}}\right].
\end{equation}
Consequently, we are left with a free-space, lossy fading channel $\mc{E}_{F_{\sigma}}(\eta) = \{ F_{\sigma}(\tau) ;\mc{E}_{\tau} \}$. Using the tools from Section \ref{sec:Prelims}, we can study the capacities of free-space connections.

Hence, the capacities for free-space quantum communications (entanglement distribution or secret-key distribution) are upper bounded according to \cite{FS}
\begin{align}
\mc{C} \leq \mc{B}_{F_{\sigma}}(\eta) = -\Delta(\eta,\sigma) \log\left( 1 - \eta\right), \label{eq:WeakCap}
\end{align}
where $\Delta$ represents a correction factor to the PLOB bound due to imperfect alignment,
\begin{equation}
\Delta(\eta,\sigma) = 1 + \frac{\eta}{\text{ln}(1-\eta)} \int_{0}^{\infty} dx \frac{\exp\left[\frac{-r_0^2}{2\sigma^2}x^{\frac{2}{\gamma}}\right]}{e^x - \eta}.
\end{equation}
Through specification to a free-space trajectory, one can easily determine geometry dependent expressions for this ultimate limit. 
Importantly, for channels which are accurately described as ensembles of pure-loss channels (thermal noise is negligible), then Eq.~(\ref{eq:WeakCap}) is in fact an achievable and optimal rate, $\mc{C} = \mc{B}_{F_{\sigma}}(\eta)$. For all other scenarios where thermal noise is non-negligible, it remains an effective upper-bound.

\subsubsection{Thermal Noise}

As discussed previously, pure-loss based bounds remain ultimate bounds in the presence of thermal noise. Yet, it is still possible to construct tighter performance bounds by considering fading channels which are ensembles of thermal-loss channels. Let $\bar{n}_T$ be the mean number of input photons transmitted towards a receiver via a single free-space mode. For an instantaneous transmissivity $\tau$ the mean photon number collected at the receiver will be $\bar{n}_R = \tau \bar{n}_T + \bar{n}$, where $\bar{n}$ describes the total environmental thermal noise added to the signal. It is useful to define contributions to this environmental noise via \begin{equation}
\bar{n} \defeq \eta_{\text{eff}}\>\bar{n}_B + \bar{n}_{\text{ex}}
\end{equation}
where the receiver collected $\bar{n}_B$ mean background photons with detector efficiency $\eta_{\text{eff}}$, and $\bar{n}_{\text{ex}}$ accounts for excess setup noise. In the study of ultimate limits, $\bar{n}_{\text{ex}}$ can be considered to be approximately zero, or can be attributed to trusted noise.\par

For free-space links, the primary source of thermal noise is attributed to natural brightness within the field of view of the transmission, i.e.~the sky, Sun, Moon, etc. Using a receiver of aperture $a_R$, angular field of view $\Omega_{\text{fov}}$, a detector with time window $\Delta t$ and frequency filter $\Delta \lambda$ around $\lambda$, then the number of background thermal photons per mode is 
\begin{equation}
\bar{n}_B = H_{\lambda} \Gamma_R,~\text{where}~\Gamma_R \defeq \Delta t \Delta \lambda \Omega_{\text{fov}} a_R^2. \label{eq:ThermalB}
\end{equation}
Here, $H_{\lambda}$ describes the spectral irradiance of the environment in units of photons $\text{m}^{-2}~\text{s}^{-1}~\text{nm}^{-1}~\text{sr}^{-1}$, and is unique to the operational setting and trajectory. Using the general bound from Eq.~(\ref{eq:P2PTLfad}) and specifying to free-space beam wandering dynamics with variance $\sigma^2$, we can write the free-space thermal upper-bound,
\begin{equation}
\mc{C} \leq {\mc{L}}_{F_{\sigma}}(\eta,\bar{n}) = \mc{B}_{F_{\sigma}}(\eta) - \mc{T}_{F_{\sigma}}(\eta, \bar{n}), \label{eq:ThermCap}
\end{equation}
where the thermal correction is given explicitly by,
\begin{align}
\begin{aligned}
\mc{T}_{F_{\sigma}}(\eta, \bar{n}) \defeq &\left[ 1 - \exp\left(\frac{-r_0^2}{2\sigma^2}{\text{ln}\left[\frac{\eta}{\bar{n}}\right]^{\frac{2}{\gamma}}}\right) \right] \\
&\times \left[ \frac{\bar{n}\log_2(\bar{n})}{1-\bar{n}} + h(\bar{n}) \right] - \mc{B}_{F_{\sigma}}(\bar{n}). 
\end{aligned}
\end{align}
This result applies to settings of weak and intermediate turbulence, such that one can substitute the appropriate reconcilable wandering variance and maximum transmissivity into this result. 

\subsubsection{Noise Suppression and Frequency Filters}

As seen in Eq.~(\ref{eq:ThermalB}), the number of background thermal photons per mode has a strong dependence on the frequency filter, $\Delta \lambda$. The frequency filter assists in blocking out noise, and thus the use of ultra-narrow filters is highly desirable. In discrete-variable quantum communications, physical frequency filters are typically limited to around $\Delta \lambda = 1 \text{ nm}$. However, using CV quantum systems and appropriate interferometric measurements it is possible to achieve much narrower effective filters. 

Many CV-based protocols rely on the use of a local-oscillator (or phase reference) in order to perform homodyne or heterodyne measurements at the output. This phase reference may be co-propagated with signal pulses, or  alternatively reconstructed at the receiver. This reconstruction method involves interleaving the signal pulses with strong reference pulses that carry information about the local-oscillator \cite{AdvCrypt}. Since the output of a homodyne measurements is proportional to the mean photon number in the local-oscillator modes, the ability to utilize bright references pulses over free-space channels introduces an \textit{effective homodyne filter}. Thermal noise mode-matching with the local-oscillator and the signal will be detected, but all other noise will be filtered out. This allows for the implementation of ultra-narrow effective filters on the order of $\Delta \lambda = 0.1 \text{ pm}$ with practical CV protocols, and can dramatically reduce the magnitude of the thermal background noise  (see Ref.~\cite{FS} for more details).

\subsection{Important Free-Space Channels}

\subsubsection{Ground-Based Channels}

Wireless classical communication networks are ubiquitous and fundamental to everyday modern life. Thus the desire for a free-space quantum analogue is obvious, enabling access to future wireless quantum technologies. Nonetheless, it is intuitive that such communication will be limited to short-range thanks to prominent decoherence obtained at ground-level. At a fixed altitude, beam trajectories are horizontal paths with the simple altitude/propagation functions $h_L(z) = h,~z_L(h) = z$. The absence of a variable altitude in the beam path simplifies a number of key quantities such as the extinction-induced transmissivity 
\begin{equation}
\eta_{\text{atm}}(z) = \exp\left[ -\alpha(h) z\right], \label{eq:ExtGr}
\end{equation}
and the spherical-wave coherence length
\begin{equation}
\rho_0 = \left[0.548 k^2 C_n^2(h)z\right]^{-\frac{3}{5}}, \label{eq:SphGr}
\end{equation}
which can be used to accurately describe decoherence and fading dynamics on the ground. Here, turbulence is a major factor and must be stringently considered. A useful parameter for assessing the validity of turbulent regimes on the ground is the Rytov variance,
\begin{equation}
\sigma_{\text{Ry}}^2 = 1.23 \> k^{\frac{7}{6}} z^{\frac{11}{6}} C_n^2(h).
\end{equation}
In particular, weak turbulence requires that $\sigma_{\text{Ry}}^2 \lesssim 1$. Using a Gaussian beam with $\lambda = 800\text{ nm}$ and altitudes close to sea-level during typical day-time conditions, weak turbulence is only guaranteed for distances of $z \lesssim 1\text{ km}$. Beyond this, as in the intermediate ($\sigma_{\text{Ry}}^2 \gtrsim 1$) and strong ($\sigma_{\text{Ry}}^2 \gg 1$) turbulent regimes, the long-term spot size must be adopted, leading to poorer ultimate channel capacities \cite{Masoud_StrongTurb}.  \par

Fig.~\ref{fig:LossPlots}(a) illustrates the behavior of transmissivity in ground-based free-space channels with respect to propagation length. Within the weak-turbulence regime the loss properties of free-space channels limited to $\sim 4$ dB for communications over $1$ km, encouraging the utility of short-range, optical free-space quantum communications.

For the assessment of thermal bounds, the primary source of thermal-noise at ground-level is attributed to the brightness of the sky. This provides a spectral irradiance ranging from 
\begin{equation}
H_{\lambda}^{\text{sky}} \approx
\begin{cases}
1.9\times 10^{13}, & \text{ full-Moon, clear night},\\
1.9\times 10^{18}, & \text{ cloudy day time}.
\end{cases} \label{eq:GrSky}
\end{equation}
in units of ${\text{photons } \text{m}^{-2}\>\text{s}^{-1}\>\text{nm}^{-1}\>\text{sr}^{-1}}$. Using this information, the expressions in Eqs.~(\ref{eq:ExtGr}) and (\ref{eq:SphGr}), and the general capacity bounds developed in the previous sections, we can accurately assess the ultimate limits of free-space quantum communications on the ground (see Ref.~\cite{FS} for further details and derivations).

\subsubsection{Ground-Satellite Channels}

For communication between ground/satellite stations, there are two unique configurations that must be considered:~Transmissions directed from the ground towards a satellite (uplink) or from a satellite towards the ground (downlink). The quantum channel descriptions of these configurations are very different.

Consider a Gaussian beam propagated in uplink. The beam immediately undergoes turbulence upon generation at low altitude, and thus has a large decohering impact which must be carefully considered. However, pointing errors are less critical $\sigma_{\text{p}}^2 \ll \sigma_{\text{t}}^2$ thanks to the availability of adaptive optics to optimize the beam trajectory from the ground station. Therefore we must model uplink as a fading channel predominantly due to turbulent effects.
Meanwhile, a Gaussian beam in downlink experiences the opposite; the beam does not undergo serious levels of turbulence until it reaches lower altitudes. But by this point, its spot-size has already been spread by diffraction, hence turbulence does not present a serious factor and $\sigma_{\text{t}}^2 \approx 0$. Yet, in this setting pointing errors become much more relevant due to the lack of onboard access and optimization ability. Hence, atmospheric decoherence associated with uplink and downlink is physically asymmetric, invoking two unique fading channels.\par

We specify the trajectory of ground-satellite communication according to a target satellite altitude $h$ and zenith angle $\theta$, which describes the angle formed between the zenith point at the ground station and the direction of observation towards the satellite. The zenith angle takes values  $\theta \in [-\frac{\pi}{2},\frac{\pi}{2}]$, such that when $\theta = 0$ the satellite is at the zenith. The distance that the beam physically travels from its point of generation $z$ (known as its slant distance) can then be expressed with respect to this geometry. Defining the functions,
\begin{align}
\begin{aligned}
h_{\theta}(z) &=  \sqrt{R_E^2 +z^2 + 2zR_E \cos \theta} - R_E,\\
z_{\theta}(h) &=  \sqrt{h^2 + 2hR_E + R_E^2 \cos^2 \theta} - R_E\cos\theta.
\end{aligned}
\end{align}
we may then introduce the altitude/propagation functions with respect to uplink and downlink communications \cite{SQC},
\begin{align}
&z_{\theta}^{\text{up}}(h)= z_{\theta}(h),~~ z_{\theta}^{\text{down}}(h) =  z_{\theta}(h_{\max}) - z_{\theta}(h),\\
&h_{\theta}^{\text{up}}(z)= h_{\theta}(z),~~ h_{\theta}^{\text{down}}(z) =  h_{\theta}\big[  z_{\theta}(h_{\max}) - z\big].
\end{align}\par

Fig.~\ref{fig:LossPlots}(b) illustrates the behavior of transmissivity in ground-satellite channels with respect to uplink, downlink and satellite altitude. Here we plot both the  the expected transmissivity when averaged over the respective fading processes and maximum transmissivity (a best-case loss in the absence of fading). Crucially, it can be shown that for beam trajectories with relatively small zenith angles $\theta \leq 1 \text{ radiant}$, we can assume the regime of weak turbulence for the ground-satellite fading channel (see Appendix C of \cite{FS}). 
Within this angular window we can accurately resolve the fading dynamics, and by inserting the beam trajectory expressions into the machinery of Section \ref{sec:Turb}, it is possible to derive loss-based ultimate limits for both uplink and downlink quantum communications using the Eq.~(\ref{eq:WeakCap}), and thermal-loss-based limits using Eq.~(\ref{eq:ThermCap}).

The sources of environmental thermal-noise are also unique to both uplink and downlink configurations, and  operational settings such as the time of day and weather.  In uplink during the day, the primary source of thermal-noise is sunlight being reflected from the Earth to the satellite detector. Meanwhile, at night, this noise is diminished but there still exists sunlight being reflected from the Moon to the Earth, and back towards the satellite. For uplink, we may write
\begin{equation}
\bar{n}_B^{\text{up}} = \kappa H_{\lambda}^{\text{sun}} \Gamma_R.
\end{equation}
Here $\kappa$ is a parameter that accounts for the Earth/Moon albedos and ranges from $\kappa_{\text{night}} = 7.36\times 10^{-7}$ for a clear night with a full Moon, to $\kappa_{\text{day}} = 0.3$ during clear day-time. 
Meanwhile, for the optical wavelength $\lambda = 800 \text{ nm}$, we can approximate that in uplink the solar spectral irradiance is $H_{\lambda}^{\text{sun}} = 4.61 \times 10^{18} {\text{ photons } \text{m}^{-2}\>\text{s}^{-1}\>\text{nm}^{-1}\>\text{sr}^{-1}}$. \par

For downlink, the receiver is now a detector on the ground and the main source of noise is more simply attributed to the sky (as it was in the ground-based scenario). In this setting, and for $\lambda = 800 \text{ nm}$, the spectral irradiance of the sky follows Eq.~(\ref{eq:GrSky}).
For a much more detailed analysis, see Appendix D, Ref.~\cite{SQC}.

\begin{figure}
\includegraphics[width=\linewidth]{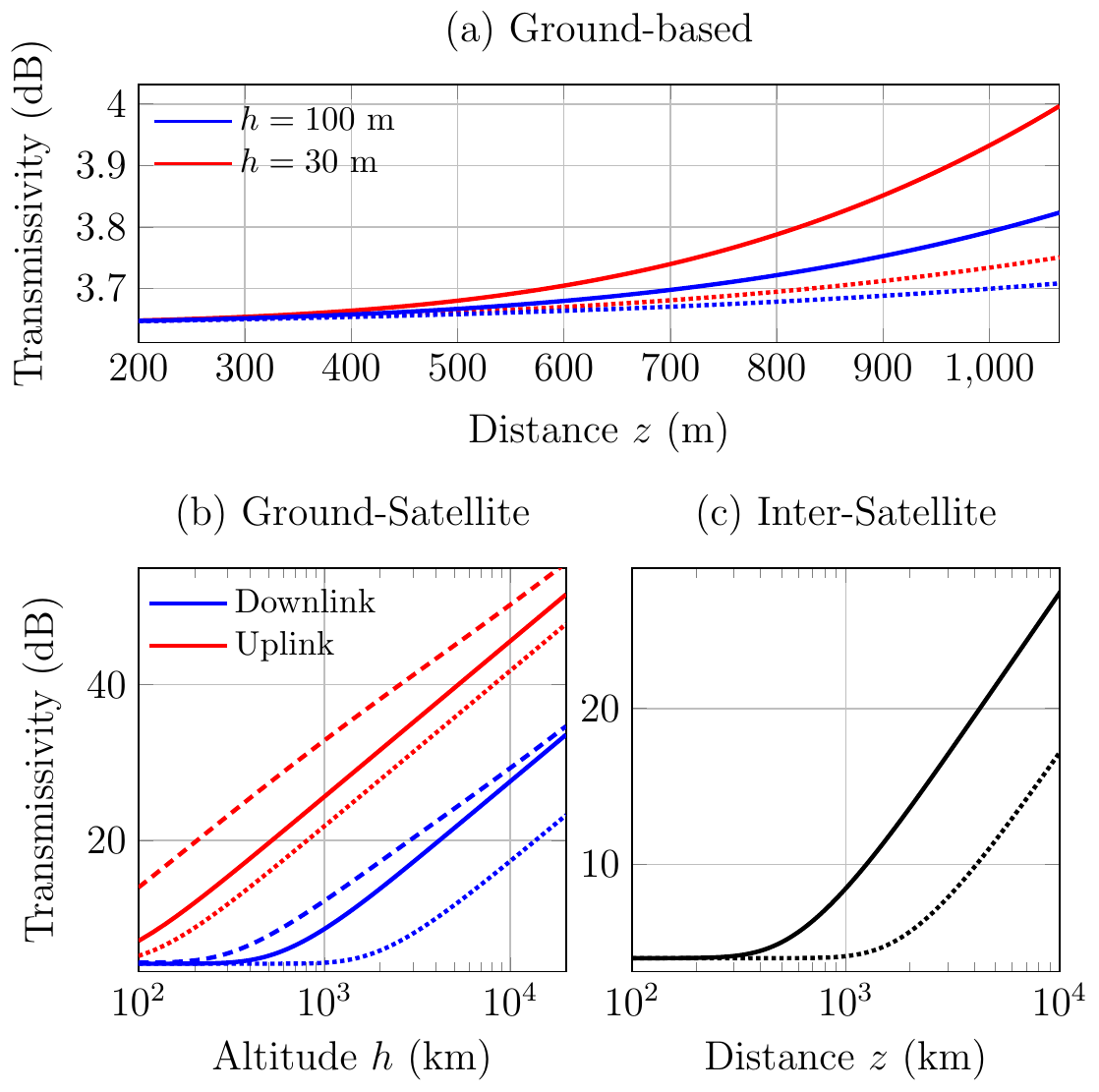}
\caption{Free-space transmissivity associated with (a) ground-based, (b) ground-satellite and (c) intersatellite communication links. In each plot, solid lines depict the average transmissivity (attenuation averaged over fading dynamics), while dotted lines describe the best-case transmissivity (absence of fading). The dashed lines in Panel (b) describe a ground-satellite free-space link with zenith angle $\theta = 1$ radiant, while the others consider $\theta = 0$. The operational setup in (a) is consistent with the parameters in Table~\ref{table:FSF_Setups} while (b) and (c) are consistent with Setup (\#1) in Table~\ref{table:Setups}. }
\label{fig:LossPlots}
\end{figure}

\subsubsection{interSatellite Channels}
Finally, we can consider free-space quantum communication between satellites in orbit. This represents a high quality free-space quantum channel which is free from atmospheric decoherence, and thus will not experience losses due to extinction nor undergo turbulence. 
Indeed, these intersatellite link losses are characterized by free-space diffraction only. Assuming negligible pointing errors, then the intersatellite channel is simply a lossy channel with transmissivity given by $\eta_{\text{d}}(z)$ as a function of the propagation distance between satellites, $z$. This lets us write an ideal upper-bound on the intersatellite channel capacity, 
\begin{equation}
\mc{C} \leq \mc{B}(\eta_{\text{d}}) = \frac{2a_R^2}{w_{\text{d}}^2(z) \ln 2}. \label{eq:InterSatCap}
\end{equation}

Due to the lack of onboard access and adaptive optics, it is possible that pointing errors become important and must be considered. If pointing errors are non-negligible, $\sigma_{\text{p}}^2 > 0$, then we instead must consider a lossy fading channel $\mc{E}_{F_{\sigma_{\text{p}}}} = \{ F_{\sigma_{\text{p}}} ; \mc{E}_{\tau} \}$ with maximum transmissivity $\eta_{\text{d}}(z)$. As discussed in earlier sections, pointing errors occur on a sufficiently slow time-scale such that they are reconcilable by the receiver. Hence, the capacity for this channel can be accessed via Eq.~(\ref{eq:WeakCap}), such that
\begin{equation}
\mc{C} \leq \mc{B}_{F_{\sigma_{\text{p}}}} (\eta_{\text{d}}) = \frac{2a_R^2}{w_{\text{d}}^2(z) \ln 2} \Delta(\eta_{\text{d}}(z), \sigma_{\text{p}}), \label{eq:InterSatCap_Fad}
\end{equation}
where $\Delta$ acts as a correction factor to the PLOB bound. It is clear that when $\sigma_{\text{p}}^2 = 0$ we retrieve Eq.~(\ref{eq:InterSatCap}). In Fig.~\ref{fig:LossPlots}(c) the loss properties of an optical intersatellite channel are illustrated with respect to distance between communicating satellites. This depicts similar transmissivity behavior to ground-satellite downlink channels with zenith angle $\theta=0$ without the additional degradation associated with atmospheric interactions. 

We have some important considerations to note. First of all, intersatellite channels can only be formed between satellites that fall within each other's line-of-sight. This naturally implies a limit to the maximum distance over which an intersatellite channel can be physically established. For any two satellites in circular orbits, at some point the Earth blocks the free-space between them, prohibiting transmittance. If two satellites orbit at altitudes $h_1$ and $h_2$, then we find this limit to be
\begin{align}
z_{\text{sight}}^{\max} \defeq \frac{h_1(h_1 + 2R_{\text{E}}) }{h_1 + R_{\text{E}} } + \frac{h_2(h_2 + 2R_{\text{E}}) }{h_2 + R_{\text{E}} }. \label{eq:LOS_Lim}
\end{align}
This is derived through basic geometric considerations (see the Supplementary Material for a derivation).

Secondly, let us justify the modeling of intersatellite channels as pure-loss channels. The number of thermal photons impinging upon a satellite detector is determined by the orientation and field of view of the detector. For communication between satellites, the transmitters and detectors do not occupy fixed orientations with respect to the main sources of brightness. Indeed, there will exist best case and worst case orientations: In the best-case scenario, the satellite detector will face completely away from the Earth or Moon, so that their albedos are not within the detector's field of view whatsoever. In a worst case scenario, the detector will be oriented directly facing the Earth (as in uplink).

However, point-to-point quantum communication can always be optimized by choosing the physically directed channel which results in less thermal background photons at the detector; irrespective of the logical direction of communication. That is, each intersatellite channel can exchange quantum systems in the direction which achieves the best detector orientation with respect to background noise. By optimizing the physical orientation of an intersatellite quantum network, each receiver will only ever experience a fraction of the worst background noise experienced by satellite uplink channels for which thermal corrections are minimal for link lengths of ${z \lesssim 10000 \text{ km}}$ \cite{SQC}.
We leave more formal treatments of these channel properties to future works, with the confidence that pure-loss channels accurately model such free-space links.

Hence, we can reliably model intersatellite free-space links as pure-loss channels. As such, we treat the upper-bound in Eq.~(\ref{eq:InterSatCap_Fad}) as an achievable rate so that  ${\mc{C} = \mc{B}_{F_{\sigma_{\text{p}}}} (\eta_{\text{d}})}$ can be accomplished by an optimal point-to-point protocol.

\section{Modular Quantum Networks \label{sec:Mod}}

\subsection{Network Model}

In this work, we construct a simple model for the study of modular quantum networks. Namely, we consider a global network $\mc{N} = (P,E)$ which consists of a collection of sub-networks called \textit{communities}, where the $i^{\text{th}}$ community is denoted by the undirected sub-graph 
\begin{equation}
{\mc{N}_{c_{i}} = (P_{c_{i}}, E_{c_{i}})},~ P_{c_{i}} \subset P, ~ E_{c_{i}} \subset E.
\end{equation}
Here, $P_{c_{i}}$ defines a subset of all network nodes that compose the $i^{\text{th}}$ community, while  $E_{c_{i}}$ denotes the subset of all network edges that connect them. For now, we consider each community network to be completely general, and can adopt an arbitrary topology. 
Here, we focus on quantum networks which observe \textit{spatial-}modularity \cite{Gross2020}, such that communities are spatially separated. This means that each community is completely disconnected from every other community, i.e.~the community node sets are all pairwise disjoint ${P_{c_i} \cap P_{c_j} = \varnothing}$, for all $i,j$.

In order to mediate communication between different communities, we introduce a \textit{backbone network} ${\mc{N}_b = (P_b, E_b)}$. This is a large-scale network for which none of its nodes $\bs{x}\in P_b$ are user nodes, used purely to facilitate end-to-end communications between users contained in different communities. Crucially, we assume that each community possesses a set of undirected edges which connect a set of community nodes to backbone network nodes. We refer to these as \textit{intercommunity edges}, such that the set of intercommunity edges
\begin{align}
E_{c_{i}:b} &\defeq \{ (\bs{x},\bs{y})\in E ~|~\bs{x}\in P_{c_{i}}, \bs{y}\in P_b\},
\end{align}
gives each community access to the backbone.

More precisely, we can define an intercommunity sub-network $\mc{N}_{c_{i}:b} = (P_{c_{i}:b}, E_{c_{i}:b})$, which describe the undirected graph that emerges between the $i^{\text{th}}$ community and the backbone. The set $P_{c_{i}:b}$ defines the complete collection of nodes that are interconnected between the community and the backbone. However, the nodes $\bs{x}\in P_{c_{i}:b}$ are already contained within $\mc{N}_{c_i}$ or ${\mc{N}}_b$; hence, it is important to distinguish between the community nodes and the backbone nodes which comprise this sub-network. For this, we introduce the notation
\begin{align}
P_{c_{i} | b} &\defeq P_{c_{i}:b} \cap P_{c_{i}} \subseteq P_{c_i},\\
P_{b|c_{i}}  &\defeq P_{c_{i}:b} \cap P_{b} \subseteq P_{b}.
\end{align}
Intuitively, $P_{c_{i} | b}$ can be thought of as the subset of nodes from the community $P_{c_{i}}$ conditioned on being connected to $\mc{N}_b$ (and vice versa for $P_{b | c_{i}}$).

This modular structure takes a very intuitive form and is remarkably useful for modeling realistic, hybrid quantum networks. When an equivalence relation is enforced between nodes in similar communities, the network quotient graph can be viewed as a star-network \footnote{In the Supplementary Material we present more general aspects of networks with community structures from which this modular network emerges as a useful and highly desirable class}. It allows us to completely separate communities and the backbone from one another. This makes it easier to compartmentalize different sub-network structures which may operate in completely different physical domains. Furthermore, it helps to derive independent network conditions on each of the sub-networks in accordance with some global objective. We summarize this architecture in the following definition which has also been illustrated in Fig.~\ref{fig:ModNet}(a).

\begin{defin} \emph{(Modular Network)}: A modular network $\mc{N} = (P,E)$ is a network architecture constituent of $n$ community sub-networks $\{ \mc{N}_{{c_{i}}} \}_{i=1}^n$, and a backbone sub-network $\mc{N}_b$. Each community sub-network is connected to the backbone via a set of edges $E_{c_{i} : b}$, described by the intercommunity sub-networks $\{ \mc{N}_{c_{i} : b} \}_{i=1}^n$, and there are no direct links between communities.
\label{defin:Mod}
\end{defin}

\begin{figure*}
\includegraphics[width=0.58\linewidth]{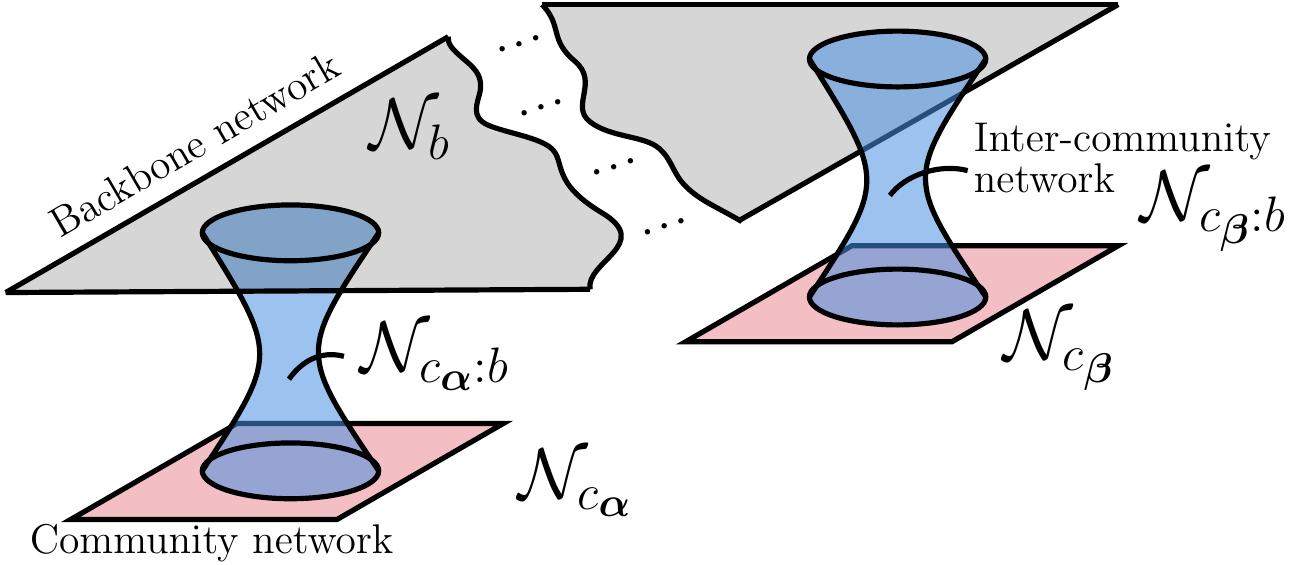}\hspace{0.8cm}
\includegraphics[width=0.3\linewidth]{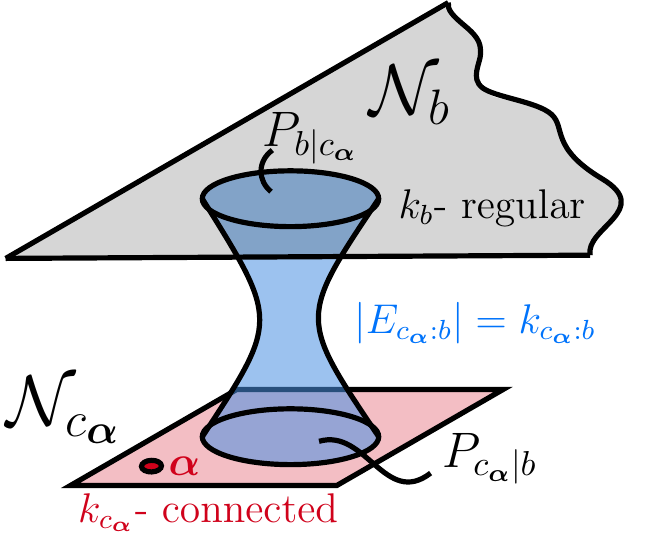}\\
\hspace{1cm}\\
\hspace{1.cm} (a) General Modular Network, \hspace{4.5cm} (b) Connectivity Properties.
\caption{(a) A modular quantum network architecture, constituted from community sub-networks $\mc{N}_{c_{\bs{\alpha}}}$, $\mc{N}_{c_{\bs{\beta}}}$ and a backbone network $\mc{N}_b$. Each community is connected to the backbone via the sub-networks $\mc{N}_{c_{\bs{\alpha}} : b}$ and $\mc{N}_{c_{\bs{\beta}} : b}$. Nodes from the community $c_{\bs{j}}$ which are directly connected to the backbone are contained in $P_{c_{\bs{j}}|b}$, while the nodes in the backbone which are connected to the community are contained in $P_{b|c_{\bs{j}}}$. (b) We may idealize this modular structure by placing ideal connectivity constraints on the each of the sub-networks. }
\label{fig:ModNet}
\end{figure*}

\subsection{Modular Network Capacities}

As discussed in Section~\ref{sec:Prelims}, the optimal end-to-end performance within a quantum network is quantified by its multi-path or \textit{flooding capacity} $\mc{C}^m( \mc{N})$, which describes the optimal number of target bits that can be transmitted between end-users per use of a flooding protocol. 
Any quantum network  $\mc{N} = (P,E)$, including the modular designs introduced, can be represented by a global distribution of channels $\{ \mc{E}_{\bs{xy}} \}_{(\bs{x},\bs{y})\in E}$ and a corresponding distribution of single-edge channel capacities $\{ \mc{C}_{\bs{xy}} \}_{(\bs{x},\bs{y})\in E}$, where $\mc{C}_{\bs{xy}} \defeq \mc{C}(\mc{E}_{\bs{xy}})$.
For general fading networks, it is always possible to use these distributions and the general expressions from Eqs.~(\ref{eq:MPR_loss}) and (\ref{eq:MPR_th}) in order to determine the flooding capacity. 

However, the translation into a modular architecture means that there exist particular classes of network cuts which are performed on different sub-networks. It becomes very useful to formally define a number of the important multi-edge capacities associated with these classes of cuts. In each of the following settings, we consider a pair of end-users  $\{ \bs{\alpha}, \bs{\beta} \}$ contained within remote communities of a generic, global modular network i.e.~$\bs{\alpha} \in P_{c_{\bs{\alpha}}}$ and $\bs{\beta} \in P_{c_{\bs{\beta}}}$ such that $c_{\bs{\alpha}} \neq {c_{\bs{\beta}}}$. It is now useful to denote community sub-networks with respect to the end-user that they contain, i.e.~we may write $c_{\bs{\alpha}}$ and $c_{\bs{\beta}}$ respectively. We assume each sub-network to adopt arbitrary topologies and capacity distributions.

\subsubsection{Local-Community Capacities}

We define a \textit{local-community cut} $C_{c_{\bs{j}}}$ as that which partitions two end-users within the network by exclusively collecting edges within one of the user communities $c_{\bs{j}}$, for either $\bs{j}\in\{\bs{\alpha},\bs{\beta}\}$. That is, a local-community cut-set takes the form
$
{\tilde{C}_{c_{\bs{j}}} = \{ (\bs{x},\bs{y})\in E_{c_{\bs{j}}} ~|~ \bs{x} \in \bff{A}, \bs{y} \in \bff{B} \}}.
$
This restricted form of network cut will generate an associated multi-edge capacity, which we label a \textit{local-community capacity},
\begin{equation}
\mc{C}_{c_{\bs{j}}}^m \defeq \min_{C_{c_{\bs{j}}}}\hspace{-0mm}\sum_{(\bs{x},\bs{y})\in \tilde{C}_{c_{\bs{j}}}} \mc{C}_{\bs{xy}}.
\end{equation}
For end-user nodes $\bs{j}$ which do not share a direct connection with the backbone (i.e.~$\bs{j}  \notin P_{c_{\bs{j}} | b}$) then this form of restricted cut always exist. 

However, if an end-user node does share a direction connection with the backbone, then a valid local-community cut will not exist. In this case it is never sufficient to remove edges solely from the community networks, and one must cut at least one edge from the set of intercommunity edges $E_{c_{\bs{j}} :b}$. To this end, we must slightly modify the local-community cut so that it removes any direct connections from the user node to the backbone, and then to identify the optimal set of edges to be removed from the community. Hence, a valid cut-set becomes
$
{\tilde{C}_{c_{\bs{j}}}^{\prime} = \{ (\bs{j} ,\bs{y}) \in E~|~\bs{y} \in P_b\} \cup \tilde{C}_{c_{\bs{j}}}.}
$
We can then define an analogous local-community capacity according to this class of network cut.

\subsubsection{Backbone Capacities}

A \textit{backbone cut} $C_b$ is a network cut that exclusively collect edges on the backbone network in order to partition the two end-users. This kind of cut-set takes the form
$
{\tilde{C}_{b} = \{ (\bs{x},\bs{y})\in E_{b} ~|~ \bs{x} \in \bff{A}, \bs{y} \in \bff{B} \}},
$
which generates an associated multi-edge \textit{backbone capacity},
\begin{equation}
\mc{C}_{b}^m \defeq \min_{C_{b}}\hspace{-0mm}\sum_{(\bs{x},\bs{y})\in \tilde{C}_{b}} \mc{C}_{\bs{xy}}.
\end{equation}
In the modular network architecture we are investigating, when considering end-users contained in unique communities, these kinds of cuts always exist. It is always sufficient to perform a cut on the backbone since there does not exist any other collection of edges that can be used to form a valid path between communities.

\subsubsection{Global-Community Capacities}

Finally, we can formalize a multi-edge capacity associated with exclusively collecting intercommunity edges. The end-user communities $\mc{N}_{c_{\bs{\alpha}}}$ and  $\mc{N}_{c_{\bs{\beta}}}$ are connected to the backbone via the sets of intercommunity edges $E_{c_{\bs{\alpha}}:b}$ and $E_{c_{\bs{\beta}}:b}$ respectively. If we removed either of these sets of edges, then the two remote users would be automatically partitioned. 
Hence, the edge sets $E_{c_{\alpha}:b}$ and $E_{c_{\beta}:b}$ both correspond to valid cuts on the network and each generate a multi-edge capacity
\begin{equation}
\mc{C}_{c_{\bs{j}}:b}^m \defeq  \sum_{(\bs{x},\bs{y})\in {E}_{c_{\bs{j}} :b}} \mc{C}_{\bs{xy}},
\end{equation}
for $\bs{j}\in\{\bs{\alpha},\bs{\beta}\}$. We can then minimize over the end-users to define a multi-edge capacity,
\begin{equation}
\mc{C}_{c:b}^m \defeq \min_{\bs{j}\in\{\bs{\alpha},\bs{\beta}\}} \mc{C}_{c_{\bs{j}}:b}^m. \label{eq:GlobCap}
\end{equation}
Clearly, this form of network cut always exists. We refer to this kind of partitioning as community isolation, since it isolates a community sub-network entirely from the rest of the network. Furthermore, we name $\mc{C}_{c:b}^m$ the \textit{global-community capacity}, as it refers to globally isolating the entire community sub-network \footnote{While it might be more convenient to call this the \textit{intercommunity capacity}, such a name might be confused as a more general term for the capacity when the end-users are located in different communities (which is implied). The global-community capacity is intended to be more distinct than this, and specify a particular network cut.}.

\subsection{Idealized Modular Networks}

Arbitrary architectures can always be treated using the capacity expressions from Section \ref{sec:Prelims} for general fading networks. However, the generality of these arguments make it difficult to present rigorous analytical statements about specific features or tangible network properties. In order to understand the ultimate potential of quantum networks, we need to simultaneously optimize the point-to-point channels and the network architecture in which they are arranged. Hence, it is desirable to strike a balance between realism and ideality in such a way that allows us to derive informative results about quantum networks and end-to-end performance. 
In the following we propose sub-network connectivity constraints that strike this balance.

\subsubsection{Backbone Regularity}
Firstly, we can impose regularity on the network backbone, demanding that the degree of each node is constant. This leads to a highly-connected network structure which is ideal for multi-path routing strategies. Let the function $\text{deg}(\bs{x})$ compute the degree of the node $\bs{x}$. Then we impose
\begin{align}
\text{deg}(\bs{x}) &= k_{b}, \forall \bs{x} \in P_{b},
\end{align}
which defines the regularity parameter of the backbone. It is important to make clear that these constraints only apply to intra-network connections. Indeed, a node on the backbone can have $k_b$ connections to neighbors on the backbone network, but also possess additional intercommunity connections via the sub-network $\mc{N}_{c_{\bs{j}}:b}$, without any further constraint. It is useful to quantify the number of intercommunity connections permitted between the backbone and communities using the notation,
\begin{equation}
k_{c_{\bs{j}} : b} = | E_{c_{\bs{j}} : b} |,  ~\bs{j} \in \{\bs{\alpha},\bs{\beta}\}.
\end{equation}
While regularity is an idealized property of realistic networks, in the context of a non-user repeater network such as the backbone it is very much feasible and extremely useful in order to understand the limits of quantum networks. 

\subsubsection{Community Connectivity}

Community sub-networks are likely to be smaller scale and less predictable structures than the backbone, partly due to the presence of user-nodes. Thus flexibility in their design is important. Here, we do not impose regularity but instead define classes of communities in accordance with the smallest local-community cut that they contain.

\begin{defin} \emph{($k_{c}$-connectivity)}: Consider a community sub-network $\mc{N}_c$. We say the community is ${k_{c}\text{-connected}}$ if $k_c$ is the smallest number of edges that must be removed in order to disconnect a pair of community nodes, minimized over all possible node pairs $\bs{x} \neq \bs{y} \in P_c$. More precisely, 
\begin{equation}
k_c \defeq \min_{\bs{x}\neq\bs{y} \in P_{c}} |\tilde{C}_{c}|.
\end{equation}
where $\tilde{C}_{c}$ denotes a community cut-set between the nodes $\bs{x}$ and $\bs{y}$. 
\end{defin}

Hence, $k_c$ defines the minimum local community cut-set cardinality, given some network topology and choice of end-users. This is a completely general property which is unique for all community networks, using the most easily disconnected pair of nodes in the network as a metric for how well it is connected. Regular networks are an example of an architecture for which their $k_c$-connectivity is simply equal to the network regularity. Hence, we can consider community sub-networks to be $k_{c}$-connected while encompassing a very large set of architectures.

\subsubsection{Idealized Modular Network}

Combining the constraints of regularity on the backbone and $k_{c}$-connectivity on the community sub-networks, it is possible to define an ideal modular quantum network architecture in terms of these parameters. This generates a structure that can be investigated analytically in the following sections.

\begin{defin} \emph{(Ideal Modular Network)}: An ideal modular network $\mc{N}^* = (P,E)$ is a network architecture constituent of $n$ community sub-networks $\{ \mc{N}_{{c_{{i}}}} \}_{i=1}^n$ each of which are $k_{c_i}$-connected, and a backbone sub-network $\mc{N}_b$ which is $k_b$-regular. Each community sub-network is connected to the backbone via $k_{c_{{i}} : b}$ edges, described by the intercommunity sub-networks $\{ \mc{N}_{c_{{i}} : b} \}_{i=1}^n$, and there are no direct links between communities.
\label{defin:RegMod}
\end{defin}

An illustration of this architecture can be found in Fig.~\ref{fig:ModNet}(b). When focusing on a particular pair of end-user nodes $\{ \bs{\alpha}, \bs{\beta}\}$ from two remote communities in the global network, we can then specify their $k_{c_{\bs{j}}}$-connectivity properties.

\subsection{Minimum-Cut as Community Isolation}

Care must be taken when constructing this form of modular network so to ensure not only high-rate communication within each community, but also high-rate communication between different communities mediated by the backbone. If the backbone network is poorly connected, or possesses weak links, it will not effectively assist long-distance communication. Meanwhile, even if communities are connected to a high quality backbone, insufficiently strong capacities in a local-community can compromise its use. Hence, there exists a careful balance between all of the sub-networks in the modular model, and their connectivity/capacity properties throughout. It is therefore highly desirable to identify a relationship between the quality of channels within the backbone and the quality of channels within the communities. 

In order to better grasp these relationships, we can investigate the ideal modular networks $\mc{N}^*$ defined in Definition~\ref{defin:RegMod}. Regular networks (such as that on the backbone) possess very convenient qualities which allow for useful insight into minimum network cuts. As such, they can be analytically studied as highly connected, ideal network structures and used to reveal fundamental limitations for end-to-end communication. 

Our mission becomes the following: To derive conditions on each of the sub-networks such that the flooding capacity between the remote users is always their global-community capacity. In this way, the minimum cut is always achieved by community isolation on either of the end-user communities. Equivalently, it means that the minimum cut can always be found on a simplified quotient graph of the modular network, vastly simplifying its analysis \cite{Note1}.
When this is the case, the end-to-end capacities between any two unique communities are always \textit{distance-independent}, i.e.~the ultimate rate between two end-user communities does not change with respect to the physical separation of those communities. This is an extremely desirable property of a quantum network, particularly on large-scales. 

If a modular network satisfies this property, it means that (\textit{i}) the backbone network is of sufficiently high quality that it never impedes the network performance over (potentially very) long distances, and (\textit{ii}) that the local-communities are of sufficiently high quality that neither compromises local or network-wide communication. Furthermore, by imposing that the minimum cut be the intercommunity edges, it allows us to reveal unique constraints on each sub-network, which are summarized in the following theorem.

\begin{theorem}
Consider an ideal modular network of the form $\mc{N}^*$ introduced in Definition~\ref{defin:RegMod}. Select any pair of end-users $\{\bs{\alpha},\bs{\beta}\}$ contained in remote communities $\bs{\alpha} \in P_{c_{\bs{\alpha}}}$ and $\bs{\beta} \in P_{c_{\bs{\beta}}}$. For all $\bs{j}\in\{\bs{\alpha},\bs{\beta}\}$, there exist single-edge threshold capacities on the communities $\mc{C}_{c_{\bs{j}}}^{\min}$ and backbone $\mc{C}_{b}^{\min}$ sub-networks for which the network flooding capacity is given by the global-community capacity,
\begin{align}
\begin{rcases}
\mc{C}_{\bs{xy}} \geq \mc{C}_{c_{\bs{j}}}^{\min}, \forall (\bs{x},\bs{y}) \in E_{c_{\bs{j}}},\\
{\mc{C}_{\bs{xy}} \geq \mc{C}_{b}^{\min}, \forall (\bs{x},\bs{y}) \in E_{b},}
\end{rcases}
\implies
\mc{C}^m( \mc{N}) = \mc{C}_{c : b}^m.
\end{align}
The threshold capacities are given by,
\begin{align}
 \mc{C}_{c_{\bs{j}}}^{\min} \defeq \frac{\mc{C}_{c : b}^m}{k_{c_{\bs{j}}} },~~
 \mc{C}_{b}^{\min} &\defeq \frac{\mc{C}_{c : b}^m}{ H_{\min}^* },
\end{align}
where $H_{\min}^*$ is the minimum cut-set cardinality on the backbone network. If these threshold capacities are violated, then the global-community capacity becomes an upper-bound on the end-to-end capacity, $\mc{C}^m( \mc{N}) \leq \mc{C}_{c : b}^m$. 
\label{theorem:Cond1}
\end{theorem}

A detailed proof can be found in Section~I of the Supplementary Material. Thanks to backbone regularity and community connectivity, the minimum cut-set cardinalities that occur within each sub-network can be easily identified. Then, it is straightforward to enforce single-edge capacity constraints which ensure that the local-community/backbone capacities are always larger than the global-community capacity. 

In this theorem, we have used the fact that the cardinality of the smallest backbone cut-set between two end-users in remote communities can be analytically derived, thanks to network regularity. This minimum cardinality takes the form
\begin{equation}
H_{\min}^* \defeq \min_{\bs{j}\in\{\bs{\alpha},\bs{\beta}\}} H_{\min}(k_b, P_{b| c_{\bs{j}}}),
\end{equation}
where $H_{\min}(k_b, P_{b| c_{\bs{j}}})$ is a function that computes the minimum number of edges that must be cut to isolate all the nodes $P_{b|c_{\bs{j}}}$ on the backbone which are also connected to the community $c_{\bs{j}}$. The explicit form of this expression in found in the Supplementary Material, and depends on the precise spatial arrangement of connections from the community to the backbone. However, we can generally bound this quantity using
\begin{equation}
k_b \leq H_{\min}(k_b, P_{b| c_{\bs{j}}}) \leq k_b |P_{b| c_{\bs{j}}}|.
\end{equation}
The lower-bound $k_b$ corresponds to a worst-case spatial distribution of community-to-backbone connections, when all the community nodes are connected to the same node on the backbone, i.e.~$|P_{b| c_{\bs{j}}}| = 1$. Then it is sufficient to isolate just one backbone node to perform a valid end-user cut, collecting only $k_b$ edges (since the backbone is $k_b$-regular). The upper-bound corresponds to a best-case scenario; when all the community nodes are connected to backbone nodes which don't share any neighbors or edges. In this case, the smallest cut-set restricted to the backbone is found by isolating all nodes individually. As a result, this cut collects exactly $k_b |P_{b| c_{\bs{j}}}|$ edges.

Fig.~\ref{fig:Kmins_Mod} depicts a number of examples of minimum backbone cut-sets for remote communities connected to a Manhattan backbone ($k_b=4$). In these figures we display only one end-user community and assume that the other end-user community is sufficiently distant that it does not share intercommunity connected nodes on the backbone.

As a result, we can always present a best and worst-case single-edge threshold capacity for the backbone network, $\mc{C}_b^{\min}$. That is, we can sandwich the backbone threshold capacity according to 
\begin{equation}
\frac{\mc{C}_{c:b}^m}{k_b |P_{b| c_{\bs{j}}}|}  \leq \mc{C}_{b}^{\min} \leq \frac{\mc{C}_{c:b}^m}{k_b} ,~~ \bs{j}\in\{\bs{\alpha},\bs{\beta}\}.
\end{equation}
The more effectively that the intercommunity connections are dispersed across the backbone, the weaker the single-edge constraint that must be forced upon it.

\begin{figure}
(a) $H_{\min} = 16$,\hspace{8mm} (b) $H_{\min} = 12$,\hspace{8mm} (c) $H_{\min} = 4$.\\
\includegraphics[width=0.3255\linewidth]{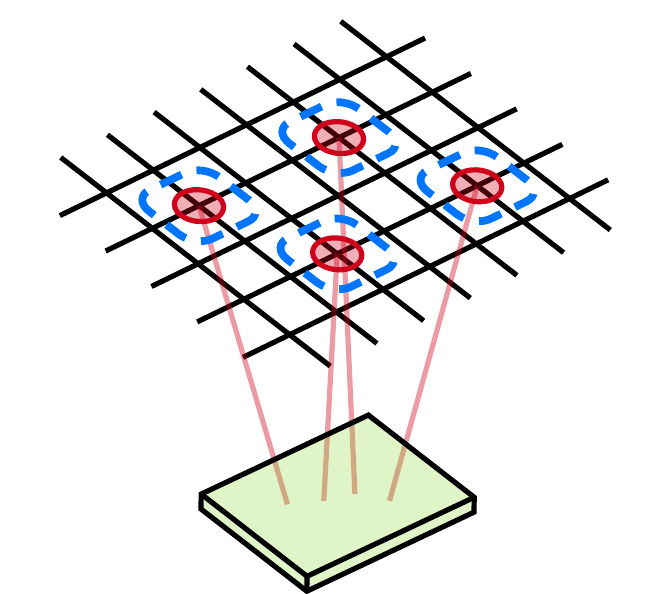}
\includegraphics[width=0.3255\linewidth]{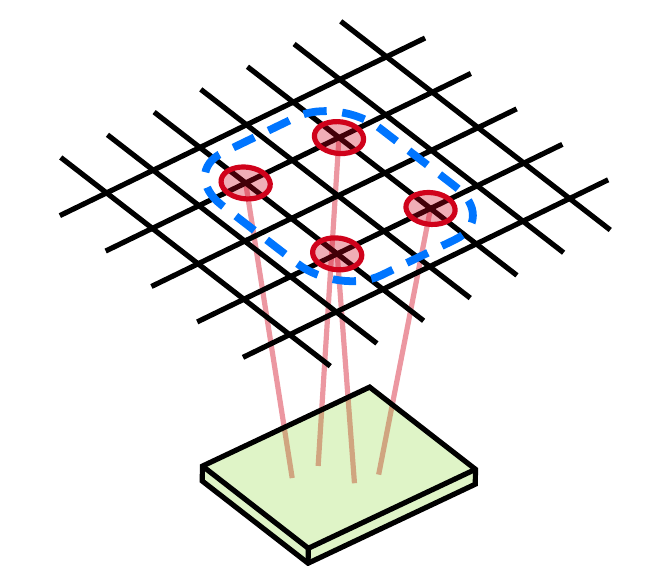}
\includegraphics[width=0.3255\linewidth]{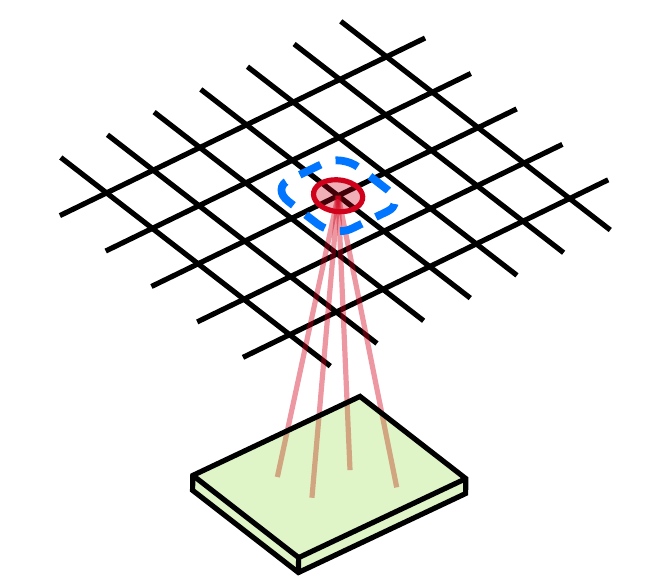}
\caption{Examples of minimum cardinality intercommunity cut-sets for connections from an arbitrary community to a Manhattan backbone network ($k_b=4$). These are valid cuts which isolate remote communities (only one community is illustrated here), and are performed exclusively on the backbone. Panel (a) captures the best-case spatial distribution of the largest potential cut-set when no target-nodes share any edges or neighbors, (b) illustrates an example in which neighbor sharing can diminish the overall cut-set size, and (c) describes the worst-case spatial distribution that minimizes the cut-set size.}
\label{fig:Kmins_Mod}
\end{figure}

\section{Hybrid Free-Space Networks \label{sec:HybridFS}}

In this final section we combine results and theory from Section \ref{sec:Prelims} and \ref{sec:Mod} in order to investigate hybrid fiber/free-space modular quantum networks. Here we study two pertinent cases in an effort to reveal tangible resource requirements for future quantum networks: A fiber/satellite modular configuration and a ground-based free-space/fiber model. 

\subsection{fiber/Satellite Configuration}
\subsubsection{Motivation}
An interesting modular configuration consists of fiber-based community networks which are interconnected via a backbone satellite network. This model captures a realistic satellite-based model of the quantum internet, in which dynamic intersatellite links are used to facilitate long distance quantum communication at high-rates. 
In this scenario, the weakest links are typically the ground-to-satellite free-space connections, due to the impact of atmospheric decoherence and turbulence on a transmitted beam. Therefore, the constraints revealed in Theorem~\ref{theorem:Cond1} are very realistic, as community isolation is likely to be the minimum cut in many settings. \par

In Theorem~\ref{theorem:Cond1} we devise single-edge capacity lower bounds on the community networks which guarantee the network flooding capacity is equal to the global-community capacity. For fiber-based networks, these single-edge lower bounds can be used to identify a \textit{maximum tolerable fiber-length}, $d_{c_{\bs{j}}}^{\max}$ that is permitted within the fiber-network. 
In the context of a satellite-based backbone network, the single-edge capacity lower bound can be translated into a \textit{maximum intersatellite separation}, $z_{b}^{\max}$ which describes the maximum propagation distance that is permitted for free-space channels between satellites in the backbone. 
These are critical quantities which directly motivate the construction of ground-based and satellite-based networks for global quantum communication.

\subsubsection{Optimal Performance}

We wish to enforce that the minimum cut is always achieved by community isolation, generating the global-community capacity $\mc{C}^m(\mc{N}) = \mc{C}_{c:b}^m$. In this physical setting, each intercommunity edge is described by ground-to-satellite channel which may be an uplink or downlink channel. Thanks to teleportation, a network protocol can always choose the physical channel direction that maximizes its point-to-point capacity independently from the desired logical direction of community. Downlink channels are always superior to uplink, and therefore we can simply model the global community capacity as the sum of a downlink capacities. This multi-edge capacity will be bounded by
\begin{align}
\mc{C}_{c:b}^m 
&\leq \min_{\bs{j}\in\{\bs{\alpha},\bs{\beta}\}} \hspace{-0mm} \sum_{(\bs{x},\bs{y}) \in E_{c_{\bs{j}}:b}}\hspace{-2mm}  \mc{L}_{F_{\bs{xy}}}(\eta_{\bs{xy}}, \bar{n}_{\bs{j}}),\\
&\leq \min_{\bs{j}\in\{\bs{\alpha},\bs{\beta}\}} \hspace{-0mm} \sum_{(\bs{x},\bs{y}) \in E_{c_{\bs{j}}:b}}\hspace{-2mm}  \mc{B}_{F_{\bs{xy}}}(\eta_{\bs{xy}}).
\end{align}
where $F_{\bs{xy}}$ and $\eta_{\bs{xy}}$ capture the fading dynamics and maximum transmissivity of each downlink channel that connect $c_{\bs{j}}$ to the backbone, and depend on beam trajectory. Meanwhile, $\bar{n}_{\bs{j}}$ infers community-wide thermal-noise conditions. Since all of the intercommunity edges in $E_{c_{\bs{j}} : b}$ are connected to a relatively small area, we can assume identical operational conditions for all downlink edges. However, these operational conditions will not be consistent for both end-users; when communicating on a global scale, one user may be in night-time while the other is in day-time with independent weather conditions.

\begin{table}[t!]
\begin{tabular}{ |l|c|c| } 
 \hline
\textit{\hspace{8.5mm}Parameter} & ~\textit{Symbol}~ & ~\textit{Value}~ \\ 
 \hline  \hline
Beam Curvature & $R_{0}$ & $\infty$  \\  \hline
Wavelength & $\lambda$ & 800~nm  \\  \hline
Initial spot-size & $\omega_0$ & $\begin{array}{l} 40~\text{cm}\hspace{1.75mm}\text{ - Setup (\#1)}\\ 20~\text{cm }\hspace{0.7mm}\text{ - Setup (\#2)}\end{array}$  \\  \hline
Receiver Aperture & $a_R$ & $\begin{array}{l} 1~\text{m }\hspace{3.9mm}\text{ - Setup (\#1)}\\ 40~\text{cm}\hspace{2mm}\text{ - Setup (\#2)}\end{array}$ \\  \hline
Detector Efficiency & $\eta_{\text{eff}}$ & 0.4 \\  \hline
Detector Noise & $\bar{n}_{\text{ex}}$ & $\approx$ 0 \\  \hline
Pointing error  & $\sigma_{\text{p}}^2$ & $1~\mu$rad $\approx (10^{-6} z)^2$ \\   \hline
Pulse Duration  & $\Delta t$ & $10$ ns \\   \hline
Field of View  & $\Omega_{\text{fov}}$ & $10^{-10}$ sr \\   \hline
Frequency Filter & $\Delta \lambda$ & $\begin{array}{l} 0.1~\text{pm}\hspace{1.75mm}\text{- Setup (\#1)}\\ 1~\text{nm}\hspace{3.2mm}\text{ - Setup (\#2)}\end{array}$   \\  \hline
interCommunity Link & ICL & Downlink\\ \hline
fiber Loss-Rate & $\gamma$ & $0.02$ per km \\ \hline
\end{tabular}
\caption{Parameter table for the fiber/satellite modular network configuration. Here we consider two similar setups using a collimated Gaussian beam at 800~nm wavelength, but differ in initial spot-size $w_0$, receiver aperture $a_R$ and frequency filter $\Delta \lambda$. }
\label{table:Setups}
\end{table}

\begin{figure*}[t!]
\hspace{-1cm}Setup (\#1), Clear day-time,\hspace{2cm} Setup (\#2), Clear night-time.\\
\includegraphics[width=\linewidth]{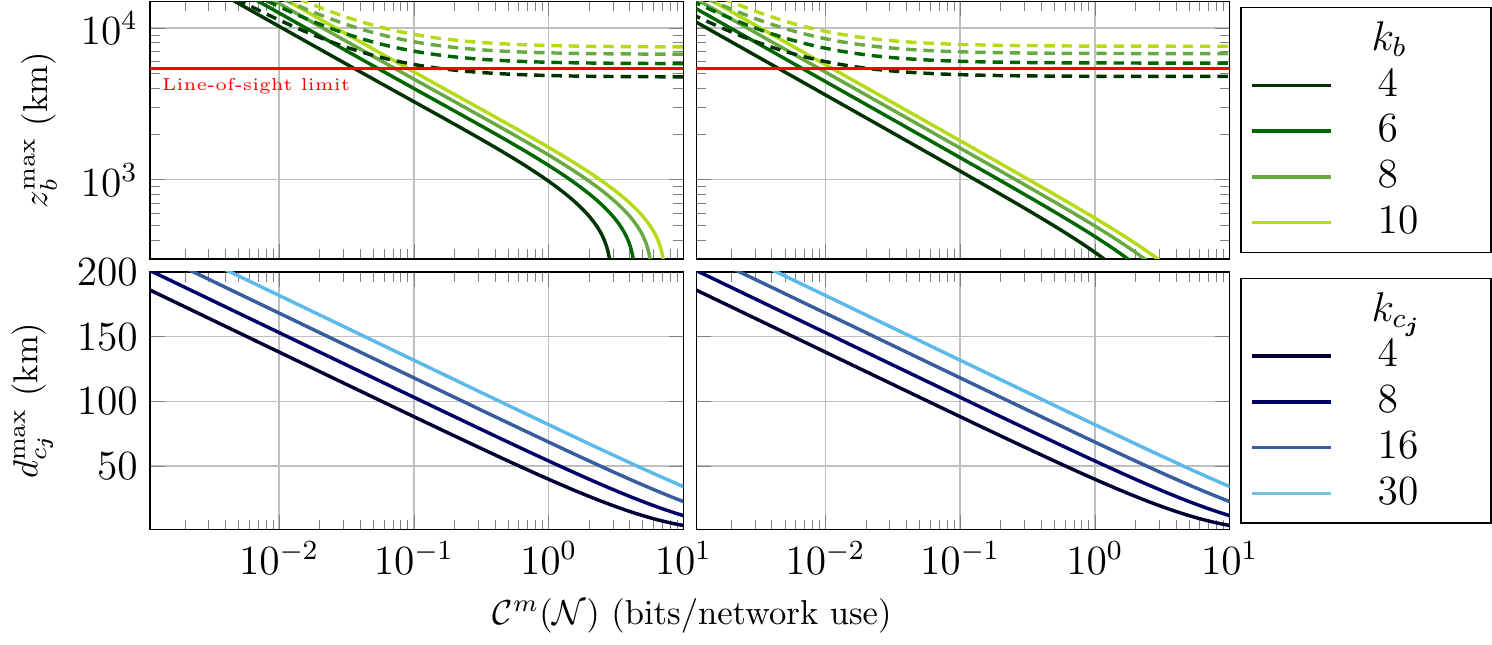}\\
\caption{Optimal end-to-end performance for an ideal modular network consisting of fiber communities interconnected to a satellite-based backbone. 
In order to guarantee an optimal flooding rate along the $x$-axis then the maximum internodal separations in each sub-network on the $y$-axis must be less than or equal to the plotted bounds. 
We consider operational settings in Setup (\#1) for (a), (c) and Setup (\#2) for (b), (d) which are described in Table \ref{table:Setups}. The weather/time conditions are those experienced by the worst-case end-user community. Given an optimal flooding capacity $\mc{C}^m(\mc{N})$, we plot the maximum intersatellite separation $z_{b}^{\max}$ for different backbone connectivity parameters, and the maximum fiber-length in each community $d_{c_{\bs{j}}}^{\max}$ for different community connectivity parameters. The dashed lines in Figs (a) and (c) plot an upper bound the maximum intersatellite separation based on the optimal spatial distribution of (a finite number of) community connected satellite nodes $P_{b|c_{\bs{j}}}$ at a maximum altitude $h^{\max} = 1500$ km, while the solid lines plot the lower bound based on the worst spatial distribution (for any altitude). The red line indicates the maximum achievable channel length that can be achieved for two satellites at altitude ${1500 \text{ km}}$, such that ${z_{\text{sight}}^{\max} \approx 5428 \text{ km}}$.  }
\label{fig:Satfiber_Res}
\end{figure*}

We can derive single-link distance constraints which guarantee $\mc{C}_{c:b}^m$ to be the optimal network capacity. These conditions follow directly from Theorem~\ref{theorem:Cond1} and are summarized in the following corollary:

\begin{corollary}
Consider an ideal modular network of the form $\mc{N}^*$ introduced in Definition~\ref{defin:RegMod}, and assume optical-fiber communities networks $\mc{N}_{c_{\bs{\alpha}}}$, $\mc{N}_{c_{\bs{\beta}}}$ and a satellite-based backbone $\mc{N}_{b}$. Select any pair of end-users $\{\bs{\alpha},\bs{\beta}\}$ located in remote communities $\bs{\alpha} \in P_{c_\alpha}$ and $\bs{\beta} \in P_{c_\beta}$. 
There exists a maximum fiber-length in each community
\begin{equation}
d_{{{c_{\bs{j}}}}}^{\max} \defeq -\frac{1}{\gamma} \log_{10}\left(1 - 2^{-{\mc{C}_{c : b}^m}/{k_{c_{\bs{j}}} }}\right),
\end{equation}
and a maximum intersatellite separation in the backbone
\begin{equation}
z_{b}^{\max} \defeq \argmin_{z} \left| \log\left(\frac{ H_{\min}^* \mc{B}_{F_{\sigma_{\text{\emph{p}}}}}(\eta) }{{\mc{C}_{c:b}^m}} \right) \right|. \label{eq:MaxIntSat}
\end{equation}
for which the network flooding capacity is equal to the global-community capacity, \begin{equation}
\mc{C}^m( \mc{N}) = \mc{C}_{c : b}^m.
\end{equation}
Otherwise, if any intersatellite links violate this condition ${\exists\> z_{\bs{xy}} > z_{{b}}^{\max}}, (\bs{x},\bs{y})\in E_b$ or the local community links are in violation, ${\exists \>d_{\bs{xy}} > d_{{{c_{\bs{j}}}}}^{\max}}, (\bs{x},\bs{y})\in E_{c_{\bs{j}}}$, for either $\bs{j} \in \{\bs{\alpha},\bs{\beta}\}$, then this becomes an upper-bound on the network flooding capacity, $\mc{C}^m( \mc{N}) \leq \mc{C}_{c : b}^m$. 
\label{corollary:Cor1}
\end{corollary}

The analytical simplicity of the maximum fiber-length follows from the remarkably compact PLOB bound for bosonic lossy channels. However, the maximum intersatellite separation in Eq.~(\ref{eq:MaxIntSat}) must be computed numerically due to the more complex PLOB bound which accounts for fading due to pointing errors. The lack of onboard access makes it difficult to perfectly optimize beam trajectory, and thus pointing errors cannot be ignored. However, it is possible to analytically upper and lower bound the quantity $z_{b}^{\max}$. \par

An upper-bound is found by considering a lack of pointing errors, which means the channel is no longer a fading channel but is instead a fixed lossy channel with the maximum possible transmissivity. This idealizes the intersatellite channel by removing the potential for beam wandering, resulting in an upper-bound for the maximum separation. Meanwhile, we can find a lower-bound on the maximum intersatellite separation by considering the use of slow detectors. A slow detector at the receiver will not be able to resolve pointing errors, resulting in a lossy channel with fixed transmissivity averaged over the entire fading process. Interestingly, the rate in bits per channel use via slow detection can in some instances be higher than that for fast detectors which resolve the fading dynamics. However, the slower detection time severely limits the \text{operational rate} at which the channel can be used (or clock rate). As a result, the point-to-point communication rate via slow detection will be orders of magnitude smaller than those with fading-resolving setups. This information can be used to write a lower-bound on $z_{b}^{\max}$. For explicit details on these bounds, see the Supplementary Material.

The maximum intersatellite separation $z_{b}^{\max}$ describes a maximum \textit{tolerable} channel length permitted within the backbone network. Yet, it is not always true that such a channel length is achievable due to line-of-sight limitations associated with orbital geometry. This is quantified by the maximum line-of-sight distance from Eq.~(\ref{eq:LOS_Lim}), a function of the altitudes of the communicating satellites. Crucially, if we find that ${z_{b}^{\max} \geq z_{\text{sight}}^{\max}}$ for some network configuration and desirable rate, this means that the satellites within the backbone can reliably communicate with \textit{any} other satellite that fall within its line-of-sight, without compromising performance. This is an extremely useful property, providing significant flexibility for satellite backbone networks.

\subsubsection{Discussion}

Fig.~\ref{fig:Satfiber_Res} offers insight into the constraints proposed by Corollary~\ref{corollary:Cor1} for satellite-fiber modular networks corresponding to a number of different physical settings and network properties. Here we consider two free-space communication setups described in Table~\ref{table:Setups}: Setup (\#1) in Figs.~(a) and (c) and Setup (\#2) in Figs.~(b) and (d). 

Consider a flooding capacity $\mc{C}^m(\mc{N})$ that is desired between the two end-users who are located in remote, fiber communities. The actual ground distance between the users or unique communities is irrelevant, and can be arbitrarily situated at any location across the Earth. If that flooding capacity is to be achieved, then for a given modular architecture there exists a maximum fiber-length $d_{c_{\bs{j}}}^{\max}$ permitted within the user community $c_{\bs{j}}$, and a maximum intersatellite separation $z_{b}^{\max}$ permitted throughout the backbone network. 

In Figs.~\ref{fig:Satfiber_Res}(a) and (b) we plot the behavior of the maximum intersatellite separation with respect to desired flooding capacity. In the solid lines, we plot the worst-case $z_b^{\max}$, which corresponds to the situation where all the downlink channels are connected to the same node on the backbone, allocating a single satellite to connect to a community. This is a worst-case situation because it means that the minimum cut on the backbone is very small, $H_{\min}^* = k_b$. Yet, even in this scenario, thanks to the lack of atmospheric decoherence we find that very large distances are permitted between satellites, such that ${z_b^{\max}\sim 10^{3}-10^4~\text{km}}$ can still ensure high flooding rates between the end-users communities on the Earth.

Meanwhile, the dashed lines plot $z_{b}^{\max}$ for the best-case spatial distribution of downlink connections on the backbone when the maximum satellite altitude is ${h^{\max} = 1500  \text{ km}}$ and all downlink beam trajectories are within a 1 steradian angular window. This means that the smallest backbone cut-set has the total number of edges
\begin{equation}
H_{\min}^* = k_b |P_{b|c_{\bs{j}}}|.
\end{equation}
In this case, the minimum cut-set cardinality on the backbone is very large, as the number of downlink channels must be increased in order to obtain the chosen flooding capacity. In this best-case scenario, as $\mc{C}^m(\mc{N})$ increases $z_b^{\max}$ begins to plateau, permitting large intersatellite separations even at large flooding capacities. This confirms a strong dependence between the distribution of intercommunity edges and the single-edge capacity properties of a backbone network. For all other distributions of intercommunity connections $P_{b|c_{\bs{j}}}$, the behavior of the maximum intersatellite separation falls between these bounds. 

We also display the maximum line-of-sight distance ${z_{\text{sight}}^{\max} \approx 5428 \text{ km}}$ between any pair of satellites orbiting at an altitude ${h^{\max} = 1500 \text{ km}}$. This is the longest intersatellite channel that can be established due to orbital geometry. Interestingly, even in the worst-case backbone configuration (each community possesses many connections to a single satellite) the line-of-sight limit is exceeded by $z_{b}^{\max}$ at relatively good rates such that ${\mc{C}^m(\mc{N}) \in [10^{-2}, 10^{-1}]}$ bits per network use. When $z_{b}^{\max} \geq z_{\text{sight}}^{\max}$ is true, satellites in the backbone may connect to any other satellite within its line-of-sight; hence this promises achievable and flexible constraints for intersatellite networks.

Figs.~\ref{fig:Satfiber_Res}(c) and (d) depict the maximum fiber-lengths permitted within $k_{c_{\bs{j}}}$-connected community networks to ensure a desired end-to-end flooding capacity. Of course, the quality of the bosonic lossy channels do not change with respect to Setups (\#1) and (\#2) and therefore Figs.~(c) and (d) are identical. As one would expect, the permissible channel lengths for strong end-to-end rates depend upon the community channels being ${d_{c_{\bs{j}}}^{\max} \lesssim 100 \text{ km}}$, even in a highly connected network setting. But thanks to the modular network configuration, this is not problematic. In this configuration, the community fiber-networks are designed to cover small areas (relative to the satellite backbone) and facilitate local communication. Quantum communication over global distances is then appropriately mediated by the satellite backbone.

As an example, let us focus on Setup (\#1) and consider a satellite backbone network with regularity $k_b = 4$ used to mediate long-distance quantum communication between two end-users $\{\bs{\alpha},\bs{\beta}\}$ contained within fiber-networks which are $k_{c_{\bs{\alpha}}} = 4$ and $k_{c_{\bs{\beta}}} = 8$ connected. What are the network constraints required to ensure that their flooding capacity is $\mc{C}^m(\mc{N}) = 1$ bit per network use? Provided that $z_{b}^{\max} \lesssim 1000$ km, that $d_{c_{\bs{\alpha}}}^{\max} \lesssim 30$ km and $d_{c_{\bs{\beta}}}^{\max} \lesssim 50$ km, then it is guaranteed that this flooding rate is achievable. This provides extremely valuable information for future quantum network designs; if an ideal modular network cannot exceed these constraints, then less ideal structures should take even stronger heed of them.

\subsection{Ground-Based Free-Space/fiber Configuration}

\subsubsection{Motivation}
It is also interesting to investigate the limits of ground-based quantum networks which are composed from a mixture of fiber channels and free-space channels. For this purpose, modular network architectures offer an appropriate and physically relevant model. One may consider a metropolitan network area which is spanned by a collection of free-space quantum networks, or ``hotspots". These are short-range communities within which reliable free-space quantum communications can take place. In order to communicate over a larger area and between free-space communities we can use an underlying optical-fiber backbone which mediates longer distance communication.

Utilizing the recently derived ultimate limits of ground-based, free-space quantum communication \cite{FS} we wish to determine whether free-space links are reliable enough to enable high-rate quantum communication in this setting. Furthermore, it is important to understand the requirements of the optical-fiber backbone required to facilitate wireless quantum networking.

\begin{table}[t!]
\begin{tabular}{ |l|c|c| } 
 \hline
\textit{\hspace{8.5mm}Parameter} & ~\textit{Symbol}~ & ~\textit{Value}~ \\ 
 \hline  \hline
Beam Curvature & $R_{0}$ & $\infty$  \\  \hline
Wavelength & $\lambda$ & 800~nm  \\  \hline
Initial spot-size & $\omega_0$ & 5 cm \\  \hline
Receiver Aperture & $a_R$ & 5 cm \\  \hline
Detector Efficiency & $\eta_{\text{eff}}$ & 0.5 \\  \hline
Detector Noise & $\bar{n}_{\text{ex}}$ & 0.05 \\  \hline
Pointing error  & $\sigma_{\text{p}}^2$ & $1~\mu$rad $\approx (10^{-6} z)^2$ \\   \hline
Pulse Duration  & $\Delta t$ & $10$ ns \\   \hline
Field of View  & $\Omega_{\text{fov}}$ & $10^{-10}$ sr \\   \hline
Frequency Filter & $\Delta \lambda$ & 1~\text{nm}   \\  \hline
Altitude & $h$ & 30\text{ m} \\ \hline
fiber Loss-Rate & $\gamma$ & $0.02$ per km \\ \hline
interCommunity Link & ICL & $\begin{array}{c} \text{Free-Space} \\ \text{(Clear day-time)}\end{array}$ \\  \hline
\end{tabular}
\caption{Parameter table for the free-space/fiber modular network configuration.}
\label{table:FSF_Setups}
\end{table}

\subsubsection{Optimal Performance}

It is possible to once more translate Theorem~\ref{theorem:Cond1} to establish conditions for which the flooding capacity is given by the global-community capacity, ensuring optimal end-to-end performance. Now, each community is a ground-based free-space community located at an altitude of $h = 30$ m, and we consider the intercommunity edges connecting each community to the backbone to also be free-space links. Furthermore, since our rigorous free-space capacities are restricted to the regime of weak-turbulence, then we must investigate free-space channels $\mc{E}_{\bs{xy}}$ which are no longer than $z_{\bs{xy}} \approx 1066$ m \cite{FS}. 

While this may at first appear restrictive, we remind the reader of the physical context; free-space communities are inherently designed for short-range networks with mobile users. Indeed, with network nodes that are limited to line-of-sight connections in a potentially urban area, focusing on the weakly turbulent range is natural. This leaves us with the remaining questions: Are free-space quantum channels resilient enough within this range to offer high-rate communication, and what are the resource requirements of the fiber backbone? We provide insight in the following corollary.

\begin{corollary}
Consider an ideal modular network of the form $\mc{N}^*$ introduced in Definition~\ref{defin:RegMod}, and assume free-space community networks $\mc{N}_{c_{\bs{\alpha}}}$, $\mc{N}_{c_{\bs{\beta}}}$ and an optical-fiber backbone $\mc{N}_{b}$. Select any pair of end-users $ \{\bs{\alpha},\bs{\beta}\}$ located in unique communities $\bs{\alpha} \in P_{c_\alpha}$ and $\bs{\beta} \in P_{c_\beta}$. 
There exists a maximum free-space link length in each community
\begin{equation}
z_{c_{\bs{j}}}^{\max} \leq \argmin_{z} \left| \log\left(\frac{k_{c_{\bs{j}}} \mc{L}_{F_{\sigma}}(\eta,\bar{n}_{\bs{j}}) }{{\mc{C}_{c:b}^m}} \right) \right|, \label{eq:GFSmax}
\end{equation}
and a maximum fiber length in the backbone
\begin{equation}
d_{b}^{\max} \defeq -\frac{1}{\gamma} \log_{10}\left(1 - 2^{-{\mc{C}_{c : b}^m}/{H_{\min}^* }}\right),
\end{equation}
for which the network flooding capacity is equal to the global-community capacity, \begin{equation}
\mc{C}^m( \mc{N}) = \mc{C}_{c : b}^m.
\end{equation}
Otherwise, if any fiber links violate this condition ${\exists\> d_{\bs{xy}} > d_{{b}}^{\max}}, (\bs{x},\bs{y})\in E_b$ or the local community links are in violation, ${\exists \>z_{\bs{xy}} > z_{{{c_{\bs{j}}}}}^{\max}}, (\bs{x},\bs{y})\in E_{c_{\bs{j}}}$, for either $\bs{j} \in \{\bs{\alpha},\bs{\beta}\}$, then this becomes an upper-bound on the network flooding capacity, $\mc{C}^m( \mc{N}) \leq \mc{C}_{c : b}^m$. 
\label{corollary:Cor2}
\end{corollary}

Notice that we now obtain an upper-bound on the maximum free-space link length, as it is not known whether the single-edge quantity $\mc{L}_{F_{\sigma}}(\eta,\bar{n}_{\bs{j}})$ is achievable or not. However, this bound has been shown to be tight and thus offers an accurate bound on $z_{c_{\bs{j}}}^{\max}$ \cite{FS}. Furthermore, this maximum free-space link length must be computed numerically due to the complex nature of the free-space PLOB bound which accounts for fading and thermal effects. Yet, the maximum fiber length within the backbone can be readily determined for an arbitrary distribution of intercommunity connections. 

\subsubsection{Discussion}

Fig.~\ref{fig:FSFib_Res} provides example network constraints using Corollary \ref{corollary:Cor2} for ideal modular networks and a variety of community and backbone connectivity properties. Operational parameters are found in Table~\ref{table:FSF_Setups} for this modular architecture. 
Given a desired end-to-end flooding capacity, we generate a maximum fiber length in the backbone $d_{b}^{\max}$ and maximum free-space link length in each community $z_{c_{\bs{j}}}^{\max}$ in Figs.~(a) and (b) respectively, such that this flooding capacity is achieved by the global-community capacity.

\begin{figure}[t!]
\includegraphics[width=\linewidth]{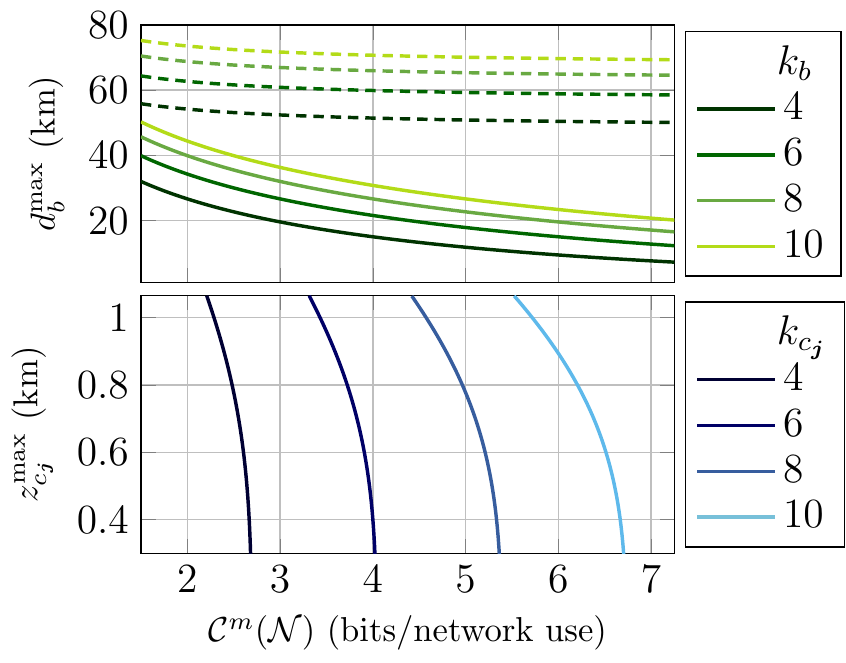}\\
\caption{Optimal end-to-end performance for an ideal modular network consisting of free-space communities interconnected to a fiber-based backbone. 
In order to guarantee an optimal flooding rate along the $x$-axis then the maximum internodal separations in each sub-network depicted on the $y$-axis must be less than or equal to the plotted bounds. We use the operational settings in Table \ref{table:FSF_Setups} during clear day-time. Given an optimal flooding capacity $\mc{C}^m(\mc{N})$, we plot the maximum fiber-length $d_{b}^{\max}$ for different backbone connectivity parameters, and the maximum free-space link-length in each community $z_{c_{\bs{j}}}^{\max}$ for different community connectivity parameters. The dashed lines plot an upper bound on the maximum fiber-length based on the optimal spatial distribution of community connected backbone nodes $P_{b|c_{\bs{j}}}$, while the solid lines plot a lower bound based on the worst-case spatial distribution. }
\label{fig:FSFib_Res}
\end{figure}

Immediately we notice that the flooding capacities plotted are large. This is because, as seen in Fig.~\ref{fig:FSFib_Res}(b), the free-space links are sufficiently capable in the weakly turbulent regime so that $z_{c_{\bs{j}}}^{\max} > 1$ km for flooding capacities as high as ${\mc{C}^m(\mc{N}) \approx 2 \text{ bits/network use}}$, even when the community connectivity is low e.g.~$k_{c_{\bs{j}}} = 4$. As the community connectivity gets larger, the free-space capacities become increasingly reliable within this distance range, and do not compromise the minimum cut until the flooding capacity becomes very large.

Yet, these large end-to-end capacities simultaneously place greater demands on the backbone network, demanding shorter links as the global-community capacity increases. The solid lines in Fig.~\ref{fig:FSFib_Res}(b) plot the maximum fiber-length corresponding to the worst-case spatial distribution of free-space connections from the communities to the backbone, i.e.~all intercommunity links are focussed on a single backbone node. Meanwhile, the dashed lines consider a best-case scenario in which all the intercommunity links are of maximum length $z_{c:b} = 1$ km, and are oriented such that they maximize the backbone cut-set cardinality $H_{\min}^* = k_b |P_{b|c_{\bs{j}}}|$. 

We find that this free-space/fiber modular architecture reports very feasible constraints on the free-space hotspots and fiber-backbone in order to guarantee a high end-to-end performance. For a regular fiber-based backbone with $k_b = 4$, and end-user communities which are $k_{c_{\bs{j}}} > 4$ connected, then one can guarantee an achievable flooding capacity of $\mc{C}^m(\mc{N}) = 2$ bits/network use given that the free-space links all fall within the weakly turbulent range, and at worst $d_b^{\max} \lesssim 25$ km. Within a metropolitan setting, such constraints can be satisfied with realistic resources, supporting the development of wireless quantum networks. Furthermore, confidence in the use of free-space links within this setting reduces the need for wired fiber connections in small areas.

\section{Conclusion}
In this work we have investigated the end-to-end capacities of free-space and hybrid quantum networks, combining recently developed results in quantum information theory and well established theories of free-space optical communication. After collecting and reviewing these recent results, we introduced a modular network architecture for the purposes of constructing hybrid quantum networks using both free-space and fiber links. With these tools in hand, we specified our analysis to ideal modular networks which utilize an underlying regular backbone. Through this ideality it was possible to study ultimate limits for highly relevant modular architectures, revealing critical network properties that assure optimal performance. 

For the first time we have performed a detailed analysis of the ultimate limits of a satellite-based quantum internet; leveraging the properties of fiber-networks on the ground, ground-satellite connective structures and intersatellite networks in space. This theoretically demonstrates that high-rate global quantum communication can be efficiently mediated by a satellite quantum network with realistic connectivities and tolerable intersatellite separations on the order of ${\sim 10^{3} - 10^{4} \text{ km}}$. Such designs allow for effective quantum communication between arbitrarily distant end-users on the Earth. These analyses also indicate that careful consideration of the spatial distribution of ground-satellite connections can more effectively alleviate separation constraints, rather than increasing the nodal degree. 

Furthermore, we studied the ultimate limits of a free-space/fiber modular network configuration, discussing the efficacy of free-space sub-networks within metropolitan areas. We have shown that within the weakly-turbulent regime (where free-space links are limited to ${\sim 1 \text{ km}}$) high-rate intercommunity communication can be readily achieved, using a fiber-backbone with realistic resources. 

These results offer promising first steps in the direction of understanding the ultimate limits of free-space and hybrid quantum networks; motivating its future study both theoretically and experimentally. Our analyses offer a rigorous demonstration of the efficacy of free-space quantum links in a network setting, emphasizing that the integration of free-space and fiber can be reliably performed within future quantum networks. 
Hybrid architectures can and should be designed to take advantage of the strengths of different modes of quantum communication. This work may serve as a platform for future investigations that account for full technical details of the nodes; exploiting these tools to study more realistic, random architectures of hybrid networks which can be benchmarked against the ideal designs studied here.

\acknowledgements
C.H and A.I.F acknowledge funding from the EPSRC via a Doctoral Training Partnership (EP/R513386/1). S.P acknowledges funding by the European Union via “Continuous Variable Quantum Communications” (CiViQ, Grant Agreement No.~820466).


%

\newpage

\newpage
\newpage 
\bigskip
\begin{widetext}
\begin{minipage}{\textwidth}
\centering
\bigskip
\large 
\bf Supplementary Material: End-to-End Capacities of Hybrid Quantum Networks
\end{minipage}
\bigskip

\thispagestyle{empty}

\bigskip
In the main-text, we considered a specific modular network structure, using the idea of disjoint communities connected to a backbone quantum network. Here, using basic notions from graph theory and network theory \cite{SlepianNets,CoverThomas, TanenbaumNets, GamalNets}, we aim to generalize the concept of modular quantum networks, outlining a framework from which the ideal architecture in Definitions~1 and 3 emerge. In doing so, we derive general constraints which guarantee specific end-to-end performance bounds, for communication between local community users and remote community users. 

\setcounter{section}{0}
\setcounter{equation}{0}
\setcounter{figure}{0}

\section{General Aspects of Quantum Networks with Community Structure}

\subsection{General Structure}
Let us first consider general networks which display community structure. Consider a completely general architecture $\mc{N}=(P,E)$ such that $P$ is the collection of all nodes, and $E$ the set of all undirected edges. As discussed in the main text, it is possible to divide $P$ into sub-collections of communities,
\begin{equation}
P = \bigcup_{i} P_{c_{i}},~ P_{c_{i}} \subset P.
\end{equation}
In general, the community structure on a given network is not unique, and the sets of community nodes can overlap, i.e.~the subsets of nodes $P_{c_i}$ are not necessarily pairwise disjoint, i.e.~$P_{c_i} \cap P_{c_j} \neq \varnothing$, for all $i,j$. However, as we are physically motivated by separate communities connected via a backbone, we restrict our attention to the case in which each node can be uniquely assigned to a single community,
\begin{equation}
P = \bigcup_i P_{c_i}, \text{ s.t } P_{c_i} \cap P_{c_j} = \varnothing, \forall i,j.
\end{equation}

This assumption is appropriate for large-scale communication networks, and applies for spatially modular networks \cite{Gross2020}, e.g.~each community represents a separate metropolitan area. We make no further assumptions on the topology of the underlying communities. 

The community structure additionally partitions the edges into distinct sets. The $i^{\text{th}}$ community $c_i$ has its own set of intra-community edges,
\begin{equation}
E_{c_i}=\{(\bs{x}, \bs{y})\in E~|~{\bs{x},\bs{y} \in P_{c_{i}}} \},
\end{equation}
while any two communities $c_i$ and $c_j$ are connected by a set of intercommunity edges
\begin{equation}
E_{c_i:c_j}=\{(\bs{x}, \bs{y})\in E~|~ {\bs{x} \in P_{c_{i}}, \bs{y}\in P_{c_j}} \}.
\end{equation}
Hence, for a network comprised of $n$ communities we may define two global classes of edges: Intra-community edges and intercommunity edges respectively,
\begin{equation}
E_{c}\defeq\bigcup_{i=1}^{n} E_{c_i},~~E_{c:c^{\prime}} \defeq \hspace{-1mm}\bigcup_{ i\neq j =1}^n E_{c_i:c_j}.
\end{equation}
Using these notions we may introduce two related networks which will simplify our analysis. A \emph{community sub-network} $\mc{N}_{c_i}=(P_{c_i},E_{c_i})$ is defined as the graph consisting of all the nodes in the community $c_i$ connected by the intra-community edges $E_{c_i}$. 

\begin{figure}[t!]
\includegraphics[width=0.5\linewidth]{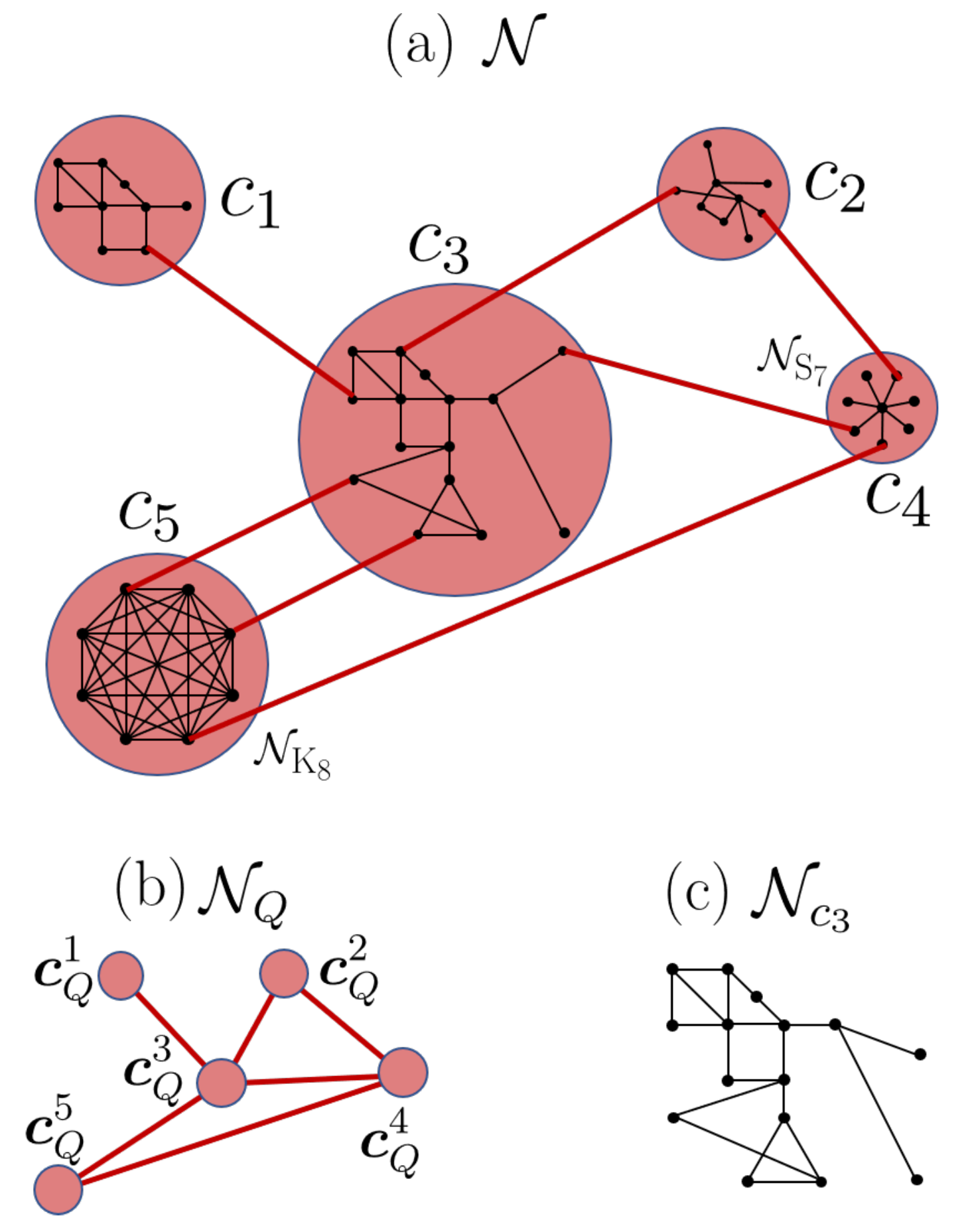}
\caption{(a) A network $\mc{N}$ displaying community structure. Each community is labelled and shown in red and red edges depict intercommunity channels. Black circles represent individual nodes and black edges represent intra-community channels. (b) The quotient network $\mc{N}_Q = \mc{N}/R$. All the nodes in the same community collapse to a community single node in the quotient network. The quotient network remains a simple graph even when multiple edges exist between communities. The single-edge capacity $C_{\bs{c}_3\bs{c}_5}$ on the quotient network is given by the sum of the capacities of the intercommunity edges between $\bs{c}_3$ and $\bs{c}_5$ on the original network, defined in Eq.~(\ref{eq:multi_intercom}). (c) The community sub-network for community $c_3$. Note also that the community sub-network $\mc{N}_{c_4} = \mc{N}_{\mathrm{S}_7}$ is a star network with seven child nodes, and $\mc{N}_{c_5} = \mc{N}_{\mathrm{K}_8}$ is a fully connected network with eight nodes.}
    \label{fig:general_modular}
\end{figure}

\subsection{Simplified Quotient Network}

Let $\mc{N}$ be a network with an $n$-community structure. Since we consider only non-overlapping communities, we may define an equivalence relation $R$ on the nodes of the network in the following way.
\begin{equation}
\label{eq: equiv relation}
\bs{x} \sim \bs{y} \ \mathrm{iff} \ \bs{x},\bs{y} \in P_{c_i}, \forall i \in [1,n].
\end{equation}
That is two nodes are equivalent if they are contained within the same community. This equivalence relation is a means partitioning the network into a simplified form, such that nodes contained within equivalent classes (communities) are pooled and redefined as a unified, collective node. 
Then, $R$ permits us to define a \emph{quotient network},
\begin{align}
\mc{N}_{Q} \defeq \mc{N}/R = (P_{Q},E_{Q}) ,
\end{align}
where $P_Q$ is a set of quotient nodes, and $E_{Q}$ is a set of quotient edges.
The set of quotient nodes is given by 
\begin{equation}
P_Q \defeq P / R = \{ \bs{c}_Q^i \}_{i=1}^n,\text{ where } \bs{c}_Q^i =  P_{c_i},
\end{equation}
where by the equivalence relation $R$ we have reduced the set of community nodes $P_{c_i}$ into a single quotient node $\bs{c}_Q^i$. 
Meanwhile, there exists a quotient edge between the two community nodes $\bs{c}_Q^i$ and $\bs{c}_Q^j$ if there exists at least one intercommunity edge between a node $\bs{x} \in P_{c_i}$ and a node $\bs{y} \in P_{c_j}$. Therefore the set of edges on the quotient network $E_{Q}$ is given by,
 \begin{equation}
E_{Q} \defeq E/R = \left\{ \Big( \bs{c}_Q^i , \bs{c}_Q^j \Big) |~\exists~(\bs{x},\bs{y}) \in E_{c_i:c_j} \right\}_{i\neq j=1}^n.
 \end{equation}
 
It is important to note that there may be more than one intercommunity edge between two given communities; yet our definition of the quotient network is still a simple graph. To account for this, the single-edge capacity of an edge in the quotient graph is actually defined as a multi-edge capacity from the original network. More precisely, the single-edge capacity of each quotient edge is equal to the sum of the capacities of the intercommunity edges,
 \begin{equation}
\{ \mc{C}_{\bs{xy}} \}_{(\bs{x},\bs{y})\in E_Q} =  \{ \mc{C}_{c_i : c_j}^m \}_{i\neq j =1}^n,
\end{equation}
where we have defined the multi-edge capacity between communities
\begin{equation}
\mc{C}_{c_{i}:c_{j}}^m \defeq \sum_{(\bs{x},\bs{y}) \in E_{c_i:c_j}} \hspace{-2mm} \mc{C}_{\bs{xy}}. \label{eq:multi_intercom}
 \end{equation}

These notions are depicted for a modular network in Fig~\ref{fig:general_modular}. This community structure is extremely useful for simplifying investigations of end-to-end capacities.
With this established, we can differentiate between two key scenarios for the end-to-end capacity: end-to-end communication in the same community, or between distinct communities.

\subsection{Intra-Community Capacities}
Let us focus on a pair of end-users $\bs{i} = \{\bs{\alpha},\bs{\beta}\}$ which are located within the same community, $\bs{\alpha},\bs{\beta} \in P_{c_i}$. While it may be intuitive to assume that the flooding capacity for communication between these nodes is determined by a min-cut performed exclusively on $\mc{N}_{c_i}$, this is not always the case. Indeed, it is possible that a minimum cut will collect edges not only within the community $\mc{N}_{c_i}$, but also intercommunity edges, and edges from other communities. In general, we can write the following lemma:

\begin{lemma}
Consider two end-user nodes $\bs{i} = \{\bs{\alpha},\bs{\beta}\}$ which are located within the same community $\mc{N}_{c_i} = (P_{c_i}, E_{c_i})$, such that $\bs{\alpha},\bs{\beta} \in P_{c_i}$. Let $\mc{C}_{c_i}^m$ be the end-to-end flooding capacity computed exclusively on the sub-network $\mc{N}_{c_i}$. Then the intra-community flooding capacity $\mc{C}^m(\bs{i}, \mc{N})$ is bounded by
\begin{align}
\mc{C}_{c_i}^m \leq \mc{C}^m(\bs{i}, \mc{N}) \leq \mc{C}_{c_i}^m + \mc{C}_{E \setminus c_i}^{m},
\end{align}
where $\mc{C}_{E \setminus c_i}^{m}$ is an additional capacity contribution associated with non-local community edges. 
\end{lemma}

\begin{proof}
Consider the two end-user nodes $\bs{i} = \{\bs{\alpha},\bs{\beta}\}$. We may exclusively investigate the flooding capacity of this induced sub-network, $\mc{N}_{c_i}$ by ignoring all intercommunity edges. In this way, we can identify a minimum cut restricted on the community by minimizing over all the local-community cuts,
\begin{equation}
\mc{C}_{c_i}^m = \min_{C_{c_i}}\> \mc{C}^{m}(C_{c_i}) = \min_{C_{c_i}}\hspace{-0mm}\sum_{(\bs{x},\bs{y})\in \tilde{C}_{c_i}} \mc{C}_{\bs{xy}}.
\end{equation}
Here a local-community cut $C_{c_i}$ is a cut performed exclusively on the community network, and $\tilde{C}_{c_i}$ is its cut-set.

Now consider the addition of intercommunity edges $\{ E_{c_i : c_j } \}_{j\neq i}$ which provide access to other remote communities. It is possible that these intercommunity edges will compromise the validity of a community cut $C_{c_i}$, since there may exist an end-to-end route that traverses the global network. In this scenario, it is necessary to cut additional edges from the rest of the network in order to consolidate the cut. We collect these additional edges within the following set $\tilde{C}_{E\setminus c_i} \subset E \setminus E_{c_i}$.
More precisely, given a valid network cut $C$, we can always separate its cut-set into community edges, and non-community edges
\begin{align}
\tilde{C} &= \{ (\bs{x},\bs{y}) \in E ~|~ \bs{x}\in \bff{A}, \bs{y} \in \bff{B} \}, \\
&= \tilde{C}_{c_i} \cup \tilde{C}_{E\setminus c_i}.
\end{align}
We can then say that the community cut-set $\tilde{C}_{c_i}$ is generated via a community cut $C_{c_i}$, while $ \tilde{C}_{E\setminus c_i}$ is generated via an additional non-community cut ${C}_{E\setminus c_i}$. 
The network flooding capacity is thus generally given by
\begin{align}
\mc{C}^m( \bs{i},\mc{N}) &= \min_{C} \mc{C}^{m}(C),\\
&=  \min_{C} \Big[ \mc{C}^{m}(C_{c_i}) + \mc{C}^{m}(C_{E \setminus c_i}) \Big]. \label{eq:IntraUB}
\end{align}

The cut ${C}_{c_i}$ always forms a valid partition of the user pair when we are restricted to the sub-network $\mc{N}_{c_i}$. 
Meanwhile, on its own, $\tilde{C}_{E \setminus c_i}$ is never a valid network cut between local end-users. Crucially, the addition of the non-community edges into the cut-set can never decrease the total flooding capacity between users, only increase it. Therefore we can separate the minimization in Eq.~(\ref{eq:IntraUB}) and write
\begin{align}
\mc{C}^m( \bs{i},\mc{N}) &=  \min_{C} \Big[ \mc{C}^{m}(C_{c_i}) + \mc{C}^{m}(C_{E \setminus c_i}) \Big],\\
&\leq  \min_{C_{c_i}} \mc{C}^{m}(C_{c_i}) + \min_{C_{E \setminus c_i}} \mc{C}^{m}(C_{E \setminus c_i}),\\
&= \mc{C}_{c_i}^m + \mc{C}_{E \setminus c_i}^{m},
\end{align}
where $\mc{C}_{E \setminus c_i}^{m}$ denotes the multi-edge capacity of the minimized non-community cut that validates the end-user partition. 

It is then clear that we can write the following bounds on the global network flooding capacity,
\begin{align}
\mc{C}_{c_i}^m \leq \mc{C}^m(\bs{i}, \mc{N}) \leq \mc{C}_{c_i}^m + \mc{C}_{{E \setminus c_i}}^{m}.
\end{align}
Here, the lower bound refers to the situation when non-community cuts are not required ($ \mc{C}_{{E \setminus c_i}}^{m}=0$), and the upper bound refers to when they are ($ \mc{C}_{{E \setminus c_i}}^{m} > 0$).
\end{proof}\\

Hence, the intra-community capacity is always lower-bounded by the local-community capacity of a local network. The saturation of either the upper or lower bounds is completely determined via the network structure. 

\subsection{interCommunity Capacities}
We now turn our attention to the case in which the two end-users lie in distinct communities. This is the setting focussed on in the main text, and is of most interest for (relatively) long-distance communication within large-scale, hybrid networks. Indeed, the intercommunity capacity depends more strongly on the interplay between sub-network properties, rendering its characterization more difficult than the intra-community capacity. 
Nonetheless, through the community structure developed in this appendix, and the simplifications offered by the quotient graph representation, it is possible to glean conditions for which the end-to-end intercommunity capacity is analytically obtainable. 

To achieve this, we must develop a number of helpful lemmas. We shall first show that any cut which collects an intra-community edge automatically invokes a valid cut between the two nodes connected by that edge on a community sub-network. 

\begin{lemma}
Consider two end-user nodes contained in remote communities $\bs{\alpha} \in P_{c_{\bs{\alpha}}}$ and $\bs{\beta} \in P_{c_{\bs{\beta}}}$, and a cut $C$ between them with a corresponding cut-set $\tilde{C}$. If $\tilde{C}$ contains at least one intra-community edge $(\bs{x},\bs{y})\in E_{c_{i}}$ from an arbitrary community $c_i$, then $\tilde{C}$ contains a subset $\tilde{C}'$ which is a valid cut between $\bs{x}$ and $\bs{y}$ on the induced sub-network $\mc{N}_{c_{i}}$.
\label{lemma:cut_lemma}
\end{lemma}

\begin{proof}
Consider a cut $C$ such that the corresponding cutset $\tilde{C}$ contains an intra-community edge ${(\bs{x},\bs{y}) \in E_{c_i}}$. Without loss of generality, the cut partitions the network nodes into two sets $\bff{A}=\{ \bs{\alpha}, \bs{x}, ... \}$ and $\bff{B}=\{ \bs{\beta}, \bs{y}, ... \}$. We can therefore identify two subsets $\bff{A}'\subseteq \bff{A}$ and $\bff{B}' \subseteq \bff{B}$ that consist solely of nodes that lie in the same community,
\begin{equation}
\bff{A}^{\prime}=  \{\bs{x}~|~\bs{x}\in\bff{A}\cap P_{c_i} \} ,~
\bff{B}^{\prime}= \{\bs{y}~|~\bs{y}\in\bff{B}\cap P_{c_i} \}.
\end{equation}
It can clearly be seen that this forms a bi-partition for the nodes in $P_{{c}_i}$ and thus forms a valid cut between the arbitrary nodes $\bs{x}$ and $\bs{y}$ on the community network $\mc{N}_{c_i}$. The corresponding cut-set $\tilde{C}'$ may be formed as usual from these sets,
\begin{equation}
\tilde{C}'=\{(\bs{x},\bs{y}) \in E~|\ \bs{x} \in \bff{A}^{\prime}, \ \bs{y} \in \bff{B}^{\prime}\},
\end{equation}
Comparing this to the original cut-set
\begin{equation}
\tilde{C}=\{(\bs{x},\bs{y}) \in E ~|\ \bs{x} \in \bff{A}, \ \bs{y} \in \bff{B}\},
\end{equation}
it can clearly be seen that $\tilde{C}' \subseteq \tilde{C}$ since $\bff{A}' \subseteq \bff{A}$ and $\bff{B}' \subseteq \bff{B}$.
\end{proof}\\

This is actually a very useful result. It tells us that a hybrid cut between two remote user-nodes $\bs{\alpha}$ and $\bs{\beta}$ which collects edges from a network community will necessarily invoke a local-community cut between some arbitrary pair of nodes. Let us now make the following definition which will simplify our notation.

\begin{defin} \emph{(Min-Local Community Capacity):} Consider a community sub-network given by $\mc{N}_{c_i}$. We define the minimum local-community capacity as the smallest flooding capacity that can be generated between any two nodes on community network,
\begin{equation}
\mc{C}_{c_i}^{*m} \defeq \underset{\bs{x} \neq \bs{y} \in P_{c_i} }{\min} ~ \mc{C}^m( \{ \bs{x},\bs{y}\}, \mc{N}_{c_i}).
\end{equation}
\end{defin}

As a result of the previous lemmas, we can present the following result which can be used to relate the intra-community capacity with the minimum cut on the quotient network.

\begin{lemma} Consider a quantum network $\mc{N}$ with a disjoint community structure, and a pair of remote end-users ${\bs{i} = \{ \bs{\alpha}, \bs{\beta}\} }$ which are located in distinct communities $\bs{\alpha} \in P_{c_{\bs{\alpha}}}$ and $\bs{\beta} \in P_{c_{\bs{\beta}}}$. On the quotient graph $\mc{N}_Q$, we can equivalently consider the end-user-community pair $\bs{i}_Q = \{ \bs{c}_Q^{\bs{\alpha}}, \bs{c}_Q^{\bs{\beta}}\}$. It follows that if all of the minimum local-community capacities are greater than the flooding capacity on the quotient network,
\begin{equation}
\underset{c_i}{\min}~ \mc{C}_{c_i}^{*m} \geq  \mc{C}^m(\bs{i}_Q, \mc{N}_Q),
\label{eq: saturation for intra caps}
\end{equation}
then the end-to-end flooding capacity between $\bs{\alpha}$ and $\bs{\beta}$ is equal to 
\begin{equation}
 \mc{C}^m(\bs{i}, \mc{N}) = \mc{C}^m(\bs{i}_Q, \mc{N}_Q).
 \end{equation}
Otherwise, the flooding capacity on the quotient network is an upper-bound on the true flooding capacity, $ \mc{C}^m(\bs{i}, \mc{N}) \leq \mc{C}^m(\bs{i}_Q, \mc{N}_Q)$.
\label{lemma:Quotient_Cap}
\end{lemma}

\begin{proof}
Since $\bs{\alpha}$ and $\bs{\beta}$ lie in two different communities it is always possible to form cuts with cut-sets only containing intercommunity edges. These are exactly the same cuts as are possible on the quotient graph $\mc{N}_Q = \mc{N}/R$ where $R$ is the equivalence relation partitioning the nodes into their communities. Hence, we can call these cuts \textit{quotient cuts}, $C_Q$. Therefore  we can obtain an initial bound for the multi-path capacity.
\begin{equation}
\mc{C}^m(\bs{i}, \mc{N})\leq\mc{C}^m(\bs{i}_Q, \mc{N}_Q)=
\min_{C} \ \mc{C}^m \left(\bs{i}_Q, C_{Q}\right),
\end{equation}
where $\mc{C}^m \left(\bs{i}_Q, C_{Q}\right)$ is the multi-edge capacity associated with a quotient cut partitioning the two communities in $\bs{i}_Q$. This is an upper bound since the cut taken on the quotient graph may not be a minimum cut. 

Now consider an arbitrary cut $C_0$ between $\bs{\alpha}$ and $\bs{\beta}$, containing at least one intra-community edge. From Lemma \ref{lemma:cut_lemma} we have that the intra-community edges form at least one valid cut between arbitrary nodes on an induced sub-network $\mc{N}_{c_i}$. Note that the corresponding cut-set will generally not correspond to a valid cut between $\bs{\alpha}$ and $\bs{\beta}$ on $\mc{N}$.
It is clear we can lower bound the capacity across $C_0$ by the minimum flooding capacity between any two nodes on $P_{c_i}$, that is
\begin{equation}
\mc{C}^m(\bs{i}, C_0)\geq \mc{C}^{*m}_{c_i}.
\end{equation}
Comparing this to the initial bound obtained on the quotient graph, we see that whenever the minimum flooding capacity between any two nodes on the community $\mc{N}_{c_i}$ satisfies
\begin{equation}
\mc{C}_{c_i}^{*m} \geq \mc{C}^m(\bs{i}_Q, \mc{N}_Q),
\end{equation}
then $C_0$ cannot be a minimum cut. Now since the intra-community edge that $C_0$ collects is arbitrary, the left hand side must be minimized over all communities to ensure that no hybrid cut can ever be a minimum cut. Therefore, whenever
\begin{equation}
\underset{c_i}{\min} \ \mc{C}_{c_i}^{*m} \geq \mc{C}^m (\bs{i}_Q, \mc{N}_Q),
\end{equation}
the minimum-cut must contain only intercommunity edges and $\mc{C}^m(\bs{i}, \mc{N})=\mc{C}^m (\bs{i}_Q, \mc{N}_Q)$.
\end{proof}\\

In general the condition given in Eq.~(\ref{eq: saturation for intra caps}) is fairly restrictive, as it places requirements on the minimum capacities between any two users in the same community. However, we shall see that in the case of communities connected to backbones, in which the quotient graph is simply a star network, the condition applies only to the two end-users' community networks $\mc{N}_{c_{\bs{\alpha}}}$ and $\mc{N}_{c_{\bs{\beta}}}$.

\subsection{Modular Networks with Backbone Structure}

We can now turn our discussion to the modular networks as defined in Definition 1 and 3 of the main text, which are specific architectures with community structure. These are modular networks where all of the communities are disjoint and disconnected, but are all connected to a municipal backbone network. By imposing regularity (and thus high connectivity) on the backbone, we are able to study ideal modular networks. 
It is clear to see that the quotient network of this kind of modular architecture produces a star network. Let us denote a star network with $m$-children nodes and a central node by $\mc{N}_{\text{S}_m}$. Each community becomes a child node of the central backbone node, and we gather a very simple network structure. This is illustrated in Fig.~\ref{fig:star_quotient}.

We find that when our modular network adopts this simple (yet very general) structure, then Lemma \ref{lemma:Quotient_Cap} also simplifies significantly. It can be shown that the conditions which require enforcing in Lemma \ref{lemma:Quotient_Cap} reduce to simple constraints only on community networks involved with the end-user pair; not on any other community sub-network. This result is captured in the following.

\begin{figure}[t!]
\includegraphics[width=\linewidth]{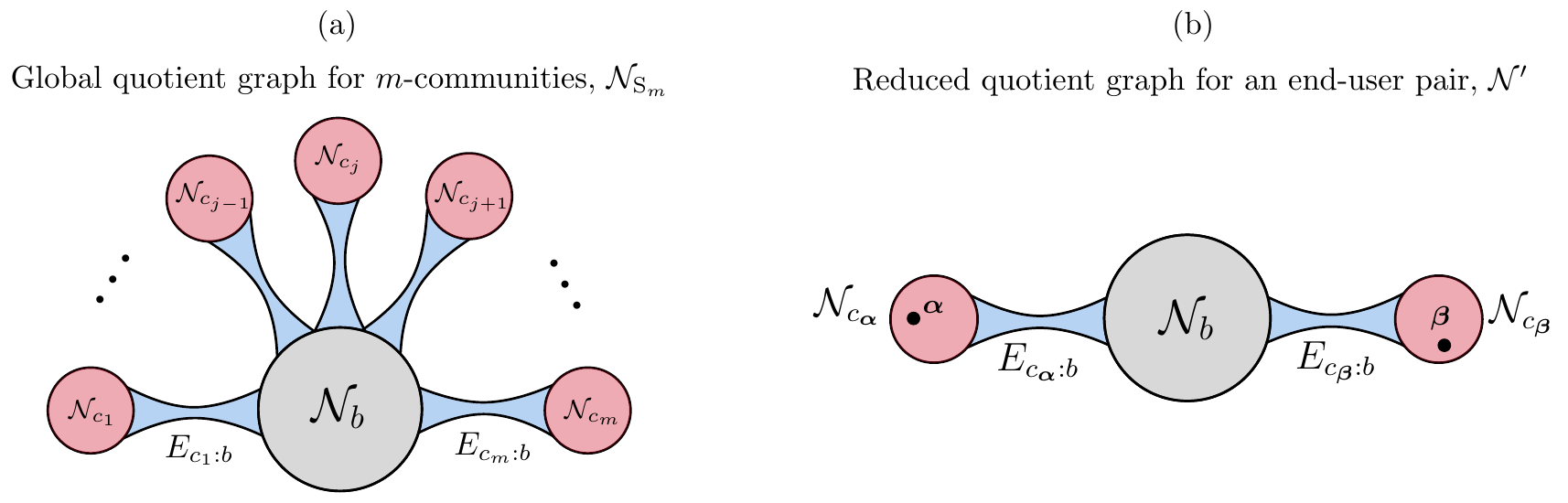}
\caption{(a) Quotient graph of a modular backbone network under the community equivalence relation for $m$-communities, resulting in a star network. (b) When considering the multi-path capacity between end-users $\bs{\alpha}$ and $\bs{\beta}$, the quotient graph can be simplified to a linear chain.}
\label{fig:star_quotient}
\end{figure}

\begin{lemma}
Consider a pair of end-users $\bs{i} = \{ \bs{\alpha}$,$\bs{\beta} \}$ and their associated pair of end-user communities ${\bs{i}_Q = \{ \bs{c}_Q^{\bs{\alpha}}, \bs{c}_Q^{\bs{\beta}} \}}$. The quotient graph $\mc{N}_{Q}$ of the network under the community equivalence relation $R$ is a star network. It then follows that if  
\begin{equation}
\underset{ x \in  \{ c_{\bs{\alpha}}, c_{\bs{\beta}}, b \} }{\min} \mc{C}_{x}^{*m} \geq \min\{ \mc{C}_{c_{\bs{\alpha}}:b}^m,\mc{C}_{c_{\bs{\beta}}:b}^m \}, \label{eq:Star_Cond}
\end{equation}
then the end-to-end flooding capacity between $\bs{\alpha}$ and $\bs{\beta}$ is equal to 
\begin{equation}
\mc{C}^m(\bs{i}, \mc{N}) = \min\{ \mc{C}_{c_{\bs{\alpha}}:b}^m,\mc{C}_{c_{\bs{\beta}}:b}^m \}.
\end{equation}
Otherwise, the flooding capacity on the quotient network is an upper-bound on the true flooding capacity 
$\mc{C}^m(\bs{i}, \mc{N}) \leq \min\{ \mc{C}_{c_{\bs{\alpha}}:b}^m,\mc{C}_{c_{\bs{\beta}}:b}^m \}$.
\label{lemma:GenThreshLemma}
\end{lemma}

\begin{proof}
It is known that we can always perform a valid cut by community isolation, i.e.~exclusively cutting the intercommunity edges between the backbone and either of the end-user communities. This type of cut equates to a cut on the quotient network, so that in general we can write the global-community capacity as an upper-bound on the flooding capacity
\begin{equation}
\mc{C}^m(\bs{i},\mc{N}) \leq \mc{C}^m(\bs{i}_Q,\mc{N}_Q) = \min\{ \mc{C}_{c_{\bs{\alpha}}:b}^m,\mc{C}_{c_{\bs{\beta}}:b}^m \}. \label{eq:star_cap_UB}
\end{equation}

Now let us impose the condition in Eq.~(\ref{eq:Star_Cond}). This condition is similar to that which is proven in Lemma~\ref{lemma:Quotient_Cap} for more general networks. However in this setting, it is not necessary to consider communities which don't contain end-users. When Eq.~(\ref{eq:Star_Cond}) holds, this means that any cut which collects an edge from the sub-networks $x \in  \{ c_{\bs{\alpha}}, c_{\bs{\beta}}, b \}$ will not be a minimum cut. More precisely, by Lemma~\ref{lemma:cut_lemma}, any cut which collects an edge from any of these sub-networks will automatically invoke a valid intra-community cut between a pair of arbitrary local nodes. But per Eq.~(\ref{eq:Star_Cond}), the minimum local-community capacity is always larger than the global-community capacity, therefore this form of cut will never be the minimum cut. 

We are now left to check that any cut which collects edges from other communities $c_i \notin  \{ c_{\bs{\alpha}}, c_{\bs{\beta}}, b \}$ will never be the minimum cut under these conditions. 
Consider a sub-graph of the original network $\mc{N}' = (P^{\prime},E^{\prime}) \subset \mc{N}$ which consists solely of the communities $c_{\bs{\alpha}}$,  $c_{\bs{\beta}}$ and $c_b$, each of the communities intra-community edges and the corresponding intercommunity edges. Therefore the sets of sub-graph nodes and edges are
\begin{align}
P' &= P_{c_{\bs{\alpha}}} \cup P_{c_{\bs{\beta}}} \cup P_{b},\\ 
E' &= ( E_{c_{\bs{\alpha}}} \cup E_{c_{\bs{\beta}}} \cup E_{b} ) \cup (E_{c_{\bs{\alpha}}:b}  \cup E_{c_{\bs{\beta}}:b}).
\end{align}
In general, the flooding capacity computed on the sub-network $\mc{N}^{\prime}$ will always be smaller than that computed on $\mc{N}$. The addition of extra communities can only ever increase the number of end-to-end multi-path routes. 
As a result, we can write the lower-bound
\begin{equation}
\mc{C}^m(\bs{i},\mc{N}) \geq \mc{C}^m(\bs{i},\mc{N}'). \label{eq:star_sub_lb}
\end{equation}
It is very important to note that $\mc{C}^m(\bs{i},\mc{N}')$ is not necessarily a valid, end-to-end capacity. This is because the minimum cut which generates $\mc{C}^m(\bs{i},\mc{N}')$ may not be a valid end-user partition on the global network $\mc{N}$. When the minimum cut which generates $\mc{C}^m(\bs{i},\mc{N}')$ is also a valid cut on $\mc{C}^m(\bs{i},\mc{N})$, then the above lower-bound saturates. 

The quotient network of the sub-graph $\mc{N}^{\prime}$ can then be reduced to a simple linear chain, as shown in Fig.~\ref{fig:star_quotient}(b). Now, thanks to the condition in Eq.~(\ref{eq:star_cap_UB}), we can equate the flooding capacity on $\mc{N}^{\prime}$ to that computed on its quotient network,
\begin{equation}
\mc{C}^m(\bs{i},\mc{N}^{\prime}) = \mc{C}^m(\bs{i}_Q,\mc{N}_Q^{\prime} ) =\min\{ \mc{C}_{c_{\bs{\alpha}}:b}^m,\mc{C}_{c_{\bs{\beta}}:b}^m \}. \label{eq:comm_iso}
\end{equation}
Crucially, the minimum cut which generates this capacity is a valid cut on the global network $\mc{N}$, since it is always possible to partition the end-users via community isolation. As a result, when we combine this lower-bound with the upper-bound in Eq.~(\ref{eq:star_cap_UB}), we gather that the end-to-end flooding capacity is given by
\begin{equation}
\mc{C}^m(\bs{i}, \mc{N}) = \min\{ \mc{C}_{c_{\bs{\alpha}}:b}^m,\mc{C}_{c_{\bs{\beta}}:b}^m \},
\end{equation}
as required. 
If the condition Eq.~(\ref{eq:Star_Cond}) is violated, we re-gather the upper-bound in Eq.~(\ref{eq:star_cap_UB}) since there may exist a cut that uses local community edges in $x \in  \{ c_{\bs{\alpha}}, c_{\bs{\beta}}, b \}$ to reduce the end-to-end capacity. This new cut \textit{may} also be a valid cut on $\mc{N}^{\prime}$, but it will not be achieved by community isolation, i.e.~the lower-bound in Eq.~(\ref{eq:star_sub_lb}) will still hold, but it will not be attributed to Eq.~(\ref{eq:comm_iso}).
\end{proof}\\

The technique used in the proof is actually rather more powerful than it may first appear. The key is to select a sub-graph whose quotient graph has exactly the same possible minimum cuts as the quotient graph of the overall network. When the condition in Eq.~(\ref{eq: saturation for intra caps}) on the sub-graph holds, this guarantees that the lower bound, found by asserting that the end-to-end capacity on the overall network must be greater than on the sub-graph,  can be saturated on the overall network. This in turn allows the lower bound to match the upper and reduce the restrictiveness of the condition given in Eq.~(\ref{eq: saturation for intra caps}) to just minimizing $\mc{C}^{*m}_{c_i}$ over the communities that exist in the sub-graph. This highlights that the degree of simplification provided by Lemma \ref{lemma:Quotient_Cap} depends on the underlying topology of the quotient network. Indeed it is clear that similar techniques can be applied to loosen the restrictions of Eq.~(\ref{eq: saturation for intra caps}) for other quotient network topologies, although we leave the exploration of these to future works.

\subsection{Threshold Capacities of Modular Networks with Backbone Structure}

Using the developments throughout this section, we can provide a concise proof of the main theorem in the text. This allows us to identify single-edge capacity thresholds for each of the end-user community networks and the backbone network, such that the end-to-end capacity is equal to the global-community capacity. As a result, we can identify unique physical constraints which can be used to motivate the construction of particular sub-networks, as was done in the main text. Here we restate the theorem for clarity:\\

\begin{theorem}
Consider an ideal modular network of the form $\mc{N}^*$ introduced in Definition~3. Select any pair of end-users $\bs{i}=\{\bs{\alpha},\bs{\beta}\}$ contained in remote communities $\bs{\alpha} \in P_{c_{\bs{\alpha}}}$ and $\bs{\beta} \in P_{c_{\bs{\beta}}}$. There exist single-edge threshold capacities on the communities $\mc{C}_{c_{\bs{j}}}^{\min}$ and backbone $\mc{C}_{b}^{\min}$ sub-networks for which the network flooding capacity is given by the global-community capacity,
\begin{align}
\begin{rcases}
\mc{C}_{\bs{xy}} \geq \mc{C}_{c_{\bs{j}}}^{\min}, \forall (\bs{x},\bs{y}) \in E_{c_{\bs{j}}},\\
{\mc{C}_{\bs{xy}} \geq \mc{C}_{b}^{\min}, \forall (\bs{x},\bs{y}) \in E_{b},}
\end{rcases}
\implies
\mc{C}^m( \mc{N}) = \mc{C}_{c : b}^m,
\end{align}
for all $\bs{j}\in\{\bs{\alpha},\bs{\beta}\}$. The threshold capacities are given by,
\begin{align}
 \mc{C}_{c_{\bs{j}}}^{\min} \defeq \frac{\mc{C}_{c : b}^m}{k_{c_{\bs{j}}} },~~
 \mc{C}_{b}^{\min} &\defeq \frac{\mc{C}_{c : b}^m}{ H_{\min}^* },
\end{align}
where $H_{\min}^*$ is the minimum cut-set cardinality on the backbone network. If these threshold capacities are violated, then the global-community capacity becomes an upper-bound on the end-to-end capacity, $\mc{C}^m( \mc{N}) \leq \mc{C}_{c : b}^m$. 
\label{theorem:Cond1}
\end{theorem}

\begin{proof}
Any modular network following the form of Definition~3 admits a star network as its quotient graph. In order to assert that the global-community capacity between any two end-users located in remote communities $\bs{i} =\{\bs{\alpha},\bs{\beta}\}$ is indeed the flooding capacity, we must reveal conditions for which all other possible cuts generate larger flooding capacities. Thanks to Lemma~\ref{lemma:GenThreshLemma} we know this condition is,
\begin{equation}
\underset{ c \in  \{ c_{\bs{\alpha}}, c_{\bs{\beta}}, b \} }{\min} \mc{C}_{c}^{*m} \geq \min\{ \mc{C}_{c_{\bs{\alpha}}:b}^m,\mc{C}_{c_{\bs{\beta}}:b}^m \} = \mc{C}_{c:b}^m, \label{eq:gen_cond}
\end{equation}
where $\mc{C}_{c:b}^m$ is the global community capacity which implicitly performs the minimization.
To satisfy the condition in Eq.~(\ref{eq:gen_cond}), it is sufficient to satisfy the set of equations,
\begin{equation}
 \mc{C}_{x}^{*m} \geq \mc{C}_{c:b}^m, \label{eq:spliced_cond} ~ \forall~x \in  \{ c_{\bs{\alpha}}, c_{\bs{\beta}}, b \}.
\end{equation}
Using this set of conditions, we are able to derive threshold capacities for each of the sub-networks of the modular structure to ensure the global community capacity is equal to the end-to-end capacity.

Let us first focus on satisfying this condition for the end-user communities, $c_{\bs{\alpha}}$ and $c_{\bs{\beta}}$. By definition, each of the communities in our idealized modular architecture adopt $k_{c_{\bs{j}}}$-connectivity, $\bs{j}\in\{\bs{\alpha},\bs{\beta}\}$. This means that the smallest possible cut between \textit{any two nodes} on either of the community networks (which contain end-users) collects exactly $k_{c_{\bs{j}}}$ edges. Let us also assume that there exists a single-edge threshold capacity for each community ${\mc{C}}_{c_{\bs{j}}}^{\min}$. Therefore, we can always say that the min-local community capacity $ \mc{C}_{c_{\bs{j}}}^{*m}$ will never be smaller than that which is generated by cutting $k_{c_{\bs{j}}}$ edges each of which have a minimum threshold capacity ${\mc{C}}_{c_{\bs{j}}}^{\min}$. 
That is, we can write
\begin{gather}
 \mc{C}_{c_{\bs{j}}}^{*m} \geq k_{c_{\bs{j}}} \mc{C}_{c_{\bs{j}}}^{\min} ,~\bs{j} \in \{\bs{\alpha},\bs{\beta}\},
\end{gather}
This lower-bound on the min-local capacity is achievable, since it is based on a valid cut on the communities. In order to satisfy Eq.~(\ref{eq:spliced_cond}) for each of communities, we must then demand
\begin{gather}
 \mc{C}_{c_{\bs{j}}}^{*m} \geq k_{c_{\bs{j}}} \mc{C}_{c_{\bs{j}}}^{\min}  \geq \mc{C}_{c:b}^m ,~\bs{j} \in \{\bs{\alpha},\bs{\beta}\}. 
\end{gather}
As a result, we derive a single-edge threshold capacity for edges within the local communities,
\begin{gather}
 \mc{C}_{c_{\bs{j}}}^{\min} = \frac{\mc{C}_{c:b}^m}{k_{c_{\bs{j}}}},~\bs{j} \in \{\bs{\alpha},\bs{\beta}\}.
\end{gather}
With this condition, we ensure that any valid cut performed exclusively on the local communities will always generate a larger multi-edge capacity than $\mc{C}_{c:b}^m$, and will not be the minimum cut.\\

We are now left to identify the single-edge constraint for the regular backbone network. While Eq.~(\ref{eq:gen_cond}) will supply a sufficient condition for the backbone network to ensure it does not compromise the global community capacity, the property of regularity lets us determine a more specific constraint. The backbone may possess many connections from the communities, meaning that the minimum number of edges in a cut-set performed exclusively on the backbone as potentially very large. The minimum cut-set size of a backbone cut depends totally on the network regularity, and the distribution of intercommunity connections, i.e.~the set of nodes $P_{b|c_{\bs{j}}}$ which tell us where the community $c_{\bs{j}}$ is directly connected to the backbone (for $\bs{j} \in \{\bs{\alpha},\bs{\beta}\}$). For regular networks, it is always possible to determine this minimum cut-set size via collective node isolation (see Section~\ref{sec:CommIso}). Hence, the minimum cardinality can be summarized by the function
\begin{equation}
| \tilde{C}_b | \geq H_{\min}^* \defeq \min_{\bs{j}\in\{\bs{\alpha},\bs{\beta}\}} H_{\min}(k_b, P_{b|c_{\bs{j}}})
\end{equation}
which chooses the minimum cut-set cardinality associated with either set of intercommunity connections. It then follows that, given some minimum single-edge capacity on the backbone network $\mc{C}_b^{\min}
$, the minimum possible multi-edge backbone capacity is given by 
\begin{equation}
\mc{C}_b^m \geq | \tilde{C}_b |\>\mc{C}_b^{\min}\geq   H_{\min}^*  \>\mc{C}_b^{\min}.
\end{equation}
In order the global-community capacity to remain a minimum cut, it must always be smaller than this lower-bound. Hence, we assert that,
\begin{equation}
\mc{C}_{c:b}^m \leq H_{\min}^*  \>\mc{C}_b^{\min}
\end{equation}
which leads to the required condition.
\end{proof}\\

It is important to note that these conditions hold for any end-user pair in remote communities; even when the user nodes possesses direct connections to the backbone. When this is the case, there will never exist a valid end-user cut that is exclusively made up of local community edges, since it is now necessary to also cut the direct connections to the backbone. 
Let $E_{\bs{x}} \defeq \{ (\bs{x},\bs{y})\in E~|~\bs{y}\in P\}$ be the set of all edges in the neighborhood of a node $\bs{x}$. We can identify the intercommunity edges which provide direct connections from a node $\bs{x}$ to the backbone via the edge set $E_{\bs{x}}\setminus E_{c_{\bs{x}}}$, i.e.~all the directly connected edges to $\bs{x}$ minus those which are community edges. Hence, we can never eliminate community-wide communication by means of a local community cut. If we impose the condition in Theorem~\ref{theorem:Cond1} \textit{anyway}, then this is sufficient to guarantee the global community capacity. Collecting $k_{c_{\bs{j}}}$ local community edges will automatically generate a multi-edge capacity which is at least as large as $\mc{C}_{c:b}^m$; hence the additional edges that one now needs to collect to consolidate the cut can only increase this multi-edge capacity. 

More precisely, the modification which minimizes the number of extra edges collected is achieved by additionally collecting the edges which connect the user-node directly to the backbone. We can denote the multi-edge capacity associated with cutting the user-connected intercommunity edges as
\begin{equation}
\mc{C}_{\bs{j}:b}^m \defeq \sum_{(\bs{x},\bs{y})\in E_{\bs{j}} \setminus E_{c_{\bs{j}}}} \mc{C}_{\bs{xy}}.
\end{equation}
The necessity of cutting additional intercommunity edges means that the min-local community capacity can \textit{never} be the flooding capacity; it can never be a valid minimum cut on its own, since the direct backbone connection means there will remain a route to the backbone (and thus to the other end-user). Instead, we perform a cut of $k_{c_{\bs{j}}}$ edges on the local community and cut these direct backbone connections $E_{\bs{x}}\setminus E_{c_{\bs{x}}}$. This results in a multi-edge capacity of 
\begin{align}
\mc{C}_{c_{\bs{j}}}^{*m} + \mc{C}_{\bs{j}:b}^m \geq  k_{c_{\bs{j}}} \>\mc{C}_{c_{\bs{j}}}^{\min} + \mc{C}_{\bs{j}:b}^m,~\forall \bs{j} \in \{ \bs{\alpha},\bs{\beta}\},
\end{align}
To ensure that this cut is never the minimum cut, we ask that 
\begin{align}
k_{c_{\bs{j}}} \>\mc{C}_{c_{\bs{j}}}^{\min} + \mc{C}_{\bs{j}:b}^m \geq \mc{C}_{c:b}^m, ~~\forall \bs{j} \in \{ \bs{\alpha},\bs{\beta}\},
\end{align}
is always true. Hence the necessity of cutting additional intercommunity edges leads to the modified condition on the local community threshold capacities,
\begin{align}
 k_{c_{\bs{j}}} \>\mc{C}_{c_{\bs{j}}}^{\min} \geq \mc{C}_{c:b}^m - \mc{C}_{\bs{j}:b}^m,~\forall \bs{j} \in \{ \bs{\alpha},\bs{\beta}\},
\end{align}
which is clearly a looser condition than that in Theorem~1. Therefore, Theorem~1 holds regardless of if the end-users are directly connected to the backbone or not.

\section{Application to Hybrid Quantum Networks}

With the main theorem from the main text now proven, it is possible to elucidate the emergence of Corollaries~1 and 2. These are simply applications of Theorem 1 in the context of fiber/satellite modular quantum networks, and ground-based free-space/fiber architectures. To assist the reader, we restate each corollary before providing their proofs.

\subsection{fiber/Satellite Configuration}

\begin{corollary}
Consider an ideal modular network of the form $\mc{N}^*$ introduced in Definition~3 in the main text, and assume optical-fiber communities networks $\mc{N}_{c_{\bs{\alpha}}}$, $\mc{N}_{c_{\bs{\beta}}}$ and a satellite-based backbone $\mc{N}_{b}$. Select any pair of end-users $\{\bs{\alpha},\bs{\beta}\}$ located in remote communities $\bs{\alpha} \in P_{c_\alpha}$ and $\bs{\beta} \in P_{c_\beta}$. 
There exists a maximum fiber-length in each community
\begin{equation}
d_{{{c_{\bs{j}}}}}^{\max} \defeq -\frac{1}{\gamma} \log_{10}\left(1 - 2^{-{\mc{C}_{c : b}^m}/{k_{c_{\bs{j}}} }}\right),
\end{equation}
and a maximum intersatellite separation in the backbone
\begin{equation}
z_{b}^{\max} \defeq \argmin_{z} \left| \log\left(\frac{H_{\min}^* \mc{B}_{F_{\sigma_{\text{\emph{p}}}}}(\eta) }{{\mc{C}_{c:b}^m}} \right) \right|. \label{eq:MaxIntSat}
\end{equation}
for which the network flooding capacity is equal to the global-community capacity, 
\begin{equation}
\mc{C}^m( \mc{N}) = \mc{C}_{c : b}^m.
\end{equation}

Otherwise, if any intersatellite links violate this condition ${\exists\> z_{\bs{xy}} > z_{{b}}^{\max}}, (\bs{x},\bs{y})\in E_b$ or the local community links are in violation, ${\exists \>d_{\bs{xy}} > d_{{{c_{\bs{j}}}}}^{\max}}, (\bs{x},\bs{y})\in E_{c_{\bs{j}}}$, for either $\bs{j} \in \{\bs{\alpha},\bs{\beta}\}$, then this becomes an upper-bound on the network flooding capacity, $\mc{C}^m( \mc{N}) \leq \mc{C}_{c : b}^m$. 
\label{corollary:Cor1}
\end{corollary}

\begin{proof}
The proof follows directly from the use of Theorem~1 and the direct substitution of single-edge capacity formulae into its results. As gathered from Theorem~1, for an ideal modular network of this form we can ensure that the end-to-end capacity between the end-users is equal to the global-community capacity if the following single-edge threshold capacities are satisfied: $\mc{C}_{c_{\bs{j}}}^{\min} \defeq {\mc{C}_{c : b}^m}/{k_{c_{\bs{j}}} }$ and $\mc{C}_{b}^{\min} \defeq {\mc{C}_{c : b}^m}/{ H_{\min}^* }$.

Since the community sub-networks are consistent of fiber channels, we can equate the single-edge community threshold capacity to the precise expression of a bosonic pure-loss channel capacity (the PLOB bound). For fiber-channels, a minimum capacity threshold corresponds to a maximum fiber-length threshold, such that 
\begin{equation}
{\mc{C}_{c_{\bs{j}}}^{\min} = -\log_{2}(1-10^{-\gamma d_{c_{\bs{j}}}^{\max}})}.
\end{equation}
This can be then be rearranged to determine the maximum permitted fiber-length within the community,
\begin{align}
d_{{{c_{\bs{j}}}}}^{\max} &= -\frac{1}{\gamma} \log_{10}\left(1 - 2^{-{\mc{C}_{c_{\bs{j}}}^{\min}}} \right) =-\frac{1}{\gamma} \log_{10}\left(1 - 2^{-{\mc{C}_{c : b}^m}/{k_{c_{\bs{j}}} }}\right).
\end{align}

We may perform a similar procedure for the backbone network, which is constructed from intersatellite channels. Assuming negligible thermal contributions (see Section~\ref{sec:InterSat}) but non-negligible pointing errors, then we can relate the threshold capacity $\mc{C}_{b}^{\min}$ to the single-edge capacity expression from Eq.~(43) in the main text,
\begin{equation}
\mc{C}_{b}^{\min} = \frac{\mc{C}_{c:b}^{m} }{H_{\min}^*} \leq \mc{B}_{F_{\sigma_{\text{p}}}} [\eta_{\text{d}}(z)] = \frac{2a_R^2 \Delta(\eta_{\text{d}}(z), \sigma_{\text{p}})}{w_{\text{d}}^2(z) \ln 2} .
\end{equation}
In general settings this is an upper-bound, as it is an extension of the PLOB bound to an ensemble of lossy channels where the convexity properties of the relative entropy of entanglement (REE) are exploited \cite{CondChanSim, FS,SQC}. However, we reliably assume the intersatellite channels to be modeled as pure-loss channels, and thus can admit equality $\mc{C}_{b}^{\min} = \mc{B}_{F_{\sigma_{\text{p}}}} [\eta_{\text{d}}(z)]$.
We are now in a position to compute the maximum tolerable intersatellite separation $z_b^{\max}$. This is the same as asking: For what channel length $z$ does the following equality hold
\begin{equation}
\mc{C}_{c:b}^{m} = H_{\min}^* \mc{B}_{F_{\sigma_{\text{p}}}} [\eta_{\text{d}}(z)]. \label{eq:solve_argmin}
\end{equation}
Due to the complicated nature of the capacity function $\mc{B}_{F_{\sigma_{\text{p}}}} [\eta_{\text{d}}(z)]$ this is not expedient analytically. However it is easy to compute numerically. Indeed, finding the maximum intersatellite separation equates to finding the minimum argument of
\begin{equation}
z_{b}^{\max} = \argmin_{z} \left| \log\left(\frac{H_{\min}^* \mc{B}_{F_{\sigma_{\text{\emph{p}}}}}(\eta) }{{\mc{C}_{c:b}^m}} \right) \right|.
\end{equation}
Here we use the absolute log-ratio to compare the right and left hand-side of Eq.~(\ref{eq:solve_argmin}) and determine for what channel length $z_{b}^{\max}$ they are equivalent. This provides a more sensitive measure than the absolute difference ${| {H_{\min}^* \mc{B}_{F_{\sigma_{\text{\emph{p}}}}}(\eta) } - {{\mc{C}_{c:b}^m}} |}$ since this can become very small at longer channel lengths, and is thus more suitable for determining the maximum intersatellite separation numerically.
 \end{proof}

\subsection{Ground-Based Free-Space fiber Configuration}

\begin{corollary}
Consider an ideal modular network of the form $\mc{N}^*$ introduced in Definition~3 in the main text, and assume free-space community networks $\mc{N}_{c_{\bs{\alpha}}}$, $\mc{N}_{c_{\bs{\beta}}}$ and an optical-fiber backbone $\mc{N}_{b}$. Select any pair of end-users $ \{\bs{\alpha},\bs{\beta}\}$ located in remote communities $\bs{\alpha} \in P_{c_\alpha}$ and $\bs{\beta} \in P_{c_\beta}$. 
There exists a maximum free-space link length in each community
\begin{equation}
z_{c_{\bs{j}}}^{\max} \leq \argmin_{z} \left| \log\left(\frac{k_{c_{\bs{j}}} \mc{L}_{F_{\sigma}}(\eta,\bar{n}_{\bs{j}}) }{{\mc{C}_{c:b}^m}} \right) \right|, \label{eq:GFSmax}
\end{equation}
and a maximum fiber length in the backbone
\begin{equation}
d_{b}^{\max} \defeq -\frac{1}{\gamma} \log_{10}\left(1 - 2^{-{\mc{C}_{c : b}^m}/{H_{\min}^* }}\right),
\end{equation}
for which the network flooding capacity is equal to the global-community capacity, \begin{equation}
\mc{C}^m( \mc{N}) = \mc{C}_{c : b}^m.
\end{equation}
Otherwise, if any fiber links violate this condition ${\exists\> d_{\bs{xy}} > d_{{b}}^{\max}}, (\bs{x},\bs{y})\in E_b$ or the local community links are in violation, ${\exists \>z_{\bs{xy}} > z_{{{c_{\bs{j}}}}}^{\max}}, (\bs{x},\bs{y})\in E_{c_{\bs{j}}}$, for either $\bs{j} \in \{\bs{\alpha},\bs{\beta}\}$, then this becomes an upper-bound on the network flooding capacity, $\mc{C}^m( \mc{N}) \leq \mc{C}_{c : b}^m$. 
\label{corollary:Cor2}
\end{corollary}

\begin{proof}
Once again, a proof follows directly from Theorem~1 and the techniques used to prove the previous Corollary. For an ideal modular network of this form we can ensure that the end-to-end capacity between the end-users is equal to the global-community capacity if the following single-edge threshold capacities are satisfied: $\mc{C}_{c_{\bs{j}}}^{\min}= {\mc{C}_{c : b}^m}/{k_{c_{\bs{j}}} }$ and $\mc{C}_{b}^{\min} = {\mc{C}_{c : b}^m}/{ H_{\min}^* }$.

In this setting, the community sub-networks are consistent of ground-based free-space quantum channels. We focus on the regime of weak turbulence, such that channel lengths are limited to ${z \lesssim 1{\text{ km}}}$. For a community containing an end-user $\bs{j}\in\{\bs{\alpha},\bs{\beta}\}$ we can write
\begin{equation}
\mc{C}_{c_{\bs{j}}}^{\min} = \frac{\mc{C}_{c:b}^{m} }{k_{c_{\bs{j}}}} \leq \mc{L}_{F_{\sigma}} [\eta(z), \bar{n}_{\bs{j}}],
\end{equation}
where $\mc{L}_{F_{\sigma}} [\eta(z), \bar{n}_{\bs{j}}]$ is the single-edge capacity upper-bound associated with a ground-based free-space link, discussed in Eq.~(32) and Section IID~1 of the main text. This incorporates atmospheric fading dynamics, and free-space background noise $\bar{n}_{\bs{j}}$ which may be present in the community $c_{\bs{j}}$. Hence, determining the maximum free-space link permitted in an end-user community is equivalent to finding the smallest channel length $z$ for which the equality 
\begin{equation}
{\mc{C}_{c:b}^{m} } = {k_{c_{\bs{j}}}}\mc{L}_{F_{\sigma}} [\eta(z), \bar{n}_{\bs{j}}],
\end{equation}
is satisfied. As before, this can be carried out numerically by finding the minimum argument
\begin{equation}
z_{c_{\bs{j}}}^{\max} \leq \argmin_{z} \left| \log\left(\frac{k_{c_{\bs{j}}} \mc{L}_{F_{\sigma}}(\eta,\bar{n}_{\bs{j}}) }{{\mc{C}_{c:b}^m}} \right) \right|.
\end{equation}
This is an upper-bound on $z_{c_{\bs{j}}}^{\max}$, since it is not known whether $\mc{L}_{F_{\sigma}}(\eta,\bar{n}_{\bs{j}})$ is an achievable rate or not. Nonetheless, this single-edge upper bound has been shown to be tight, and therefore we can accurately utilize it in order to gain insight into the reliability of free-space links in a metropolitan network setting.  

For the fiber-backbone, we possess exact expressions for single-edge capacities. Therefore, to find the maximum fiber-length we can simply compare the backbone threshold capacity to the PLOB bound and arrive at the result
\begin{align}
d_{b}^{\max} 
&=-\frac{1}{\gamma} \log_{10}\left(1 - 2^{-{\mc{C}_{c : b}^m}/{H_{\min}^*} }\right).
\end{align}
This completes the proof.
\end{proof}

\section{Collective Node Isolation \label{sec:CommIso}}

\subsection{Definition and Motivation}

Consider a network based on an underlying undirected graph $\mc{N} = (P,E)$, and some collection of $n$-network nodes $\bff{I} = \{ \bs{i}_1, \bs{i}_2, \ldots \bs{i}_n\} \subset P$ within it. Here, we will define $\bff{I}$ as a set of \textit{target-nodes} that we are interested in. We define the task of \textit{Collective Node Isolation} as that of determining the smallest cut-set of edges $\tilde{C}_{\min} \subset E$ that need to be removed from the network in order to form a sub-graph $\mc{N}_{\bff{I}} = (P_{\bff{I}}, E_{\bff{I}})$ within which all the target nodes are contained, i.e.~$\bff{I} \subseteq P_{\bff{I}}$. Importantly, this sub-graph need not be exclusively consistent of target nodes, but can also possess additional nodes. This question is relevant as it emerges within an unweighted minimum-cut problem for distant collections of nodes on a highly-connected network. That is, given some disjoint collections of sender nodes $\bff{A}$ and receiver nodes $\bff{B}$, what is the minimum cut-set cardinality required to partition these collections of end-users?

Clearly, for a completely general network it is by no means obvious what this cut-set is. However, by asserting some form of connectivity constraints it is possible to gain some useful analytical insight. In particular we are interested in $k$-regular networks, relevant for the regular backbone networks studied within the main-text. The high level of connectivity guaranteed by regularity ensures that for given a pair of individual end-user nodes, the cut-set with the smallest cardinality (neglecting boundary effects) will always be found via nodal isolation. This is because regularity guarantees a high growth rate for cut-set sizes as one moves further away from either end-user; hence the closer one remains to either end-user, the smaller the cut-set will be. Regular networks with this property are defined as super-connected. 

Hence, collective node isolation can be used to identify minimum-cut set sizes on a super-connected graph when it is necessary to isolate a number of particular nodes, $\bff{I}$. This is a generalization of the work in Ref.~\cite{OPGQN} in which the focus is nodal isolation on weakly-regular graphs. In a modular network setting, collective node isolation is important for identifying minimum cut-set cardinalities when restricted to a particular sub-network of the global model. This is made clear via its application in the main text. 
In the following, we devise the general result for the cut-set size of collective node isolation on regular networks. 

\subsection{Minimum Cut-Set Cardinality}

Consider a $k$-regular network $\mc{N} = (P,E)$ and two specific disjoint collections of target nodes labelled $\{\bff{A},\bff{B}\} \subset P$. These collections are used to represented end-user connected nodes on an intermediate sub-network within a modular structure. We wish to derive an expression for the minimum-cut set size required to completely partition the collections.

Let us define two Kronecker-delta like functions which are useful in this context. First, we define a function which specifies whether a generic node $\bs{x} \in P$ is actually a target node from one of the collections,
\begin{equation}
\delta_{\bff{I}}(\bs{x}) \defeq
\begin{cases}
1, & \text{ if } \bs{x} \in \bff{I},\\
0, & \text{ otherwise},
\end{cases}
\text{ where } \bff{I} \in \{\bff{A},\bff{B}\}.
\end{equation}
Then, we define the following neighbor-sharing counting function. For a given non-target node $\bs{x} \in P\setminus \{\bff{A},\bff{B}\}$, this counts the total number of connections that $\bs{x}$ has to target-nodes,
\begin{equation}
F_{\bff{I}} (\bs{x}) = \sum_{\bs{y} \in N_{\bs{x}}} \delta_{\bff{I}}(\bs{y})
\end{equation}
where $N_{\bs{x}} \defeq \{ \bs{y} \in P~|~(\bs{x},\bs{y}) \in E \}$ defines the neighborhood nodes of the node $\bs{x}$.

Consider either collection of target-connected nodes, $\bff{I}\in\{\bff{A},\bff{B}\}$. A regular network is super-connected, hence the minimum cut-set cardinality is always achieved by neighborhood isolation. In this way, the largest set generated by neighborhood isolation occurs when all the users are sufficiently separated so that they do not share any edges or any neighbors. Then the cut-set has cardinality $|\tilde{C}| = k |\bff{I}|$. This is always an upper-bound on the minimum cut-set size. However, the potential for target-nodes sharing edges and sharing neighbors can diminish this cut-set size, since redundant edges may emerge. Therefore, we can introduce corrective terms which remove redundant edges from the cut-set (dependent on the distribution of target-nodes).

The first corrective factor removes all copies of edges that are directly shared between target-nodes, since they do not facilitate information flow outside of the partition,
\begin{equation}
\mc{S}_{\text{E}}(\bff{I}) = \sum_{\bs{i}\in\bff{I}; \bs{x}\in N_{\bs{i}}} \delta_{\bff{I}}(\bs{x}).
\end{equation}
The second corrective term accounts for the effect of non-target-nodes which are connected to multiple target-nodes. Let us define the set of all non-user nodes in the network as $P^{\prime} = P\setminus\{\bff{A},\bff{B}\}$. Then, we further define the set of non-target nodes that are also neighbors of target-nodes, given by
\begin{equation}
P_{\bff{I}}^{\prime} \defeq  P^{\prime} \bigcap_{\bs{i}\in \bff{I}} N_{\bs{i}}.
\end{equation}
When non-target nodes are connected to multiple targets, sometimes it is better to cut its external edges rather than the internal edges connected to the targets. This is because there may be less external edges which when removed are still capable of partitioning the target nodes (this detail is illustrated in Fig.~2 in the main-text). Overall, the correction is found by iterating over all the nodes in $\bs{x} \in P_{\bff{I}}^{\prime}$ and deciding whether a superior cut can be found,
\begin{equation}
\mc{S}_{\text{N}}(k,\bff{I}) = \sum_{\bs{x}\in P_{\bff{I}}^{\prime}} \max\{ 0, 2F_{\bff{I}} (\bs{x}) -k\}.
\end{equation}
Piecing these corrections together, we arrive at a completely general function that computes the minimum cut-set size for a target-node distribution $\bff{I}$ on a regular network,
\begin{gather}
 H_{\min}(k,\bff{I}) \defeq k|\bff{I}| -\mc{S}_{\text{E}}(\bff{I}) - \mc{S}_{\text{N}}(k,\bff{I}) . \label{eq:MinCutH}
\end{gather}

Hence, in the context of a minimum cut between two collections of target-nodes $\bff{A}, \bff{B}$, we can then simply choose the set which minimizes the cut-set size, $\min_{\bff{I}\in\{\bff{A},\bff{B}\}}  H_{\min}(k,\bff{I})$. 

\subsection{Weakly-Regular Neighborhood Isolation}
We can use the example of weakly-regular networks from Ref.~\cite{OPGQN} to show the generality of the previous expression. In that work, it was important to determine the minimum cut-set cardinality that could be achieved when one is not permitted to cut neighborhood edges of some potential end-users, $\bs{\alpha}$ and $\bs{\beta}$. In doing so, it was possible to derive conditions on the weakly-regular network for which the flooding capacity was always the minimum neighborhood capacity between the users. 
Interestingly, this is equivalent to asking: What is the minimum cut-set size related to collectively isolating a user neighborhood $\bff{A} = N_{\bs{\alpha}}$ or $\bff{B} = N_{\bs{\beta}}$ on a weakly regular network? We can show that the result in Eq.~(\ref{eq:MinCutH}) can reproduce the result found from this investigation. 

A weakly-regular network is a network architecture based on an undirected graph $\mc{N}_{\text{WR}}=(P,E)$ which has the specific connectivity properties. In a $(k,\bs{\Lambda})$-weakly-regular network, for any node $\bs{x}$, there is a multiset of values $\bs{\lambda}_{\bs{x}}$ which collects  the number of of common neighbors shared between $\bs{x}$ and each $\bs{y} \in N_{\bs{x}}$. That is,
\begin{equation}
\bs{\lambda}_{\bs{x}} \defeq \{ |N_{\bs{x}} \cap N_{\bs{y}}| ~|~\bs{y}\in N_{\bs{x}} \}.
\end{equation}
This is known as the \textit{adjacent commonality multiset}
A network is $(k,\bs{\Lambda})$-weakly-regular if each node is connected to exactly $k$ other nodes, and each adjacent commonality multiset belongs to the superset $\bs{\Lambda}$ such that $\bs{\lambda}_{\bs{x}} \in \bs{\Lambda}$ for all $\bs{x}\in P$. For more details on these kinds of network, please see Ref.~\cite{OPGQN}.

 
 Here, we focus on a scenario in which there is only one non-degenerate adjacent commonality multiset, i.e.~$\bs{\Lambda} = \{ \bs{\lambda} \}$, with the implicit understand that this can be extended. Hence, each node has $k$-neighbors, and the distribution of adjacent commonalities always follows $\bs{\lambda} = \{ \lambda_1,\ldots,\lambda_k\}$. As a result, every node $\bs{i}$ has $k$-neighbors, and shares a unique number of common neighbors with each of them $\lambda_j \in \bs{\lambda}$. Consider performing collective isolation of the $k$-element neighborhood of some node $\bs{i}$ in this network. The maximum cut-set size is of course 
\begin{equation}
k |\bff{I}| = k| N_{\bs{i}}| = k^2,
\end{equation}
but this must be reduced due to edge-sharing and neighbor-sharing corrective factors. Indeed, the $j^{\text{th}}$ neighbor is connected to exactly $\lambda_j$ other neighbors of $\bs{i}$, leading to the edge-sharing correction
\begin{equation}
\mc{S}_E ( N_{\bs{i}} ) = \sum_{j=1}^k \lambda_j.
\end{equation}
Meanwhile, clearly the entire neighborhood shares a single non-target node (the original user node $\bs{i}$), meaning that $P_{\bff{I}}^{\prime} = \{\bs{i} \}$. This leads to the neighbor-sharing correction
\begin{equation}
\mc{S}_{\text{N}}(k,N_{\bs{i}}) = \max\{ 0, 2F_{\bff{I}} -k\}  = k.
\end{equation}
As a result, the cut-set size is given by,
\begin{align}
 H_{\min}(k,N_{\bs{i}}) &= k|N_{\bs{i}}| - \sum_{j=1}^k \lambda_j - k = \sum_{j=1}^{k} (k-\lambda_j -1). 
\end{align}
This is the result reported in Ref.~\cite{OPGQN} where it was derived in a more direct fashion.

\subsection{Bounds on Collective Node Isolation for Backbone Cuts}

The function $ H_{\min}(k,\bff{I})$ can be used for any general distribution of target-nodes. However, in the modular setting it is very easy to write bounds on the minimum cut-set size for network cuts performed exclusively on the backbone. Given a $k$-regular backbone network, we can write the following bounds
\begin{equation}
k \leq   H_{\min}(k,\bff{I}) \leq k |\bff{I}|.
\end{equation}
The lower bound corresponds to a situation where all the intercommunity edges are connected to a single node on the backbone. In this case it is sufficient to simply isolate the connected node on the backbone. The upper-bound refers to a situation where the intercommunity edges are connected to the backbone in such a way that the target-nodes do not share any edges or neighbors. The cut-set size for all other distributions $\bff{I}$ fall within these bounds. Examples are illustrated for a Manhattan backbone network in Fig.~2 in the main text\\

\section{Considerations for interSatellite Networks \label{sec:InterSat}}

\subsection{Line-of-Sight Distance Limits for interSatellite Channels}
Consider two satellites in circular orbits arounds the Earth at altitudes. Satellites can only communicate with one another of they are within each other's line of sight. Hence, it's easy to identify a geometric upper bound on the intersatellite separation between two satellites. 

First, consider the two satellites at positions $A$ and $B$ to be at equivalent altitudes $h$. By drawing a chord AB tangential to the Earth's surface from one satellite to the other, we can identify the maximum line-of-sight separation, $z_{\text{sight}}^{\max}$ (see Fig.~\ref{fig:GeoE} for a geometrical insight). Label the centre of the Earth E, its radius $R_{\text{E}}$ and the point at which the chord touches the Earth's surface $S$. Denoting the angle $\angle{EAS}$ as $\alpha$, this will satisfy
\begin{equation}
\sin \alpha = \frac{R_{\text{E}}}{R_{\text{E}} + h}
\end{equation}
With this in hand, we find that the distance $AB = 2 AS$, since the triangle EBA is clearly isosceles. That is,
\begin{align}
z_{\text{sight}}^{\max} &= 2 (h + R_{\text{E}}) \cos \alpha 
= \frac{2h(h + 2R_{\text{E}}) }{h + R_{\text{E}} }.
\end{align}
For two satellites which are at different altitudes $h_1 \leq h_2$, this maximum distance is extended to
\begin{align}
z_{\text{sight}}^{\max} &= 
\frac{h_1(h_1 + 2R_{\text{E}}) }{h_1 + R_{\text{E}} } + \frac{h_2(h_2 + 2R_{\text{E}}) }{h_2 + R_{\text{E}} },
\end{align}
which follows intuitively from the previous geometrical considerations.

\begin{figure}
\includegraphics[width=0.6\linewidth]{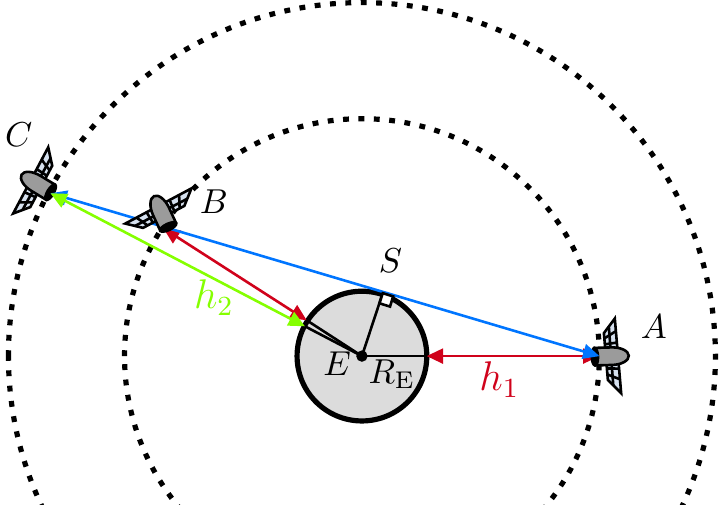}
\caption{The longest possible intersatellite quantum channel is limited by the line-of-sight separation of two satellites. }
\label{fig:GeoE}
\end{figure}

\subsection{Analytical Bounds for the Maximum interSatellite Separation}

It is possible to analytically upper and lower bound the maximum intersatellite separation within the backbone $z_{b}^{\max}$ from Corollary 1 in the main text. An upper-bound is found by considering a lack of pointing errors, which means the channel is no longer a fading channel but is instead the a fixed lossy channel with the maximum possible transmissivity. Therefore we can always write the upper-bound,
\begin{align}
z_{b}^{\max}  &\leq z_R \sqrt{\frac{2 a_R^2}{w_0^2 \ln\left[ \frac{\eta_{\text{{eff}}}}{\eta_{\text{{eff}}} -1 + 2^{-{ \mc{C}_{c : b}^m}/{ H_{\min}^* }}} \right]} - 1} .
\end{align}
This is easily derived using the ideal pure-loss single-edge capacity upper-bound from Theorem 1.

Meanwhile, we can find a lower-bound on the maximum intersatellite separation by considering the use of slow detectors. A slow detector at the receiver will not be able to resolve pointing errors, resulting in a lossy channel with fixed transmissivity averaged over the entire fading process. For intersatellite channels this is considered through the long-term spot size, by replacing the ideal diffraction limited spot size $w_d^2$ with $w_{\text{lt}}^2 = w_d^2 + \sigma_{\text{p}}^2$ into the capacity formula,
\begin{equation}
\mc{C} \leq \mc{B}_{\text{slow}}(\eta_{\text{lt}}) = \frac{2a_R^2}{w_{\text{lt}}^2 \text{ln} 2}.
\end{equation}
Interestingly, the rate in bits \textit{per channel use} via slow detection can be higher than that for fast detectors which actually resolve the fading dynamics. But do not be mistaken; the slower detection time severely limits the \text{operational rate} at which the channel can actually be used (or clock rate). As a result, the point-to-point communication rate via slow detection will be orders of magnitude smaller than those with fading-resolving setups.
Therefore, it is essential to explicitly consider the clock rate $\alpha$ (channel uses/second) when comparing fast and slow detector protocols \cite{FS,SQC}. In any case, the maximum intersatellite separation will be lower-bounded by
\begin{align}
z_{b}^{\max} &\geq \sqrt{\frac{2 a_R^2}{\big(\frac{w_0^2}{z_R^2} +\epsilon_{\text{p}}^2\big) \ln\left[ \frac{\eta_{\text{eff}}}{\eta_{\text{eff}} - 1 + 2^{-{ \tilde{\alpha} \mc{C}_{c : b}^m}/{  H_{\min}^* }}} \right]} - w_0^2},
\end{align}
where $\epsilon_{\text{p}} = 10^{-6}$ comes from the point error variance $\sigma_{\text{p}}^2 = (\epsilon_{\text{p}} z)^2$, and $\tilde{\alpha} = \alpha_{{c:b}}/ \alpha_{{b}}$ is the ratio between the clocks used by the intercommunity sub-network and the backbone network.


\begin{thebibliography}{52}%
\makeatletter
\providecommand \@ifxundefined [1]{%
 \@ifx{#1\undefined}
}%
\providecommand \@ifnum [1]{%
 \ifnum #1\expandafter \@firstoftwo
 \else \expandafter \@secondoftwo
 \fi
}%
\providecommand \@ifx [1]{%
 \ifx #1\expandafter \@firstoftwo
 \else \expandafter \@secondoftwo
 \fi
}%
\providecommand \natexlab [1]{#1}%
\providecommand \enquote  [1]{``#1''}%
\providecommand \bibnamefont  [1]{#1}%
\providecommand \bibfnamefont [1]{#1}%
\providecommand \citenamefont [1]{#1}%
\providecommand \href@noop [0]{\@secondoftwo}%
\providecommand \href [0]{\begingroup \@sanitize@url \@href}%
\providecommand \@href[1]{\@@startlink{#1}\@@href}%
\providecommand \@@href[1]{\endgroup#1\@@endlink}%
\providecommand \@sanitize@url [0]{\catcode `\\12\catcode `\$12\catcode
  `\&12\catcode `\#12\catcode `\^12\catcode `\_12\catcode `\%12\relax}%
\providecommand \@@startlink[1]{}%
\providecommand \@@endlink[0]{}%
\providecommand \url  [0]{\begingroup\@sanitize@url \@url }%
\providecommand \@url [1]{\endgroup\@href {#1}{\urlprefix }}%
\providecommand \urlprefix  [0]{URL }%
\providecommand \Eprint [0]{\href }%
\providecommand \doibase [0]{https://doi.org/}%
\providecommand \selectlanguage [0]{\@gobble}%
\providecommand \bibinfo  [0]{\@secondoftwo}%
\providecommand \bibfield  [0]{\@secondoftwo}%
\providecommand \translation [1]{[#1]}%
\providecommand \BibitemOpen [0]{}%
\providecommand \bibitemStop [0]{}%
\providecommand \bibitemNoStop [0]{.\EOS\space}%
\providecommand \EOS [0]{\spacefactor3000\relax}%
\providecommand \BibitemShut  [1]{\csname bibitem#1\endcsname}%
\let\auto@bib@innerbib\@empty
\bibitem [{\citenamefont {Slepian}(1968)}]{SlepianNets}%
  \BibitemOpen
  \bibfield  {author} {\bibinfo {author} {\bibfnamefont {P.}~\bibnamefont
  {Slepian}},\ }\href@noop {} {\emph {\bibinfo {title} {Mathematical
  Foundations of Network Analysis}}}\ (\bibinfo  {publisher}
  {Springer-Verlag},\ \bibinfo {address} {New York},\ \bibinfo {year}
  {1968})\BibitemShut {NoStop}%
\bibitem [{\citenamefont {Cover}\ and\ \citenamefont
  {Thomas}(2006)}]{CoverThomas}%
  \BibitemOpen
  \bibfield  {author} {\bibinfo {author} {\bibfnamefont {T.~M.}\ \bibnamefont
  {Cover}}\ and\ \bibinfo {author} {\bibfnamefont {J.~A.}\ \bibnamefont
  {Thomas}},\ }\href@noop {} {\emph {\bibinfo {title} {Elements of Information
  Theory}}}\ (\bibinfo  {publisher} {Wiley},\ \bibinfo {address} {New Jersey},\
  \bibinfo {year} {2006})\BibitemShut {NoStop}%
\bibitem [{\citenamefont {Tanenbaum}\ and\ \citenamefont
  {Wetherall}(2010)}]{TanenbaumNets}%
  \BibitemOpen
  \bibfield  {author} {\bibinfo {author} {\bibfnamefont {A.~S.}\ \bibnamefont
  {Tanenbaum}}\ and\ \bibinfo {author} {\bibfnamefont {D.~J.}\ \bibnamefont
  {Wetherall}},\ }\href@noop {} {\emph {\bibinfo {title} {Computer
  Networks}}},\ \bibinfo {edition} {5th}\ ed.\ (\bibinfo  {publisher}
  {Pearson},\ \bibinfo {year} {2010})\BibitemShut {NoStop}%
\bibitem [{\citenamefont {El~Gamal}\ and\ \citenamefont
  {Kim}(2011)}]{GamalNets}%
  \BibitemOpen
  \bibfield  {author} {\bibinfo {author} {\bibfnamefont {A.}~\bibnamefont
  {El~Gamal}}\ and\ \bibinfo {author} {\bibfnamefont {Y.-H.}\ \bibnamefont
  {Kim}},\ }\href@noop {} {\emph {\bibinfo {title} {Network Information
  Theory}}}\ (\bibinfo  {publisher} {Cambridge University Press},\ \bibinfo
  {year} {2011})\BibitemShut {NoStop}%
\bibitem [{\citenamefont {Kimble}(2008)}]{KimbleQI}%
  \BibitemOpen
  \bibfield  {author} {\bibinfo {author} {\bibfnamefont {H.~J.}\ \bibnamefont
  {Kimble}},\ }\bibfield  {title} {\bibinfo {title} {The quantum internet},\
  }\href {https://doi.org/10.1038/nature07127} {\bibfield  {journal} {\bibinfo
  {journal} {Nature}\ }\textbf {\bibinfo {volume} {453}},\ \bibinfo {pages}
  {1023} (\bibinfo {year} {2008})}\BibitemShut {NoStop}%
\bibitem [{\citenamefont {Pirandola}\ and\ \citenamefont
  {Braunstein}(2016)}]{UniteQInt}%
  \BibitemOpen
  \bibfield  {author} {\bibinfo {author} {\bibfnamefont {S.}~\bibnamefont
  {Pirandola}}\ and\ \bibinfo {author} {\bibfnamefont {S.~L.}\ \bibnamefont
  {Braunstein}},\ }\bibfield  {title} {\bibinfo {title} {Physics: Unite to
  build a quantum internet},\ }\href {https://doi.org/10.1038/532169a}
  {\bibfield  {journal} {\bibinfo  {journal} {Nature}\ }\textbf {\bibinfo
  {volume} {532}},\ \bibinfo {pages} {169} (\bibinfo {year}
  {2016})}\BibitemShut {NoStop}%
\bibitem [{\citenamefont {Razavi}(2018)}]{RazaviQNet}%
  \BibitemOpen
  \bibfield  {author} {\bibinfo {author} {\bibfnamefont {M.}~\bibnamefont
  {Razavi}},\ }\href {http://dx.doi.org/10.1088/978-1-6817-4653-1} {\emph
  {\bibinfo {title} {An Introduction to Quantum Communications Networks}}},\
  2053-2571\ (\bibinfo  {publisher} {Morgan \& Claypool Publishers},\ \bibinfo
  {year} {2018})\BibitemShut {NoStop}%
\bibitem [{\citenamefont {Pirandola}\ \emph {et~al.}(2020)\citenamefont
  {Pirandola}, \citenamefont {Andersen}, \citenamefont {Banchi}, \citenamefont
  {Berta}, \citenamefont {Bunandar}, \citenamefont {Colbeck}, \citenamefont
  {Englund}, \citenamefont {Gehring}, \citenamefont {Lupo}, \citenamefont
  {Ottaviani},\ and\ \citenamefont {et~al.}}]{AdvCrypt}%
  \BibitemOpen
  \bibfield  {author} {\bibinfo {author} {\bibfnamefont {S.}~\bibnamefont
  {Pirandola}}, \bibinfo {author} {\bibfnamefont {U.~L.}\ \bibnamefont
  {Andersen}}, \bibinfo {author} {\bibfnamefont {L.}~\bibnamefont {Banchi}},
  \bibinfo {author} {\bibfnamefont {M.}~\bibnamefont {Berta}}, \bibinfo
  {author} {\bibfnamefont {D.}~\bibnamefont {Bunandar}}, \bibinfo {author}
  {\bibfnamefont {R.}~\bibnamefont {Colbeck}}, \bibinfo {author} {\bibfnamefont
  {D.}~\bibnamefont {Englund}}, \bibinfo {author} {\bibfnamefont
  {T.}~\bibnamefont {Gehring}}, \bibinfo {author} {\bibfnamefont
  {C.}~\bibnamefont {Lupo}}, \bibinfo {author} {\bibfnamefont {C.}~\bibnamefont
  {Ottaviani}},\ and\ \bibinfo {author} {\bibnamefont {et~al.}},\ }\bibfield
  {title} {\bibinfo {title} {Advances in quantum cryptography},\ }\href
  {http://dx.doi.org/10.1364/AOP.361502} {\bibfield  {journal} {\bibinfo
  {journal} {Advances in Optics and Photonics}\ }\textbf {\bibinfo {volume}
  {12}},\ \bibinfo {pages} {1012} (\bibinfo {year} {2020})}\BibitemShut
  {NoStop}%
\bibitem [{\citenamefont {Pirandola}\ \emph {et~al.}(2009)\citenamefont
  {Pirandola}, \citenamefont {Garc\'{\i}a-Patr\'on}, \citenamefont
  {Braunstein},\ and\ \citenamefont {Lloyd}}]{PirPatron09}%
  \BibitemOpen
  \bibfield  {author} {\bibinfo {author} {\bibfnamefont {S.}~\bibnamefont
  {Pirandola}}, \bibinfo {author} {\bibfnamefont {R.}~\bibnamefont
  {Garc\'{\i}a-Patr\'on}}, \bibinfo {author} {\bibfnamefont {S.~L.}\
  \bibnamefont {Braunstein}},\ and\ \bibinfo {author} {\bibfnamefont
  {S.}~\bibnamefont {Lloyd}},\ }\bibfield  {title} {\bibinfo {title} {Direct
  and reverse secret-key capacities of a quantum channel},\ }\href
  {https://link.aps.org/doi/10.1103/PhysRevLett.102.050503} {\bibfield
  {journal} {\bibinfo  {journal} {Phys. Rev. Lett.}\ }\textbf {\bibinfo
  {volume} {102}},\ \bibinfo {pages} {050503} (\bibinfo {year}
  {2009})}\BibitemShut {NoStop}%
\bibitem [{\citenamefont {Pirandola}\ \emph {et~al.}(2017)\citenamefont
  {Pirandola}, \citenamefont {Laurenza}, \citenamefont {Ottaviani},\ and\
  \citenamefont {Banchi}}]{PLOB}%
  \BibitemOpen
  \bibfield  {author} {\bibinfo {author} {\bibfnamefont {S.}~\bibnamefont
  {Pirandola}}, \bibinfo {author} {\bibfnamefont {R.}~\bibnamefont {Laurenza}},
  \bibinfo {author} {\bibfnamefont {C.}~\bibnamefont {Ottaviani}},\ and\
  \bibinfo {author} {\bibfnamefont {L.}~\bibnamefont {Banchi}},\ }\bibfield
  {title} {\bibinfo {title} {Fundamental limits of repeaterless quantum
  communications},\ }\href {https://doi.org/10.1038/ncomms15043} {\bibfield
  {journal} {\bibinfo  {journal} {Nature Communications}\ }\textbf {\bibinfo
  {volume} {8}},\ \bibinfo {pages} {15043} (\bibinfo {year}
  {2017})}\BibitemShut {NoStop}%
\bibitem [{\citenamefont {Pirandola}(2019)}]{End2End}%
  \BibitemOpen
  \bibfield  {author} {\bibinfo {author} {\bibfnamefont {S.}~\bibnamefont
  {Pirandola}},\ }\bibfield  {title} {\bibinfo {title} {End-to-end capacities
  of a quantum communication network},\ }\href
  {https://doi.org/10.1038/s42005-019-0147-3} {\bibfield  {journal} {\bibinfo
  {journal} {Communications Physics}\ }\textbf {\bibinfo {volume} {2}},\
  \bibinfo {pages} {51} (\bibinfo {year} {2019})}\BibitemShut {NoStop}%
\bibitem [{\citenamefont {Zhuang}\ and\ \citenamefont
  {Zhang}(2021)}]{QuntaoRandQNets}%
  \BibitemOpen
  \bibfield  {author} {\bibinfo {author} {\bibfnamefont {Q.}~\bibnamefont
  {Zhuang}}\ and\ \bibinfo {author} {\bibfnamefont {B.}~\bibnamefont {Zhang}},\
  }\bibfield  {title} {\bibinfo {title} {Quantum communication capacity
  transition of complex quantum networks},\ }\href
  {https://link.aps.org/doi/10.1103/PhysRevA.104.022608} {\bibfield  {journal}
  {\bibinfo  {journal} {Phys. Rev. A}\ }\textbf {\bibinfo {volume} {104}},\
  \bibinfo {pages} {022608} (\bibinfo {year} {2021})}\BibitemShut {NoStop}%
\bibitem [{\citenamefont {Zhang}\ and\ \citenamefont
  {Zhuang}(2021)}]{ZhangQInt}%
  \BibitemOpen
  \bibfield  {author} {\bibinfo {author} {\bibfnamefont {B.}~\bibnamefont
  {Zhang}}\ and\ \bibinfo {author} {\bibfnamefont {Q.}~\bibnamefont {Zhuang}},\
  }\bibfield  {title} {\bibinfo {title} {Quantum internet under random
  breakdowns and intentional attacks},\ }\href
  {http://iopscience.iop.org/article/10.1088/2058-9565/ac1041} {\bibfield
  {journal} {\bibinfo  {journal} {Quantum Science and Technology}\ }\textbf
  {\bibinfo {volume} {6}},\ \bibinfo {pages} {045007} (\bibinfo {year}
  {2021})}\BibitemShut {NoStop}%
\bibitem [{\citenamefont {Harney}\ and\ \citenamefont
  {Pirandola}(2021)}]{OPGQN}%
  \BibitemOpen
    \bibfield  {author} {\bibinfo {author} {\bibfnamefont {C.}~\bibnamefont
  {Harney}}\ and\ \bibinfo {author} {\bibfnamefont {S.}~\bibnamefont
  {Pirandola}},\ }\bibfield  {title} {\bibinfo {title} {Analytical methods for
  high-rate global quantum networks},\ }\href
  {https://link.aps.org/doi/10.1103/PRXQuantum.3.010349}
  {\bibfield  {journal} {\bibinfo  {journal}
  {PRX Quantum}} \textbf {\bibinfo
  {volume} {3}},\ \bibinfo {pages} {10349} (\bibinfo {year}
  {2022})}\BibitemShut {NoStop}%
\bibitem [{\citenamefont {Liao}\ \emph {et~al.}(2017)\citenamefont {Liao},
  \citenamefont {Lin}, \citenamefont {Ren}, \citenamefont {Liu}, \citenamefont
  {Qiang}, \citenamefont {Yin}, \citenamefont {Li}, \citenamefont {Shen},
  \citenamefont {Zhang}, \citenamefont {Cao} \emph {et~al.}}]{LiaoSQC}%
  \BibitemOpen
  \bibfield  {author} {\bibinfo {author} {\bibfnamefont {S.-K.}\ \bibnamefont
  {Liao}}, \bibinfo {author} {\bibfnamefont {J.}~\bibnamefont {Lin}}, \bibinfo
  {author} {\bibfnamefont {J.}~\bibnamefont {Ren}}, \bibinfo {author}
  {\bibfnamefont {W.}~\bibnamefont {Liu}}, \bibinfo {author} {\bibfnamefont
  {J.}~\bibnamefont {Qiang}}, \bibinfo {author} {\bibfnamefont
  {J.}~\bibnamefont {Yin}}, \bibinfo {author} {\bibfnamefont {Y.}~\bibnamefont
  {Li}}, \bibinfo {author} {\bibfnamefont {Q.}~\bibnamefont {Shen}}, \bibinfo
  {author} {\bibfnamefont {L.}~\bibnamefont {Zhang}}, \bibinfo {author}
  {\bibfnamefont {Y.}~\bibnamefont {Cao}}, \emph {et~al.},\ }\bibfield  {title}
  {\bibinfo {title} {{Space-to-Ground Quantum Key Distribution Using a
  Small-Sized Payload on Tiangong-2 Space Lab}},\ }\href@noop {} {\bibfield
  {journal} {\bibinfo  {journal} {Chinese Physics Letters}\ }\textbf {\bibinfo
  {volume} {34}},\ \bibinfo {eid} {090302} (\bibinfo {year}
  {2017})}\BibitemShut {NoStop}%
\bibitem [{\citenamefont {Yin}\ \emph {et~al.}(2017)\citenamefont {Yin},
  \citenamefont {Cao}, \citenamefont {Li}, \citenamefont {Liao}, \citenamefont
  {Zhang}, \citenamefont {Ren}, \citenamefont {Cai}, \citenamefont {Liu},
  \citenamefont {Li}, \citenamefont {Dai}, \citenamefont {Li} \emph
  {et~al.}}]{Micius2017}%
  \BibitemOpen
  \bibfield  {author} {\bibinfo {author} {\bibfnamefont {J.}~\bibnamefont
  {Yin}}, \bibinfo {author} {\bibfnamefont {Y.}~\bibnamefont {Cao}}, \bibinfo
  {author} {\bibfnamefont {Y.-H.}\ \bibnamefont {Li}}, \bibinfo {author}
  {\bibfnamefont {S.-K.}\ \bibnamefont {Liao}}, \bibinfo {author}
  {\bibfnamefont {L.}~\bibnamefont {Zhang}}, \bibinfo {author} {\bibfnamefont
  {J.-G.}\ \bibnamefont {Ren}}, \bibinfo {author} {\bibfnamefont {W.-Q.}\
  \bibnamefont {Cai}}, \bibinfo {author} {\bibfnamefont {W.-Y.}\ \bibnamefont
  {Liu}}, \bibinfo {author} {\bibfnamefont {B.}~\bibnamefont {Li}}, \bibinfo
  {author} {\bibfnamefont {H.}~\bibnamefont {Dai}}, \bibinfo {author}
  {\bibfnamefont {G.-B.}\ \bibnamefont {Li}}, \emph {et~al.},\ }\bibfield
  {title} {\bibinfo {title} {Satellite-based entanglement distribution over
  1200 kilometers},\ }\href {https://doi.org/10.1126/science.aan3211}
  {\bibfield  {journal} {\bibinfo  {journal} {Science}\ }\textbf {\bibinfo
  {volume} {356}},\ \bibinfo {pages} {1140} (\bibinfo {year}
  {2017})}\BibitemShut {NoStop}%
\bibitem [{\citenamefont {Ren}\ \emph {et~al.}(2017)\citenamefont {Ren},
  \citenamefont {Xu}, \citenamefont {Yong}, \citenamefont {Zhang},
  \citenamefont {Liao}, \citenamefont {Yin}, \citenamefont {Liu}, \citenamefont
  {Cai}, \citenamefont {Yang}, \citenamefont {Li} \emph {et~al.}}]{RenSQC}%
  \BibitemOpen
  \bibfield  {author} {\bibinfo {author} {\bibfnamefont {J.-G.}\ \bibnamefont
  {Ren}}, \bibinfo {author} {\bibfnamefont {P.}~\bibnamefont {Xu}}, \bibinfo
  {author} {\bibfnamefont {H.-L.}\ \bibnamefont {Yong}}, \bibinfo {author}
  {\bibfnamefont {L.}~\bibnamefont {Zhang}}, \bibinfo {author} {\bibfnamefont
  {S.-K.}\ \bibnamefont {Liao}}, \bibinfo {author} {\bibfnamefont
  {J.}~\bibnamefont {Yin}}, \bibinfo {author} {\bibfnamefont {W.-Y.}\
  \bibnamefont {Liu}}, \bibinfo {author} {\bibfnamefont {W.-Q.}\ \bibnamefont
  {Cai}}, \bibinfo {author} {\bibfnamefont {M.}~\bibnamefont {Yang}}, \bibinfo
  {author} {\bibfnamefont {L.}~\bibnamefont {Li}}, \emph {et~al.},\ }\bibfield
  {title} {\bibinfo {title} {Ground-to-satellite quantum teleportation},\
  }\href {http://dx.doi.org/10.1038/nature23675} {\bibfield  {journal}
  {\bibinfo  {journal} {Nature}\ }\textbf {\bibinfo {volume} {549}},\ \bibinfo
  {pages} {70–73} (\bibinfo {year} {2017})}\BibitemShut {NoStop}%
\bibitem [{\citenamefont {Yin}\ \emph {et~al.}(2020)\citenamefont {Yin},
  \citenamefont {Li}, \citenamefont {Liao}, \citenamefont {Yang}, \citenamefont
  {Cao}, \citenamefont {Zhang}, \citenamefont {Ren}, \citenamefont {Cai},
  \citenamefont {Liu}, \citenamefont {Li}, \citenamefont {Shu} \emph
  {et~al.}}]{Micius2020}%
  \BibitemOpen
  \bibfield  {author} {\bibinfo {author} {\bibfnamefont {J.}~\bibnamefont
  {Yin}}, \bibinfo {author} {\bibfnamefont {Y.-H.}\ \bibnamefont {Li}},
  \bibinfo {author} {\bibfnamefont {S.-K.}\ \bibnamefont {Liao}}, \bibinfo
  {author} {\bibfnamefont {M.}~\bibnamefont {Yang}}, \bibinfo {author}
  {\bibfnamefont {Y.}~\bibnamefont {Cao}}, \bibinfo {author} {\bibfnamefont
  {L.}~\bibnamefont {Zhang}}, \bibinfo {author} {\bibfnamefont {J.-G.}\
  \bibnamefont {Ren}}, \bibinfo {author} {\bibfnamefont {W.-Q.}\ \bibnamefont
  {Cai}}, \bibinfo {author} {\bibfnamefont {W.-Y.}\ \bibnamefont {Liu}},
  \bibinfo {author} {\bibfnamefont {S.-L.}\ \bibnamefont {Li}}, \bibinfo
  {author} {\bibfnamefont {R.}~\bibnamefont {Shu}}, \emph {et~al.},\ }\bibfield
   {title} {\bibinfo {title} {Entanglement-based secure quantum cryptography
  over 1,120 kilometres},\ }\href {https://doi.org/10.1038/s41586-020-2401-y}
  {\bibfield  {journal} {\bibinfo  {journal} {Nature}\ }\textbf {\bibinfo
  {volume} {582}},\ \bibinfo {pages} {501} (\bibinfo {year}
  {2020})}\BibitemShut {NoStop}%
\bibitem [{\citenamefont {Villar}\ \emph {et~al.}(2020)\citenamefont {Villar},
  \citenamefont {Lohrmann}, \citenamefont {Bai}, \citenamefont {Vergoossen},
  \citenamefont {Bedington}, \citenamefont {Perumangatt}, \citenamefont {Lim},
  \citenamefont {Islam}, \citenamefont {Reezwana}, \citenamefont {Tang} \emph
  {et~al.}}]{VillarSQC}%
  \BibitemOpen
  \bibfield  {author} {\bibinfo {author} {\bibfnamefont {A.}~\bibnamefont
  {Villar}}, \bibinfo {author} {\bibfnamefont {A.}~\bibnamefont {Lohrmann}},
  \bibinfo {author} {\bibfnamefont {X.}~\bibnamefont {Bai}}, \bibinfo {author}
  {\bibfnamefont {T.}~\bibnamefont {Vergoossen}}, \bibinfo {author}
  {\bibfnamefont {R.}~\bibnamefont {Bedington}}, \bibinfo {author}
  {\bibfnamefont {C.}~\bibnamefont {Perumangatt}}, \bibinfo {author}
  {\bibfnamefont {H.}~\bibnamefont {Lim}}, \bibinfo {author} {\bibfnamefont
  {T.}~\bibnamefont {Islam}}, \bibinfo {author} {\bibfnamefont
  {A.}~\bibnamefont {Reezwana}}, \bibinfo {author} {\bibfnamefont
  {Z.}~\bibnamefont {Tang}}, \emph {et~al.},\ }\bibfield  {title} {\bibinfo
  {title} {Entanglement demonstration on board a nano-satellite},\ }\href
  {http://www.osapublishing.org/optica/abstract.cfm?URI=optica-7-7-734}
  {\bibfield  {journal} {\bibinfo  {journal} {Optica}\ }\textbf {\bibinfo
  {volume} {7}},\ \bibinfo {pages} {734} (\bibinfo {year} {2020})}\BibitemShut
  {NoStop}%
\bibitem [{\citenamefont {Chen}\ \emph {et~al.}(2021)\citenamefont {Chen},
  \citenamefont {Zhang}, \citenamefont {Chen}, \citenamefont {Cai},
  \citenamefont {Liao}, \citenamefont {Zhang}, \citenamefont {Chen},
  \citenamefont {Yin}, \citenamefont {Ren}, \citenamefont {Chen}, \citenamefont
  {Han} \emph {et~al.}}]{Micius2021}%
  \BibitemOpen
  \bibfield  {author} {\bibinfo {author} {\bibfnamefont {Y.-A.}\ \bibnamefont
  {Chen}}, \bibinfo {author} {\bibfnamefont {Q.}~\bibnamefont {Zhang}},
  \bibinfo {author} {\bibfnamefont {T.-Y.}\ \bibnamefont {Chen}}, \bibinfo
  {author} {\bibfnamefont {W.-Q.}\ \bibnamefont {Cai}}, \bibinfo {author}
  {\bibfnamefont {S.-K.}\ \bibnamefont {Liao}}, \bibinfo {author}
  {\bibfnamefont {J.}~\bibnamefont {Zhang}}, \bibinfo {author} {\bibfnamefont
  {K.}~\bibnamefont {Chen}}, \bibinfo {author} {\bibfnamefont {J.}~\bibnamefont
  {Yin}}, \bibinfo {author} {\bibfnamefont {J.-G.}\ \bibnamefont {Ren}},
  \bibinfo {author} {\bibfnamefont {Z.}~\bibnamefont {Chen}}, \bibinfo {author}
  {\bibfnamefont {S.-L.}\ \bibnamefont {Han}}, \emph {et~al.},\ }\bibfield
  {title} {\bibinfo {title} {An integrated space-to-ground quantum
  communication network over 4,600 kilometres},\ }\href
  {https://doi.org/10.1038/s41586-020-03093-8} {\bibfield  {journal} {\bibinfo
  {journal} {Nature}\ }\textbf {\bibinfo {volume} {589}},\ \bibinfo {pages}
  {214} (\bibinfo {year} {2021})}\BibitemShut {NoStop}%
\bibitem [{\citenamefont {Sidhu}\ \emph {et~al.}(2021)\citenamefont {Sidhu},
  \citenamefont {Joshi}, \citenamefont {Gundogan}, \citenamefont {Brougham},
  \citenamefont {Lowndes}, \citenamefont {Mazzarella}, \citenamefont {Krutzik},
  \citenamefont {Mohapatra}, \citenamefont {Dequal}, \citenamefont {Vallone}
  \emph {et~al.}}]{AdvSQC}%
  \BibitemOpen
  \bibfield  {author} {\bibinfo {author} {\bibfnamefont {J.~S.}\ \bibnamefont
  {Sidhu}}, \bibinfo {author} {\bibfnamefont {S.~K.}\ \bibnamefont {Joshi}},
  \bibinfo {author} {\bibfnamefont {M.}~\bibnamefont {Gundogan}}, \bibinfo
  {author} {\bibfnamefont {T.}~\bibnamefont {Brougham}}, \bibinfo {author}
  {\bibfnamefont {D.}~\bibnamefont {Lowndes}}, \bibinfo {author} {\bibfnamefont
  {L.}~\bibnamefont {Mazzarella}}, \bibinfo {author} {\bibfnamefont
  {M.}~\bibnamefont {Krutzik}}, \bibinfo {author} {\bibfnamefont
  {S.}~\bibnamefont {Mohapatra}}, \bibinfo {author} {\bibfnamefont
  {D.}~\bibnamefont {Dequal}}, \bibinfo {author} {\bibfnamefont
  {G.}~\bibnamefont {Vallone}}, \emph {et~al.},\ }\bibfield  {title} {\bibinfo
  {title} {Advances in space quantum communications},\ }\href
  {https://ietresearch.onlinelibrary.wiley.com/doi/abs/10.1049/qtc2.12015}
  {\bibfield  {journal} {\bibinfo  {journal} {IET Quant. Comm.}\ ,\ \bibinfo
  {pages} {1}} (\bibinfo {year} {2021})}\BibitemShut {NoStop}%
\bibitem [{\citenamefont {Nielsen}\ and\ \citenamefont
  {Chuang}(2011)}]{Mike_Ike}%
  \BibitemOpen
  \bibfield  {author} {\bibinfo {author} {\bibfnamefont {M.~A.}\ \bibnamefont
  {Nielsen}}\ and\ \bibinfo {author} {\bibfnamefont {I.~L.}\ \bibnamefont
  {Chuang}},\ }\href@noop {} {\emph {\bibinfo {title} {Quantum Computation and
  Quantum Information: 10th Anniversary Edition}}},\ \bibinfo {edition} {10th}\
  ed.\ (\bibinfo  {publisher} {Cambridge University Press},\ \bibinfo {address}
  {USA},\ \bibinfo {year} {2011})\BibitemShut {NoStop}%
\bibitem [{\citenamefont {Weedbrook}\ \emph {et~al.}(2012)\citenamefont
  {Weedbrook}, \citenamefont {Pirandola}, \citenamefont {Garc\'{\i}a-Patr\'on},
  \citenamefont {Cerf}, \citenamefont {Ralph}, \citenamefont {Shapiro},\ and\
  \citenamefont {Lloyd}}]{GaussRev}%
  \BibitemOpen
  \bibfield  {author} {\bibinfo {author} {\bibfnamefont {C.}~\bibnamefont
  {Weedbrook}}, \bibinfo {author} {\bibfnamefont {S.}~\bibnamefont
  {Pirandola}}, \bibinfo {author} {\bibfnamefont {R.}~\bibnamefont
  {Garc\'{\i}a-Patr\'on}}, \bibinfo {author} {\bibfnamefont {N.~J.}\
  \bibnamefont {Cerf}}, \bibinfo {author} {\bibfnamefont {T.~C.}\ \bibnamefont
  {Ralph}}, \bibinfo {author} {\bibfnamefont {J.~H.}\ \bibnamefont {Shapiro}},\
  and\ \bibinfo {author} {\bibfnamefont {S.}~\bibnamefont {Lloyd}},\ }\bibfield
   {title} {\bibinfo {title} {Gaussian quantum information},\ }\href
  {https://doi.org/10.1103/RevModPhys.84.621} {\bibfield  {journal} {\bibinfo
  {journal} {Rev. Mod. Phys.}\ }\textbf {\bibinfo {volume} {84}},\ \bibinfo
  {pages} {621} (\bibinfo {year} {2012})}\BibitemShut {NoStop}%
\bibitem [{\citenamefont {Serafini}(2017)}]{SerafiniCV}%
  \BibitemOpen
  \bibfield  {author} {\bibinfo {author} {\bibfnamefont {A.}~\bibnamefont
  {Serafini}},\ }\href {https://books.google.co.uk/books?id=zHtgvgAACAAJ}
  {\emph {\bibinfo {title} {Quantum Continuous Variables: A Primer of
  Theoretical Methods}}}\ (\bibinfo  {publisher} {CRC Press, Taylor \& Francis
  Group},\ \bibinfo {year} {2017})\BibitemShut {NoStop}%
\bibitem [{\citenamefont {Goodman}(1985)}]{GoodmanOptics}%
  \BibitemOpen
  \bibfield  {author} {\bibinfo {author} {\bibfnamefont {J.}~\bibnamefont
  {Goodman}},\ }\href@noop {} {\emph {\bibinfo {title} {Statistical Optics}}}\
  (\bibinfo  {publisher} {John Wiley \& Sons},\ \bibinfo {address} {New York},\
  \bibinfo {year} {1985})\BibitemShut {NoStop}%
\bibitem [{\citenamefont {Svelto}(2010)}]{SveltoLasers}%
  \BibitemOpen
  \bibfield  {author} {\bibinfo {author} {\bibfnamefont {O.}~\bibnamefont
  {Svelto}},\ }\href@noop {} {\emph {\bibinfo {title} {Principles of
  Lasers}}},\ \bibinfo {edition} {5th}\ ed.\ (\bibinfo  {publisher}
  {Springer},\ \bibinfo {address} {New York},\ \bibinfo {year}
  {2010})\BibitemShut {NoStop}%
\bibitem [{\citenamefont {Bohren}\ and\ \citenamefont
  {Huffman}(2008)}]{BohrenScatter}%
  \BibitemOpen
  \bibfield  {author} {\bibinfo {author} {\bibfnamefont {C.}~\bibnamefont
  {Bohren}}\ and\ \bibinfo {author} {\bibfnamefont {D.}~\bibnamefont
  {Huffman}},\ }\href@noop {} {\emph {\bibinfo {title} {Absorption and
  scattering of light by small particles}}}\ (\bibinfo  {publisher} {John Wiley
  \& Sons},\ \bibinfo {year} {2008})\BibitemShut {NoStop}%
\bibitem [{\citenamefont {Tatarskii}(1971)}]{TatarskiiTurb}%
  \BibitemOpen
  \bibfield  {author} {\bibinfo {author} {\bibfnamefont {V.}~\bibnamefont
  {Tatarskii}},\ }\href@noop {} {\emph {\bibinfo {title} {The effects of the
  turbulent atmosphere on wave propagation}}}\ (\bibinfo  {publisher} {Israel
  Program for Scientific Translations},\ \bibinfo {address} {Jerusalem},\
  \bibinfo {year} {1971})\BibitemShut {NoStop}%
\bibitem [{\citenamefont {Majumdar}\ and\ \citenamefont
  {Ricklin}(2008)}]{MajumdarTurb}%
  \BibitemOpen
  \bibfield  {author} {\bibinfo {author} {\bibfnamefont {A.}~\bibnamefont
  {Majumdar}}\ and\ \bibinfo {author} {\bibfnamefont {J.}~\bibnamefont
  {Ricklin}},\ }\href@noop {} {\emph {\bibinfo {title} {Free-Space Laser
  Communications}}}\ (\bibinfo  {publisher} {Springer},\ \bibinfo {address}
  {New York},\ \bibinfo {year} {2008})\BibitemShut {NoStop}%
\bibitem [{\citenamefont {Kaushal}\ \emph {et~al.}(2017)\citenamefont
  {Kaushal}, \citenamefont {Jain},\ and\ \citenamefont {Kar}}]{KaushalTurb}%
  \BibitemOpen
  \bibfield  {author} {\bibinfo {author} {\bibfnamefont {H.}~\bibnamefont
  {Kaushal}}, \bibinfo {author} {\bibfnamefont {V.}~\bibnamefont {Jain}},\ and\
  \bibinfo {author} {\bibfnamefont {S.}~\bibnamefont {Kar}},\ }\href@noop {}
  {\emph {\bibinfo {title} {Free Space Optical Communication}}}\ (\bibinfo
  {publisher} {Springer},\ \bibinfo {address} {New York},\ \bibinfo {year}
  {2017})\BibitemShut {NoStop}%
\bibitem [{\citenamefont {Andrews}\ and\ \citenamefont
  {Phillips}(2005)}]{AndrewsTurb}%
  \BibitemOpen
  \bibfield  {author} {\bibinfo {author} {\bibfnamefont {L.}~\bibnamefont
  {Andrews}}\ and\ \bibinfo {author} {\bibfnamefont {R.}~\bibnamefont
  {Phillips}},\ }\href@noop {} {\emph {\bibinfo {title} {Laser Beam Propagation
  Through Random Medium}}}\ (\bibinfo  {publisher} {SPIE},\ \bibinfo {address}
  {Bellinghan},\ \bibinfo {year} {2005})\BibitemShut {NoStop}%
\bibitem [{\citenamefont {Pirandola}(2021{\natexlab{a}})}]{FS}%
  \BibitemOpen
  \bibfield  {author} {\bibinfo {author} {\bibfnamefont {S.}~\bibnamefont
  {Pirandola}},\ }\bibfield  {title} {\bibinfo {title} {Limits and security of
  free-space quantum communications},\ }\href
  {https://link.aps.org/doi/10.1103/PhysRevResearch.3.013279} {\bibfield
  {journal} {\bibinfo  {journal} {Phys. Rev. Research}\ }\textbf {\bibinfo
  {volume} {3}},\ \bibinfo {pages} {013279} (\bibinfo {year}
  {2021}{\natexlab{a}})}\BibitemShut {NoStop}%
\bibitem [{\citenamefont {Pirandola}(2021{\natexlab{b}})}]{SQC}%
  \BibitemOpen
  \bibfield  {author} {\bibinfo {author} {\bibfnamefont {S.}~\bibnamefont
  {Pirandola}},\ }\bibfield  {title} {\bibinfo {title} {Satellite quantum
  communications: Fundamental bounds and practical security},\ }\href
  {https://link.aps.org/doi/10.1103/PhysRevResearch.3.023130} {\bibfield
  {journal} {\bibinfo  {journal} {Phys. Rev. Research}\ }\textbf {\bibinfo
  {volume} {3}},\ \bibinfo {pages} {023130} (\bibinfo {year}
  {2021}{\natexlab{b}})}\BibitemShut {NoStop}%
\bibitem [{\citenamefont {Borregaard}\ \emph {et~al.}(2020)\citenamefont
  {Borregaard}, \citenamefont {Pichler}, \citenamefont {Schr\"oder},
  \citenamefont {Lukin}, \citenamefont {Lodahl},\ and\ \citenamefont
  {S\o{}rensen}}]{Borregaard_1WRep}%
  \BibitemOpen
  \bibfield  {author} {\bibinfo {author} {\bibfnamefont {J.}~\bibnamefont
  {Borregaard}}, \bibinfo {author} {\bibfnamefont {H.}~\bibnamefont {Pichler}},
  \bibinfo {author} {\bibfnamefont {T.}~\bibnamefont {Schr\"oder}}, \bibinfo
  {author} {\bibfnamefont {M.~D.}\ \bibnamefont {Lukin}}, \bibinfo {author}
  {\bibfnamefont {P.}~\bibnamefont {Lodahl}},\ and\ \bibinfo {author}
  {\bibfnamefont {A.~S.}\ \bibnamefont {S\o{}rensen}},\ }\bibfield  {title}
  {\bibinfo {title} {One-way quantum repeater based on near-deterministic
  photon-emitter interfaces},\ }\href
  {https://link.aps.org/doi/10.1103/PhysRevX.10.021071} {\bibfield  {journal}
  {\bibinfo  {journal} {Phys. Rev. X}\ }\textbf {\bibinfo {volume} {10}},\
  \bibinfo {pages} {021071} (\bibinfo {year} {2020})}\BibitemShut {NoStop}%
\bibitem [{\citenamefont {Childress}\ \emph {et~al.}(2006)\citenamefont
  {Childress}, \citenamefont {Taylor}, \citenamefont {S\o{}rensen},\ and\
  \citenamefont {Lukin}}]{Childress_FTQComms}%
  \BibitemOpen
  \bibfield  {author} {\bibinfo {author} {\bibfnamefont {L.}~\bibnamefont
  {Childress}}, \bibinfo {author} {\bibfnamefont {J.~M.}\ \bibnamefont
  {Taylor}}, \bibinfo {author} {\bibfnamefont {A.~S.}\ \bibnamefont
  {S\o{}rensen}},\ and\ \bibinfo {author} {\bibfnamefont {M.~D.}\ \bibnamefont
  {Lukin}},\ }\bibfield  {title} {\bibinfo {title} {Fault-tolerant quantum
  communication based on solid-state photon emitters},\ }\href
  {https://link.aps.org/doi/10.1103/PhysRevLett.96.070504} {\bibfield
  {journal} {\bibinfo  {journal} {Phys. Rev. Lett.}\ }\textbf {\bibinfo
  {volume} {96}},\ \bibinfo {pages} {070504} (\bibinfo {year}
  {2006})}\BibitemShut {NoStop}%
\bibitem [{\citenamefont {Simon}\ \emph {et~al.}(2007)\citenamefont {Simon},
  \citenamefont {de~Riedmatten}, \citenamefont {Afzelius}, \citenamefont
  {Sangouard}, \citenamefont {Zbinden},\ and\ \citenamefont
  {Gisin}}]{Simon_Repeaters}%
  \BibitemOpen
  \bibfield  {author} {\bibinfo {author} {\bibfnamefont {C.}~\bibnamefont
  {Simon}}, \bibinfo {author} {\bibfnamefont {H.}~\bibnamefont
  {de~Riedmatten}}, \bibinfo {author} {\bibfnamefont {M.}~\bibnamefont
  {Afzelius}}, \bibinfo {author} {\bibfnamefont {N.}~\bibnamefont {Sangouard}},
  \bibinfo {author} {\bibfnamefont {H.}~\bibnamefont {Zbinden}},\ and\ \bibinfo
  {author} {\bibfnamefont {N.}~\bibnamefont {Gisin}},\ }\bibfield  {title}
  {\bibinfo {title} {Quantum repeaters with photon pair sources and multimode
  memories},\ }\href {https://link.aps.org/doi/10.1103/PhysRevLett.98.190503}
  {\bibfield  {journal} {\bibinfo  {journal} {Phys. Rev. Lett.}\ }\textbf
  {\bibinfo {volume} {98}},\ \bibinfo {pages} {190503} (\bibinfo {year}
  {2007})}\BibitemShut {NoStop}%
\bibitem [{\citenamefont {Usenko}\ \emph {et~al.}(2012)\citenamefont {Usenko},
  \citenamefont {Heim}, \citenamefont {Peuntinger}, \citenamefont {Wittmann},
  \citenamefont {Marquardt}, \citenamefont {Leuchs},\ and\ \citenamefont
  {Filip}}]{UsenkoFading}%
  \BibitemOpen
  \bibfield  {author} {\bibinfo {author} {\bibfnamefont {V.~C.}\ \bibnamefont
  {Usenko}}, \bibinfo {author} {\bibfnamefont {B.}~\bibnamefont {Heim}},
  \bibinfo {author} {\bibfnamefont {C.}~\bibnamefont {Peuntinger}}, \bibinfo
  {author} {\bibfnamefont {C.}~\bibnamefont {Wittmann}}, \bibinfo {author}
  {\bibfnamefont {C.}~\bibnamefont {Marquardt}}, \bibinfo {author}
  {\bibfnamefont {G.}~\bibnamefont {Leuchs}},\ and\ \bibinfo {author}
  {\bibfnamefont {R.}~\bibnamefont {Filip}},\ }\bibfield  {title} {\bibinfo
  {title} {Entanglement of gaussian states and the applicability to quantum key
  distribution over fading channels},\ }\href
  {https://doi.org/10.1088/1367-2630/14/9/093048} {\bibfield  {journal}
  {\bibinfo  {journal} {New Journal of Physics}\ }\textbf {\bibinfo {volume}
  {14}},\ \bibinfo {pages} {093048} (\bibinfo {year} {2012})}\BibitemShut
  {NoStop}%
\bibitem [{\citenamefont {Papanastasiou}\ \emph {et~al.}(2018)\citenamefont
  {Papanastasiou}, \citenamefont {Weedbrook},\ and\ \citenamefont
  {Pirandola}}]{PanosFading}%
  \BibitemOpen
  \bibfield  {author} {\bibinfo {author} {\bibfnamefont {P.}~\bibnamefont
  {Papanastasiou}}, \bibinfo {author} {\bibfnamefont {C.}~\bibnamefont
  {Weedbrook}},\ and\ \bibinfo {author} {\bibfnamefont {S.}~\bibnamefont
  {Pirandola}},\ }\bibfield  {title} {\bibinfo {title} {Continuous-variable
  quantum key distribution in uniform fast-fading channels},\ }\href
  {https://link.aps.org/doi/10.1103/PhysRevA.97.032311} {\bibfield  {journal}
  {\bibinfo  {journal} {Phys. Rev. A}\ }\textbf {\bibinfo {volume} {97}},\
  \bibinfo {pages} {032311} (\bibinfo {year} {2018})}\BibitemShut {NoStop}%
\bibitem [{\citenamefont {Pirandola}\ \emph {et~al.}(2019)\citenamefont
  {Pirandola}, \citenamefont {Laurenza},\ and\ \citenamefont
  {Banchi}}]{CondChanSim}%
  \BibitemOpen
  \bibfield  {author} {\bibinfo {author} {\bibfnamefont {S.}~\bibnamefont
  {Pirandola}}, \bibinfo {author} {\bibfnamefont {R.}~\bibnamefont
  {Laurenza}},\ and\ \bibinfo {author} {\bibfnamefont {L.}~\bibnamefont
  {Banchi}},\ }\bibfield  {title} {\bibinfo {title} {Conditional channel
  simulation},\ }\href
  {https://www.sciencedirect.com/science/article/pii/S0003491618302999}
  {\bibfield  {journal} {\bibinfo  {journal} {Annals of Physics}\ }\textbf
  {\bibinfo {volume} {400}},\ \bibinfo {pages} {289} (\bibinfo {year}
  {2019})}\BibitemShut {NoStop}%
\bibitem [{Note1()}]{Note1}%
  \BibitemOpen
  \bibinfo {note} {The word \protect \textit {tight} in this context refers to
  how close the upper-bound is from its best known lower-bound. Indeed, there
  exists a lower-bound on the capacity of a point-to-point thermal-loss channel
  based on its reverse coherent information (RCI) \cite {RCInfo}. Hence,
  throughout our work we implicitly refer to tight upper-bounds on thermal-loss
  channel capacities (and subsequently, network capacities) as those which in
  conjunction with the RCI can tightly sandwich the exact
  capacity.}\BibitemShut {Stop}%
\bibitem [{\citenamefont {Medhi}\ and\ \citenamefont
  {Ramasamy}(2018)}]{Dijkstra}%
  \BibitemOpen
  \bibfield  {author} {\bibinfo {author} {\bibfnamefont {D.}~\bibnamefont
  {Medhi}}\ and\ \bibinfo {author} {\bibfnamefont {K.}~\bibnamefont
  {Ramasamy}},\ }\href@noop {} {\emph {\bibinfo {title} {Network Routing:
  Algorithms, Protocols, and Architectures}}},\ \bibinfo {edition} {2nd}\ ed.\
  (\bibinfo  {publisher} {Morgan Kaufmann publishers},\ \bibinfo {address}
  {Cambridge MA},\ \bibinfo {year} {2018})\BibitemShut {NoStop}%
\bibitem [{\citenamefont {Solomons}\ \emph {et~al.}(2021)\citenamefont
  {Solomons}, \citenamefont {Fletcher}, \citenamefont {Aktas}, \citenamefont
  {Venkatachalam}, \citenamefont {Wengerowsky}, \citenamefont {Lončarić},
  \citenamefont {Neumann}, \citenamefont {Liu}, \citenamefont {Željko Samec},
  \citenamefont {Stipčević}, \citenamefont {Ursin} \emph
  {et~al.}}]{Alas_Flooding}%
  \BibitemOpen
  \bibfield  {author} {\bibinfo {author} {\bibfnamefont {N.~R.}\ \bibnamefont
  {Solomons}}, \bibinfo {author} {\bibfnamefont {A.~I.}\ \bibnamefont
  {Fletcher}}, \bibinfo {author} {\bibfnamefont {D.}~\bibnamefont {Aktas}},
  \bibinfo {author} {\bibfnamefont {N.}~\bibnamefont {Venkatachalam}}, \bibinfo
  {author} {\bibfnamefont {S.}~\bibnamefont {Wengerowsky}}, \bibinfo {author}
  {\bibfnamefont {M.}~\bibnamefont {Lončarić}}, \bibinfo {author}
  {\bibfnamefont {S.~P.}\ \bibnamefont {Neumann}}, \bibinfo {author}
  {\bibfnamefont {B.}~\bibnamefont {Liu}}, \bibinfo {author} {\bibnamefont
  {Željko Samec}}, \bibinfo {author} {\bibfnamefont {M.}~\bibnamefont
  {Stipčević}}, \bibinfo {author} {\bibfnamefont {R.}~\bibnamefont {Ursin}},
  \emph {et~al.},\ }\bibfield  {title} {\bibinfo {title} {Scalable
  authentication and optimal flooding in a quantum network},\ }\href
  {https://link.aps.org/doi/10.1103/PRXQuantum.3.020311} {\bibfield  {journal} {\bibinfo  {journal}
  {PRX Quantum}} \textbf {\bibinfo {volume} {3}},\
  \bibinfo {pages} {020311} (\bibinfo {year} {2022})}\BibitemShut {NoStop}%
  \bibitem [{\citenamefont {Ford}\ and\ \citenamefont
  {Fulkerson}(1956)}]{FordFlow}%
  \BibitemOpen
  \bibfield  {author} {\bibinfo {author} {\bibfnamefont {L.~R.}\ \bibnamefont
  {Ford}}\ and\ \bibinfo {author} {\bibfnamefont {D.~R.}\ \bibnamefont
  {Fulkerson}},\ }\bibfield  {title} {\bibinfo {title} {Maximal flow through a
  network},\ }\href@noop {} {\bibfield  {journal} {\bibinfo  {journal}
  {Canadian Journal of Mathematics}\ }\textbf {\bibinfo {volume} {8}},\
  \bibinfo {pages} {399–404} (\bibinfo {year} {1956})}\BibitemShut {NoStop}%
\bibitem [{\citenamefont {Edmonds}\ and\ \citenamefont
  {Karp}(1972)}]{KarpFlow}%
  \BibitemOpen
  \bibfield  {author} {\bibinfo {author} {\bibfnamefont {J.}~\bibnamefont
  {Edmonds}}\ and\ \bibinfo {author} {\bibfnamefont {R.~M.}\ \bibnamefont
  {Karp}},\ }\bibfield  {title} {\bibinfo {title} {Theoretical improvements in
  algorithmic efficiency for network flow problems},\ }\href
  {https://doi.org/10.1145/321694.321699} {\bibfield  {journal} {\bibinfo
  {journal} {J. ACM}\ }\textbf {\bibinfo {volume} {19}},\ \bibinfo {pages}
  {248–264} (\bibinfo {year} {1972})}\BibitemShut {NoStop}%
\bibitem [{\citenamefont {{Orlin}}(2013)}]{OrlinFlow}%
  \BibitemOpen
  \bibfield  {author} {\bibinfo {author} {\bibfnamefont {J.~B.}\ \bibnamefont
  {{Orlin}}},\ }\bibfield  {title} {\bibinfo {title} {Max flows in o(nm) time,
  or better.},\ }in\ \href@noop {} {\emph {\bibinfo {booktitle} {Proceedings of
  the forty-fifth annual ACM symposium on Theory of computing, STOC’13}}}\
  (\bibinfo {year} {2013})\ pp.\ \bibinfo {pages} {765--774}\BibitemShut
  {NoStop}%
\bibitem [{\citenamefont {Ghalaii}\ and\ \citenamefont
  {Pirandola}(2021)}]{Masoud_StrongTurb}%
  \BibitemOpen
  \bibfield  {author} {\bibinfo {author} {\bibfnamefont {M.}~\bibnamefont
  {Ghalaii}}\ and\ \bibinfo {author} {\bibfnamefont {S.}~\bibnamefont
  {Pirandola}},\ }\bibfield  {title} {\bibinfo {title} {Quantum communications
  in a moderate-to-strong turbulent space},\ }\href
  {https://arxiv.org/abs/2107.12415v1} {\bibfield  {journal} {\bibinfo
  {journal} {arXiv:2107.12415}\ } (\bibinfo {year} {2021})}\BibitemShut
  {NoStop}%
\bibitem [{\citenamefont {Fante}(1975)}]{FanteFS1}%
  \BibitemOpen
  \bibfield  {author} {\bibinfo {author} {\bibfnamefont {R.}~\bibnamefont
  {Fante}},\ }\bibfield  {title} {\bibinfo {title} {Electromagnetic beam
  propagation in turbulent media},\ }\href@noop {} {\bibfield  {journal}
  {\bibinfo  {journal} {Proceedings of the IEEE}\ }\textbf {\bibinfo {volume}
  {63}},\ \bibinfo {pages} {1669} (\bibinfo {year} {1975})}\BibitemShut
  {NoStop}%
\bibitem [{\citenamefont {Fante}(1980)}]{FanteFS2}%
  \BibitemOpen
  \bibfield  {author} {\bibinfo {author} {\bibfnamefont {R.}~\bibnamefont
  {Fante}},\ }\bibfield  {title} {\bibinfo {title} {Electromagnetic beam
  propagation in turbulent media: An update},\ }\href@noop {} {\bibfield
  {journal} {\bibinfo  {journal} {Proceedings of the IEEE}\ }\textbf {\bibinfo
  {volume} {68}},\ \bibinfo {pages} {1424} (\bibinfo {year}
  {1980})}\BibitemShut {NoStop}%
\bibitem [{\citenamefont {Yura}(1973)}]{YuraFS}%
  \BibitemOpen
  \bibfield  {author} {\bibinfo {author} {\bibfnamefont {H.~T.}\ \bibnamefont
  {Yura}},\ }\bibfield  {title} {\bibinfo {title} {Short-term average
  optical-beam spread in a turbulent medium},\ }\href
  {http://www.osapublishing.org/abstract.cfm?URI=josa-63-5-567} {\bibfield
  {journal} {\bibinfo  {journal} {J. Opt. Soc. Am.}\ }\textbf {\bibinfo
  {volume} {63}},\ \bibinfo {pages} {567} (\bibinfo {year} {1973})}\BibitemShut
  {NoStop}%
\bibitem [{Note2()}]{Note2}%
  \BibitemOpen
  \bibinfo {note} {See the Supplementary Material [\textcolor{red}{LINK}] for definitions, lemmas, theorems, and their proofs for theoretical developments discussed in the main text.
  This includes discussions of networks with community structures from which this modular network emerges as a useful and highly desirable class, }\BibitemShut {NoStop}%
\bibitem [{Note3()}]{Note3}%
  \BibitemOpen
  \bibinfo {note} {While it might be more convenient to call this the \protect
  \textit {intercommunity capacity}, such a name might be confused as a more
  general term for the capacity when the end-users are located in different
  communities (which is implied). The global-community capacity is intended to
  be more distinct than this, and specify a particular network
  cut.}\BibitemShut {Stop}%
\bibitem [{\citenamefont {Garc\'{\i}a-Patr\'on}\ \emph
  {et~al.}(2009)\citenamefont {Garc\'{\i}a-Patr\'on}, \citenamefont
  {Pirandola}, \citenamefont {Lloyd},\ and\ \citenamefont {Shapiro}}]{RCInfo}%
  \BibitemOpen
  \bibfield  {author} {\bibinfo {author} {\bibfnamefont {R.}~\bibnamefont
  {Garc\'{\i}a-Patr\'on}}, \bibinfo {author} {\bibfnamefont {S.}~\bibnamefont
  {Pirandola}}, \bibinfo {author} {\bibfnamefont {S.}~\bibnamefont {Lloyd}},\
  and\ \bibinfo {author} {\bibfnamefont {J.~H.}\ \bibnamefont {Shapiro}},\
  }\bibfield  {title} {\bibinfo {title} {Reverse coherent information},\ }\href
  {https://link.aps.org/doi/10.1103/PhysRevLett.102.210501} {\bibfield
  {journal} {\bibinfo  {journal} {Phys. Rev. Lett.}\ }\textbf {\bibinfo
  {volume} {102}},\ \bibinfo {pages} {210501} (\bibinfo {year}
  {2009})}\BibitemShut {NoStop}%
  \bibitem [{\citenamefont {Gross}\ \emph {et~al.}(2020)\citenamefont {Gross},
  \citenamefont {Vaknin}, \citenamefont {Buldyrev},\ and\ \citenamefont
  {Havlin}}]{Gross2020}%
  \BibitemOpen
  \bibfield  {author} {\bibinfo {author} {\bibfnamefont {B.}~\bibnamefont
  {Gross}}, \bibinfo {author} {\bibfnamefont {D.}~\bibnamefont {Vaknin}},
  \bibinfo {author} {\bibfnamefont {S.~V.}\ \bibnamefont {Buldyrev}},\ and\
  \bibinfo {author} {\bibfnamefont {S.}~\bibnamefont {Havlin}},\ }\bibfield
  {title} {\bibinfo {title} {Two transitions in spatial modular networks},\
  }\href {https://doi.org/10.1088/1367-2630/ab8263} {\bibfield  {journal}
  {\bibinfo  {journal} {New Journal of Physics}\ }\textbf {\bibinfo {volume}
  {22}},\ \bibinfo {pages} {053002} (\bibinfo {year} {2020})}\BibitemShut
  {NoStop}%
\end{thebibliography}

\begin{thebibliography}{60}%
\makeatletter
\providecommand \@ifxundefined [1]{%
 \@ifx{#1\undefined}
}%
\providecommand \@ifnum [1]{%
 \ifnum #1\expandafter \@firstoftwo
 \else \expandafter \@secondoftwo
 \fi
}%
\providecommand \@ifx [1]{%
 \ifx #1\expandafter \@firstoftwo
 \else \expandafter \@secondoftwo
 \fi
}%
\providecommand \natexlab [1]{#1}%
\providecommand \enquote  [1]{``#1''}%
\providecommand \bibnamefont  [1]{#1}%
\providecommand \bibfnamefont [1]{#1}%
\providecommand \citenamefont [1]{#1}%
\providecommand \href@noop [0]{\@secondoftwo}%
\providecommand \href [0]{\begingroup \@sanitize@url \@href}%
\providecommand \@href[1]{\@@startlink{#1}\@@href}%
\providecommand \@@href[1]{\endgroup#1\@@endlink}%
\providecommand \@sanitize@url [0]{\catcode `\\12\catcode `\$12\catcode
  `\&12\catcode `\#12\catcode `\^12\catcode `\_12\catcode `\%12\relax}%
\providecommand \@@startlink[1]{}%
\providecommand \@@endlink[0]{}%
\providecommand \url  [0]{\begingroup\@sanitize@url \@url }%
\providecommand \@url [1]{\endgroup\@href {#1}{\urlprefix }}%
\providecommand \urlprefix  [0]{URL }%
\providecommand \Eprint [0]{\href }%
\providecommand \doibase [0]{https://doi.org/}%
\providecommand \selectlanguage [0]{\@gobble}%
\providecommand \bibinfo  [0]{\@secondoftwo}%
\providecommand \bibfield  [0]{\@secondoftwo}%
\providecommand \translation [1]{[#1]}%
\providecommand \BibitemOpen [0]{}%
\providecommand \bibitemStop [0]{}%
\providecommand \bibitemNoStop [0]{.\EOS\space}%
\providecommand \EOS [0]{\spacefactor3000\relax}%
\providecommand \BibitemShut  [1]{\csname bibitem#1\endcsname}%
\let\auto@bib@innerbib\@empty
\bibitem [\natexlab{1}{1}]{SlepianNets}%
  \BibitemOpen
  \bibfield  {author} {\bibinfo {author} {\bibfnamefont {P.}~\bibnamefont
  {Slepian}},\ }\href@noop {} {\emph {\bibinfo {title} {Mathematical
  Foundations of Network Analysis}}}\ (\bibinfo  {publisher}
  {Springer-Verlag},\ \bibinfo {address} {New York},\ \bibinfo {year}
  {1968})\BibitemShut {NoStop}%
\bibitem [{\citenamefont {Cover}\ and\ \citenamefont
  {Thomas}(2006)}]{CoverThomas}%
  \BibitemOpen
  \bibfield  {author} {\bibinfo {author} {\bibfnamefont {T.~M.}\ \bibnamefont
  {Cover}}\ and\ \bibinfo {author} {\bibfnamefont {J.~A.}\ \bibnamefont
  {Thomas}},\ }\href@noop {} {\emph {\bibinfo {title} {Elements of Information
  Theory}}}\ (\bibinfo  {publisher} {Wiley},\ \bibinfo {address} {New Jersey},\
  \bibinfo {year} {2006})\BibitemShut {NoStop}%
\bibitem [{\citenamefont {Tanenbaum}\ and\ \citenamefont
  {Wetherall}(2010)}]{TanenbaumNets}%
  \BibitemOpen
  \bibfield  {author} {\bibinfo {author} {\bibfnamefont {A.~S.}\ \bibnamefont
  {Tanenbaum}}\ and\ \bibinfo {author} {\bibfnamefont {D.~J.}\ \bibnamefont
  {Wetherall}},\ }\href@noop {} {\emph {\bibinfo {title} {Computer
  Networks}}},\ \bibinfo {edition} {5th}\ ed.\ (\bibinfo  {publisher}
  {Pearson},\ \bibinfo {year} {2010})\BibitemShut {NoStop}%
\bibitem [{\citenamefont {El~Gamal}\ and\ \citenamefont
  {Kim}(2011)}]{GamalNets}%
  \BibitemOpen
  \bibfield  {author} {\bibinfo {author} {\bibfnamefont {A.}~\bibnamefont
  {El~Gamal}}\ and\ \bibinfo {author} {\bibfnamefont {Y.-H.}\ \bibnamefont
  {Kim}},\ }\href@noop {} {\emph {\bibinfo {title} {Network Information
  Theory}}}\ (\bibinfo  {publisher} {Cambridge University Press},\ \bibinfo
  {year} {2011})\BibitemShut {NoStop}%
\bibitem[\natexlab{1}{39}]{CondChanSim}%
  \BibitemOpen
  \bibfield  {author} {\bibinfo {author} {\bibfnamefont {S.}~\bibnamefont
  {Pirandola}}, \bibinfo {author} {\bibfnamefont {R.}~\bibnamefont
  {Laurenza}},\ and\ \bibinfo {author} {\bibfnamefont {L.}~\bibnamefont
  {Banchi}},\ }\bibfield  {title} {\bibinfo {title} {Conditional channel
  simulation},\ }\href
  {https://www.sciencedirect.com/science/article/pii/S0003491618302999}
  {\bibfield  {journal} {\bibinfo  {journal} {Annals of Physics}\ }\textbf
  {\bibinfo {volume} {400}},\ \bibinfo {pages} {289} (\bibinfo {year}
  {2019})}\BibitemShut {NoStop}%
\bibitem [\natexlab{1}{32}]{FS}%
  \BibitemOpen
  \bibfield  {author} {\bibinfo {author} {\bibfnamefont {S.}~\bibnamefont
  {Pirandola}},\ }\bibfield  {title} {\bibinfo {title} {Limits and security of
  free-space quantum communications},\ }\href
  {https://link.aps.org/doi/10.1103/PhysRevResearch.3.013279} {\bibfield
  {journal} {\bibinfo  {journal} {Phys. Rev. Research}\ }\textbf {\bibinfo
  {volume} {3}},\ \bibinfo {pages} {013279} (\bibinfo {year}
  {2021}{\natexlab{a}})}\BibitemShut {NoStop}%
\bibitem [\natexlab{1}{33}]{SQC}%
  \BibitemOpen
  \bibfield  {author} {\bibinfo {author} {\bibfnamefont {S.}~\bibnamefont
  {Pirandola}},\ }\bibfield  {title} {\bibinfo {title} {Satellite quantum
  communications: Fundamental bounds and practical security},\ }\href
  {https://link.aps.org/doi/10.1103/PhysRevResearch.3.023130} {\bibfield
  {journal} {\bibinfo  {journal} {Phys. Rev. Research}\ }\textbf {\bibinfo
  {volume} {3}},\ \bibinfo {pages} {023130} (\bibinfo {year}
  {2021}{\natexlab{b}})}\BibitemShut {NoStop}%
  %
\bibitem [\natexlab{1}{14}]{OPGQN}%
  \BibitemOpen
    \bibfield  {author} {\bibinfo {author} {\bibfnamefont {C.}~\bibnamefont
  {Harney}}\ and\ \bibinfo {author} {\bibfnamefont {S.}~\bibnamefont
  {Pirandola}},\ }\bibfield  {title} {\bibinfo {title} {Analytical methods for
  high-rate global quantum networks},\ }\href
  {https://link.aps.org/doi/10.1103/PRXQuantum.3.010349}
  {\bibfield  {journal} {\bibinfo  {journal}
  {PRX Quantum}} \textbf {\bibinfo
  {volume} {3}},\ \bibinfo {pages} {10349} (\bibinfo {year}
  {2022})}\BibitemShut {NoStop}%
  \bibitem[\natexlab{1}{53}]{Gross2020}%
  \BibitemOpen
  \bibfield  {author} {\bibinfo {author} {\bibfnamefont {B.}~\bibnamefont
  {Gross}}, \bibinfo {author} {\bibfnamefont {D.}~\bibnamefont {Vaknin}},
  \bibinfo {author} {\bibfnamefont {S.~V.}\ \bibnamefont {Buldyrev}},\ and\
  \bibinfo {author} {\bibfnamefont {S.}~\bibnamefont {Havlin}},\ }\bibfield
  {title} {\bibinfo {title} {Two transitions in spatial modular networks},\
  }\href {https://doi.org/10.1088/1367-2630/ab8263} {\bibfield  {journal}
  {\bibinfo  {journal} {New Journal of Physics}\ }\textbf {\bibinfo {volume}
  {22}},\ \bibinfo {pages} {053002} (\bibinfo {year} {2020})}\BibitemShut
  {NoStop}%
\end{thebibliography}

%

\end{widetext}
\end{document}